\documentclass{article}
\usepackage[utf8]{inputenc}
\usepackage{amsmath,amsthm,amsfonts,amssymb,mathrsfs,environ}
\usepackage{booktabs}
\usepackage{geometry}
\usepackage{authblk}
\usepackage{bm}
\usepackage{bbm}
\usepackage{comment}
\usepackage{dsfont}
\usepackage{calrsfs}
\usepackage{rotating}
\usepackage{multicol}
\usepackage{multirow}
\usepackage{url}
\usepackage{enumerate}
\usepackage[round,longnamesfirst]{natbib}
\usepackage[tableposition=top,singlelinecheck=off,labelfont=bf]{caption}
\usepackage{tabularx}
\usepackage{tikz}
\usetikzlibrary{decorations.pathreplacing}
\usepackage{graphicx}
\usepackage{subfigure}
\usepackage{verbatim}
\usepackage{color}
\usepackage{pdflscape}
\usepackage{afterpage}
\usepackage{caption}
\usepackage{breqn}
\usepackage[title]{appendix}
\usepackage{algorithm,algpseudocode}
\usepackage{imakeidx}

\usepackage{multirow}
\usepackage{dcolumn}
\newcolumntype{d}[1]{D{.}{.}{#1}} 
\usepackage{tablefootnote}
\usepackage[flushleft]{threeparttable}

\usepackage{xr}
\DeclareMathOperator*{\argmin}{arg~min}

\DeclareMathOperator{\E}{\mathbb{E}}

\DeclareMathOperator{\Prob}{\mathbb{P}}

\DeclareMathOperator{\vect}{\mathrm{vec}}
\DeclareMathOperator{\diag}{\text{diag}}

\newcommand{\abs}[1]{\left\lvert#1\right\rvert}
\newcommand{\norm}[1]{\left\lVert#1\right\rVert}
\newcommand{\normoneinf}[1]{\left\lVert#1\right\rVert_{\vdash}}

\newcommand{\calF}{\mathcal{F}}

\newcommand{\bA}{\bm A}

\newcommand{\bB}{\bm B}

\newcommand{\bC}{\bm C}
\newcommand{\bc}{\bm c}
\newcommand{\bD}{\bm D}

\newcommand{\be}{\bm e}
\newcommand{\bF}{\bm F}

\newcommand{\bI}{\bm I}

\newcommand{\bM}{\bm M}

\newcommand{\bQ}{\bm Q}
\newcommand{\bq}{\bm q}

\newcommand{\bu}{\bm u}
\newcommand{\bV}{\bm V}

\newcommand{\bW}{\bm W}
\newcommand{\bw}{\bm w}

\newcommand{\bx}{\bm x}

\newcommand{\by}{\bm y}

\newcommand{\bbeta}{\bm \beta}

\newcommand{\bepsilon}{\bm \epsilon}

\newcommand{\biota}{\bm \iota}

\newcommand{\blambda}{\bm \lambda}

\newcommand{\bsigma}{\bm \sigma}

\newcommand{\bDelta}{\bm \varDelta}

\newcommand{\bSigma}{\bm \varSigma}

\newcommand{\bzeros}{\boldsymbol{0}}
\newcommand{\bZeros}{\mathbf{O}}
\newcommand{\tran}{'}
\newcommand{\ind}[1]{\mathbbm{1}_{\{#1\}}}
\newcommand{\hbanded}[2]{\mathcal{B}_{h_{#2}}\big(#1\big)}
\newcommand{\setG}{\mathcal{G}}

\newcommand{\sR}{\mathbb{R}}

\theoremstyle{definition}

\newtheorem{assumption}{Assumption}

\newtheorem{remark}{Remark}
\theoremstyle{plain}
\newtheorem{theorem}{Theorem}
\newtheorem{lemma}{Lemma}
\newtheorem{corollary}{Corollary}

\newcommand{\xoffset}{0.25}

\title{Sparse Generalized Yule–Walker Estimation for Large Spatio-temporal Autoregressions with an Application to NO\textsubscript{2} Satellite Data}
\author[1]{Hanno Reuvers\thanks{Corresponding author: Department of Econometrics, Erasmus School of Economics, Erasmus University Rotterdam, Burgemeester Oudlaan 50, 3062 PA Rotterdam, The Netherlands. E-mail address: reuvers@ese.eur.nl. Phone: +31 10 40 81257.}}
\author[2]{Etienne Wijler}
\affil[1]{Department of Econometrics, Erasmus University Rotterdam, 3062 PA Rotterdam, The Netherlands}
\affil[2]{Department of Quantitative Economics, Maastricht University, 6200 MD Maastricht, The Netherlands}
\date{\today}

\begin{document}

\maketitle
\begin{abstract}
We consider a high-dimensional model in which variables are observed over time and space. The model consists of a spatio-temporal regression containing a time lag and a spatial lag of the dependent variable. Unlike classical spatial autoregressive models, we do not rely on a predetermined spatial interaction matrix, but infer all spatial interactions from the data. Assuming sparsity, we estimate the spatial and temporal dependence fully data-driven by penalizing a set of Yule-Walker equations. This regularization can be left unstructured, but we also propose customized shrinkage procedures when observations originate from spatial grids (e.g. satellite images). Finite sample error bounds are derived and estimation consistency is established in an asymptotic framework wherein the sample size and the number of spatial units diverge jointly. Exogenous variables can be included as well. A simulation exercise shows strong finite sample performance compared to competing procedures. As an empirical application, we model satellite measured $\text{NO}_2$ concentrations in London. Our approach delivers forecast improvements over a competitive benchmark and we discover evidence for strong spatial interactions.\\


\textit{JEL-Codes}: C33, C53, C55
\end{abstract}

\section{Introduction}

Spatio-temporal models are powerful tools to explain and exploit dependencies between variables that are observed over both time and space, but they come with a number of challenges. In particular, endogeneity issues arise because the contemporaneous observations occur on both sides of the model equation. Furthermore, the inclusion of both spatial and temporal lags quickly results in heavily parameterized models. To circumvent these issues, a large part of the literature incorporates predetermined spatial weight matrices that govern the contemporaneous interactions between spatial units. Examples of this modelling strategies are: the spatial autoregressive model with a Gaussian quasi-maximum likelihood estimator (QMLE) by \citet{lee2004}; the QMLE estimation of stationary spatial panels with fixed effects detailed in \citet{yudejonglee2008}; the extension of these spatial panels to include spatially autoregressive disturbances as in \citet{leeyu2010}; a further extension to a non-stationary setting in which units can be spatially cointegrated in \citet{Yu2012}; and the computationally beneficial generalized method of moments (GMM) estimator by \citet{Lee2014}. While the choice of the spatial weight matrix is a key element of the model specification, its selection process can feel somewhat arbitrary and/or tedious. The arbitrariness might prevail when practical considerations fail to suggest a particular mechanism for the spatial interactions. Accordingly, more recent literature focuses on either incorporating multiple weight matrices \citep[e.g.][]{DebarsyLeSage2018,zhangyu2018} or, at the expense of estimating many parameters, directly inferring all spatial interactions from the data \citep[e.g.][]{lamsouza2019,Gao2019,MaGuoWang2021}. We contribute to the latter strand of literature with the development of a new estimator that provides several unique benefits.

In this paper, we propose the SPatial LAsso-type SHrinkage (SPLASH) estimator as a fully data-driven estimator of spatio-temporal interactions. Apart from a generous bandwidth upper bound, this SPLASH estimator leaves the spatial weight matrix and autoregressive matrix unspecified while employing a lasso approach to recover sparse solutions. Building upon previous works by \citet{DouParrellaYao2016} and \citet{Gao2019}, we resolve the endogeneity problem in our spatio-temporal regressions by estimating the generalized Yule-Walker equations. Our contributions are five-fold. First, assuming sparsity in the coefficient matrices and general mixing conditions on the innovations, we derive finite-sample performance bounds for the estimation and prediction error of our estimator. We subsequently utilize these bounds to derive asymptotic consistency in a variety of settings. For example, in the special case of a finitely bounded bandwidth and unstructured sparsity, it follows that the number of spatial units $N$ may grow at any polynomial rate of the number of temporal observations $T$. Second, we adopt a banded estimation procedure for the autocovariance matrices that underlie the generalized Yule-Walker equations. The faster convergence rates of these banded autocovariance matrix estimators are shown to translate into better convergence rates of our SPLASH estimator. Third, we show that dependence between neighbouring units that are ordered on a spatial grid translates to diagonally structured forms of sparsity in the spatial weight matrix. A tailored regularization procedure reminiscent of the sparse group lasso is proposed. Fourth, we generalize the model to include exogenous variable and demonstrate that an extended system of Yule-Walker equations continues to provide consistent estimators. Our simulation study confirms these points. Fifth and final, we employ SPLASH to predict $\text{NO}_2$ concentrations in London from satellite data.

Elaborating on the empirical application, we collect daily $\text{NO}_2$ column densities from August 2018 to October 2020, recorded by the TROPOspheric Monitoring Instrument (TROPOMI) on board of the Corpernicus Sentinel-5 Precursor satellite. Each spatial unit is an aggregation of a small number of pixels on the satellite image. We find that SPLASH constructs more accurate one-step ahead predictions for all spatial units compared to the procedure in \citet{Gao2019}, while outperforming a competitive penalized VAR benchmark for the majority of spatial units. In addition, we find evidence for spatial interactions between first-order neighbours and second-order neighbours (i.e. neighbours of neighbours).

There are two strands of literature that are closely linked to this work: the literature on the estimation of (nonparametric) spatial weight matrices and the literature on spatio-temporal vector autoregressions. Some similarities and differences are as follows. \citet{lamsouza2014} consider a model specification where the spatial units depend linearly on a spatial lag and exogenous regressors. The adaptive lasso is proven to select the correct sparsity pattern. To solve the endogeneity issue, they require the error variance to decay to zero as the time dimension grows large. \citet{Ahrens2015} solve the endogeneity problem using external instruments. Their two-step lasso estimation procedure selects the relevant instruments in the first step and the relevant spatial interactions in the second step. The theoretical properties of this estimator are derived using moderate deviation theory as in \citet{Jing2003}. This approach requires the instruments and the idiosyncratic component to be serially independent. Clearly, a serial independence assumption is unrealistic for the spatio-temporal models we consider here. Finally, \citet{lamsouza2019} augment a spatial lag model with a set of potentially endogenous variables (the augmenting set). They decompose the spatial weight matrix into a pre-determined component based on expert knowledge and a sparse adjustment matrix that represents specification errors. The adjustment matrix is sparsely estimated based on a penalized version of instrumental variables (IV) regression. If these instrumental variables are selected as temporal lags of the dependence variable, then their IV regressions are similar to generalized Yule-Walker estimation. In contrast to our approach, \citet{lamsouza2019} do not regularize the interactions between the dependent variables and the variables in the augmenting set, and they assume the number of such interactions to be fixed. A fixed number of interactions is inappropriate in high-dimensional settings in which the number of spatial units is allowed to diverge.

Closest related to the our work are \citet{Gao2019}, and \citet{MaGuoWang2021}. Both papers consider the same model and an estimation procedure that relies on generalized Yule-Walker equations. The key difference with our paper lies in the method by which the model complexity is controlled during estimation. \citet{Gao2019} assume the coefficient matrices to be banded with a bandwidth that is small compared to the number of spatial units. The bandwidth is determined from the data and all parameters within the selected bandwidth are left unregularized. Our SPLASH estimator, however, has the ability to exploit (structured) sparsity within the bandwidth and thereby improve estimation and forecasting performance. In addition, apart from a generous upper bound on the bandwidth to ensure identification, SPLASH does not require an a priori choice regarding the bandwidth. The recently developed bagging approach in \citet{MaGuoWang2021} does allow for sparsity within the bands, yet it also requires the calculation of so-called solution paths. That is, a forward addition and backward deletion stage are needed to determine the variables that enter the final model specification. In contrast, the SPLASH estimator provides this solution at once. Furthermore, their approach is not designed to detect diagonally structured forms of sparsity, while the ability to do so results in clear performance improvements of SPLASH in both the simulations and the empirical application considered below.

This paper is organized as follows. Section \ref{sec:spatiotempmodel} introduces the spatio-temporal vector autoregression and the banded autocovariance estimator that underlies the generalized Yule-Walker estimation approach. The SPLASH estimator and its theoretical properties are discussed in Section \ref{sec:lassoestimation}. The simulation results in Section \ref{sec:MC} and the empirical application in Section \ref{sec:empapplic} demonstrate the benefits of the SPLASH estimator. Section \ref{sec:conclusion} concludes.

\subsection*{Notation}
The indicator function $\ind{A}$ equals 1 if $A$ is true and zero otherwise. For a vector $\bx \in \sR^N$, the $L_p$-norm of $\bx$ is denoted $\norm{\bx}_p = \big(\sum_{i=1}^N |x_i|^p\big)^{1/p}$, with $\norm{\bx}_\infty = \max_i \abs{x_i}$ as an important special case. The total number elements in $\bx$ is denoted by $\abs{\bx}$ and the number of non-zero elements in $\bx$ is denoted by $\mathcal{M}(\bx) = \sum_{i=1}^N \mathbbm{1}\lbrace x_i \neq 0\rbrace$. The Orlicz norm is defined as $\norm{\cdot}_{\psi} = \inf\left\{ c>0 : \E \big[\psi\left(|\cdot|/c \right)\big]  \leq 1 \right\}$ for any $\psi(\cdot): \mathbb{R}^+ \to \mathbb{R}^+$ being a convex, increasing function with $\psi(0)=0$ and $\psi(x)\to \infty$ as $x\to \infty$. In addition, we rely on several matrix norms. For a matrix $\bA \in \sR^{M \times N}$, the matrix norms induced by the vector $L_p$-norms are given by $\norm{\bA}_p = \sup_{\bx \in \mathbb{R}^M} \big( \norm{\bA\bx}_p / \norm{\bx}_p\big)$. Noteworthy examples are: $\norm{\bA}_1 = \max_{1 \leq j \leq N} \sum_{i=1}^M \abs{a_{ij}}$, the spectral norm $\norm{\bA}_2 = \big[\lambda_{\max}(\bA\tran\bA)\big]^{1/2}$ where $\lambda_{\max}(\cdot)$ stands for the maximum eigenvalue, and $\norm{\bA}_\infty = \max_{1 \leq i \leq M} \sum_{j=1}^N \abs{a_{ij}}$. The Frobenius norm of $\bA$ is $\norm{\bA}_{F}= \big( \sum_{i=1}^M \sum_{j=1}^N |a_{ij}|^2 \big)^{1/2}$. Finally, we define $\norm{\bA}_\max = \max_{i,j} \abs{a_{ij}}$ and $\normoneinf{\bA}=\max\left\{\norm{\bA}_1,\norm{\bA}_\infty \right\}$. Let $S \subseteq \lbrace 1, \ldots, N \rbrace$ denote an index set with cardinality $\abs{S}$. Then, $\bx_S$ denotes the $\abs{S}$-dimensional vector with the elements of $\bx$ indexed by $S$, whereas $\bA_S$ denotes the $(M \times \abs{S})$-dimensional matrix containing the columns of $\bA$ indexed by $S$. In addition, we define $\mathcal{D}_A(k) = \left\lbrace a_{ij} \ \vert \ \abs{i-j} = k\right\rbrace$ as the collections of elements lying on (pairs of) the diagonals in the matrix $\bA$. Finally, $C$ is a generic constant that can change value from line-to-line.

\section{The Spatio-temporal Vector Autoregression}\label{sec:spatiotempmodel}
As in the recent paper by \cite{Gao2019}, we consider the spatio-temporal vector autoregression
\begin{equation}
    \by_t = \bA\by_t + \bB\by_{t-1} + \bepsilon_t, \qquad\qquad\qquad\qquad t=1,\ldots,T,
\label{eq:ModelSpec}
\end{equation}
where $\by_t = (y_{1t},\ldots,y_{Nt})^\prime$ stacks the observations at time $t$ over a collection of $N$ spatial units. The contemporaneous spatial dependence between these spatial units is governed by the matrix $\bA = (a_{ij})_{i,j=1}^N$ with $a_{ii} = 0$ for $i=1,\ldots,N$. The matrix $\bB = (b_{ij})_{i,j=1}^N$ incorporates dependence on past realizations. Finally, we have the innovation vector $\bepsilon_t$. We impose the following assumptions on the DGP in \eqref{eq:ModelSpec}.

\begin{assumption}[Stability]\label{assump:stability}\
\begin{enumerate}[(a)]
    \item $\normoneinf{\bA}=\max\left\{ \norm{\bA}_1,\norm{\bA}_\infty \right\} \leq \delta_A <1$. 
    \item $\normoneinf{\bB}\leq C_B$ and $\frac{C_B}{1-\delta_A}<1$.
\end{enumerate}
\end{assumption}

\begin{remark}
Assumption \ref{assump:stability} is defined in terms of $\normoneinf{\cdot}$. Since $\norm{\bA}_1 = \norm{\bA\tran}_\infty \leq \normoneinf{\bA}$ for any matrix $\bA$, the norm $\normoneinf{\cdot}$ is convenient when bounding products of matrices containing transposes.
\label{remark:oneinfnorm}
\end{remark}

\begin{assumption}[Innovations]\label{assump:momentcond}\
\begin{enumerate}[(a)]
    \item The sequence $\{\bepsilon_t\}$ is a covariance stationary, martingale difference process with respect to the filtration $\calF_{t-1}=\sigma\left(\bepsilon_{t-1},\bepsilon_{t-2},\ldots\right)$, and geometrically strong mixing ($\alpha$-mixing). That is, the mixing coefficients $\{\alpha_m\}$ satisfy $\alpha_m\leq c_2 e^{-\gamma_\alpha m }$ for all $m$ and some constants $c_2,\gamma_\alpha>0$. The largest and smallest eigenvalues of $\bSigma_\epsilon= \E(\bepsilon_1 \bepsilon_1^\prime)=(\sigma_{ij})_{i,j=1}^N$ are bounded away from $0$ and $\infty$. 
    \item Either one of the following assumptions holds:
    \begin{enumerate}
     \item[(b1)] For $\psi(x) = x^d$, we require $\sup_{i,t} \norm{\epsilon_{it}}_\psi = \left(\E |\epsilon_{it}|^d \right)^{1/d} \leq \mu_d < \infty$ for $d\geq 4$.
     \item[(b2)] For $\psi(x) = \exp(x)-1$, we have $\sup_{i,t} \norm{\epsilon_{it}}_\psi \leq \mu_\infty < \infty$.
    \end{enumerate}
 \end{enumerate}
\end{assumption}

Assumption \ref{assump:stability} ensures that $\by_t = \bA\by_t + \bB\by_{t-1} + \bepsilon_t$ has a stable reduced form VAR(1) specification. This follows from the following two observations. First, Assumption \ref{assump:momentcond}(a) bounds the maximum row and column sums of $\bA$ and thereby constraints the contemporaneous dependence between the time series. This assumption reminds of the spatial econometrics literature in which the spatial parameter $\lambda$ is bounded from above and the prespecified spatial weight matrix $\bW_N$ is standardized (see, e.g. \cite{lee2004} and \cite{leeyu2010}). Typically, the product $\lambda \bW_N$ -- the natural counterpart of the matrix $\bA$ -- is required to fulfil conditions similar to $\normoneinf{\bA}\leq \delta_A <1$.\footnote{For instance, it is not uncommon to row-normalize $\bW_N$ (each absolute row sum equal to 1) and restrict $\lambda<1$, see pages 1903-1904 of \cite{lee2004}. If $\bW_N$ is symmetric, then also $\normoneinf{\lambda\bW_N}<1$.} Invertibility of $\bI_N-\bA$ is guaranteed because $\norm{\bA}_2\leq \sqrt{\norm{\bA}_1\,\norm{\bA}_\infty}\leq \delta_A\leq 1$ and we have the reduced-form representation $\by_t = \bC \by_{t-1}+ \bD \bepsilon_t$ with $\bC=(\bI_N-\bA)^{-1} \bB$ and $\bD =(\bI_N-\bA)^{-1}$. From $\normoneinf \bD \leq \sum_{j=0}^\infty \normoneinf{\bA}^j=\frac{1}{1-\delta_A}$, we infer that the absolute row and column sum of $\bI_N-\bA$ are bounded. The latter is the logical counterpart of assumption B2 in \cite{DouParrellaYao2016}. Second, Assumption \ref{assump:momentcond}(b) controls serial dependence. Indeed, we conclude from $\norm{\bC}_2\leq \normoneinf{\bC} \leq \frac{C_B}{1-\delta_A}<1$ that both unit root and explosive behaviour of the reduced form specification are ruled out. The resulting stable VAR(1) representation is convenient to study the theoretical properties of our penalized estimator.

The assumptions on the innovation process $\{\bepsilon_t\}$, Assumption \ref{assump:momentcond}, are closely related to those in \cite{Masini2019}. Assumption \ref{assump:momentcond}(a) places restrictions on the time series properties of the error term through martingale difference (m.d.) and mixing assumptions. The m.d. assumption implies that $\E(\bepsilon_t \by_{t-j}\tran)=\bzeros$ while the mixing assumption controls the serial correlation in the data. Polynomial or exponential tail decay of the distribution of the innovations is imposed through either Assumption \ref{assump:momentcond}(b1) or Assumption \ref{assump:momentcond}(b2), respectively. The type of tail decay will directly influence the growth rates we can allow for $N$ and $T$. The discussions in \cite{Masini2019} demonstrate that Assumption \ref{assump:momentcond} allows for a wide range of innovation models.

Any further structure being absent, there are $(2N-1)N$ unknown parameters in $\bA$ and $\bB$ to estimate. Three complications are encountered when estimating these parameters. First, if $\bA\neq \bZeros$, then $\by_t$ occurs on both sides of the equation, and we face an endogeneity problem which renders OLS estimation inconsistent. Second, the number of unknown parameters grows quadratically in the cross-sectional dimension $N$. The model thus quickly becomes too large to estimate accurately without regularization. Finally, the multitude of parameters raises concerns about identifiability. These three complications are addressed by: (1) imposing structure on the matrices $\bA$ and $\bB$, and (2) estimating the unknown coefficients using the Yule-Walker equations \citep[e.g.][p. 420]{Brockwell1991}.

There are several possibilities to introduce structure into $\bA$ and $\bB$. Early spatial econometrics models, e.g. the spatial autoregressive (SAR) model or spatial Durbin model (SDM), incorporate spatial effects through the product $\lambda \bW_N$ (with $\bW_N$ pre-specified). The specification $\bA = \lambda \bW_N$ imposes substantial structure on $\bA$ and leaves only the single parameter $\lambda$ to estimate. \cite{DouParrellaYao2016} consider a more general setting in which each row of $\bW_N$ receives its own spatial autoregressive parameter. Specifically, they set $\bA=\diag(\blambda_0)\bW_N$ and $\bB=\diag(\blambda_1)+\diag(\blambda_2)\bW_N$, and estimate the $3N$ coefficients in $(\blambda_0\tran,\blambda_1\tran,\blambda_2\tran)\tran$. \cite{Gao2019} require $\bA$ and $\bB$ to be banded matrices.\footnote{The matrix $\bA$ has bandwidth $k$ if the total number of nonzero entries in any row or column is at most $k$.} We employ a similar assumption. 

\begin{assumption}[Banded matrices]\
Recall $\bA = (a_{ij})_{i,j=1}^N$, $\bB = (b_{ij})_{i,j=1}^N$, and $\bSigma_\epsilon = (\sigma_{ij})_{i,j=1}^N$. We have: (a) $a_{ij} = b_{ij} = 0$ for all $\abs{i-j} > k_0$ with $k_0 < \lfloor N/4 \rfloor$, and (b) $\sigma_{ij} = 0$ for all $\abs{i-j} > l_0$.
\label{assump:bandedness}
\end{assumption}

Assumption \ref{assump:bandedness} serves two purposes. First, for each spatial unit $i=1,\ldots,N$, the matrices $\bA$ and $\bB$ are banded to have no more than $N$ unknown parameters per equation. With $N$ moment conditions for each $i$, Assumption \ref{assump:bandedness}(a) is key in identifying the parameters. Our discussions in Section \ref{sec:lassoestimation} illustrate that this assumption is realistic when the data is observed on a regular grid. The combination of Assumptions \ref{assump:bandedness}(a)--(b) is exploited in the Yule-Walker estimation approach. This approach requires estimation of the $(N\times N)$ autocovariance matrices $\bSigma_j=\E(\by_t \by_{t-j}\tran)$. Especially in our large $N$ settings, it is crucial to rely on covariance matrix estimators that converge at a fast rate. If $\bA$, $\bB$, and $\bSigma_\epsilon$ are banded, then the following result applies.

\begin{theorem}[Convergence rates for banded sample autocovariance matrices]\label{th:diag_approx}
For any matrix $\bM=(m_{ij})$, its $h$-banded counterpart is defined as $\hbanded{\bM}{}=(m_{ij}\ind{|i-j|\leq h})$.
Define the $(N\times 2N)$ matrix $\hat{\bV}_h=\big[\hbanded{\hat\bSigma_1}{}\tran \; \hbanded{\hat\bSigma_0}{}\big]$ with $\hat \bSigma_1 = \frac{1}{T} \sum_{t=2}^T \by_t \by_{t-1}\tran$ and $\hat \bSigma_0 = \frac{1}{T} \sum_{t=2}^T \by_t \by_t\tran$, and choose
\begin{equation}
\begin{aligned}
 h=h(\epsilon) = \left(\max\Big\{s^*,\frac{\log\left(C_4 /(1-\delta_C)\epsilon \right)}{|\log(\delta_A)|} \Big\}+1\right)&\left( 2 \frac{\log(C_4 /\epsilon)}{|\log(\delta_C)|}+3 \right)(k_0-1) \\
  &+ 2l_0+1,
\end{aligned}
\end{equation}
then $\normoneinf{\widehat{\bV}_h - \bV }\leq 6 \epsilon$ with a probability of at least
\begin{enumerate}[(a)]
    \item $1-2 \mathcal{P}_1(\epsilon,N,T)$ under Assumptions \ref{assump:stability}--\ref{assump:bandedness} using Assumption \ref{assump:momentcond}(b1) (polynomial tails),
    \item $1- 2\mathcal{P}_2(\epsilon,N,T)$ under Assumptions \ref{assump:stability}--\ref{assump:bandedness} using Assumption \ref{assump:momentcond}(b2) (exponential tails),
\end{enumerate}
where
$$
 \mathcal{P}_1(\epsilon,N,T) = N^2\left[ \left(b_1 T^{(1-\delta)/3}+\frac{[2h(\epsilon)+1]b_3}{\epsilon}\right) \exp\left(-\frac{T^{(1-\delta)/3}}{2b_1^2}\right) + \frac{b_2 [2h(\epsilon)+1]^d}{\epsilon^d T^{\frac{\delta}{2}(d-1)}}\right],
$$
for some $0<\delta<1$, and
$$
 \mathcal{P}_2(\epsilon,N,T) = N^2 \left[\frac{\kappa_1[2h(\epsilon)+1]}{\epsilon} + \frac{2}{\kappa_2}\left(\frac{T \epsilon^2}{[2h(\epsilon)+1]^2}\right)^{\frac{1}{7}} \right] \exp\left(- \frac{1}{\kappa_3}\left(\frac{T \epsilon^2}{[2h(\epsilon)+1]^2}\right)^{\frac{1}{7}} \right).
$$
All constants ($C_4$, $\delta_C$, $s^*$, etc.) are positive and independent of $N$ and $T$. The proof (see the Appendix) shows how these constants are related to quantities in Assumptions \ref{assump:stability}--\ref{assump:bandedness}.
\end{theorem}

Theorem \ref{th:diag_approx} shows that banded estimators for $\bSigma_0$ and $\bSigma_1$ provide an accurate approximation to $\bV = \big[ \bSigma_1\tran \; \bSigma_0 \big]\tran$. Each of these banded matrices has at most $2h(\epsilon)+1$ nonzero elements in their columns/rows. In other words, given $\epsilon$, $l_0$ and $k_0$, Assumptions \ref{assump:stability}--\ref{assump:bandedness} guarantee that $\bSigma_0$ and $\bSigma_1$ can be well-approximated by matrices with bandwidths smaller than $N$. This improves the convergence rate of our estimator.

\section{Sparse Estimation} \label{sec:lassoestimation}

\subsection{The SPLASH($\alpha$,$\lambda$) Estimator}\label{subsec:SPLASH}

Even under Assumption \ref{assump:bandedness}, the number of unknown parameters in $\bA$ and $\bB$ continues to grow quadratically in $N$. For large $N$, the accurate estimation of all these parameters becomes infeasible rather quickly. To alleviate this curse of dimensionality, we rely on sparsity. Sparsity naturally occurs when two spatial units do not interact with each other. We demonstrate, however, that a special, and exploitable, sparsity pattern arises whenever the spatial units are ordered in a structured way.

\begin{figure}
\begin{center}
\subfigure[]{
\includegraphics[height=0.23\textheight]{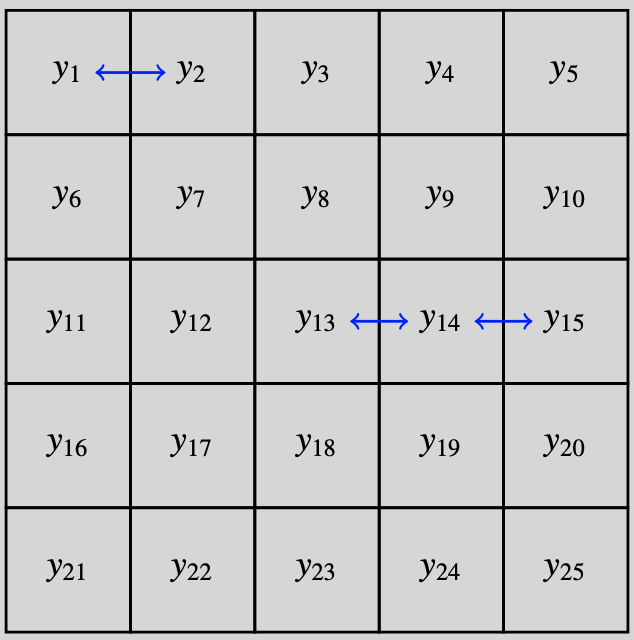}
}
\qquad
\subfigure[]{
\includegraphics[height=0.23\textheight]{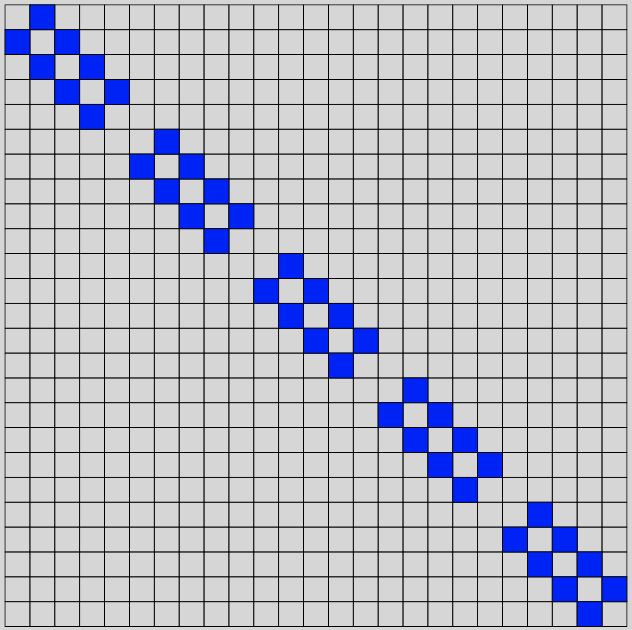}
}\\
\subfigure[]{
\includegraphics[height=0.23\textheight]{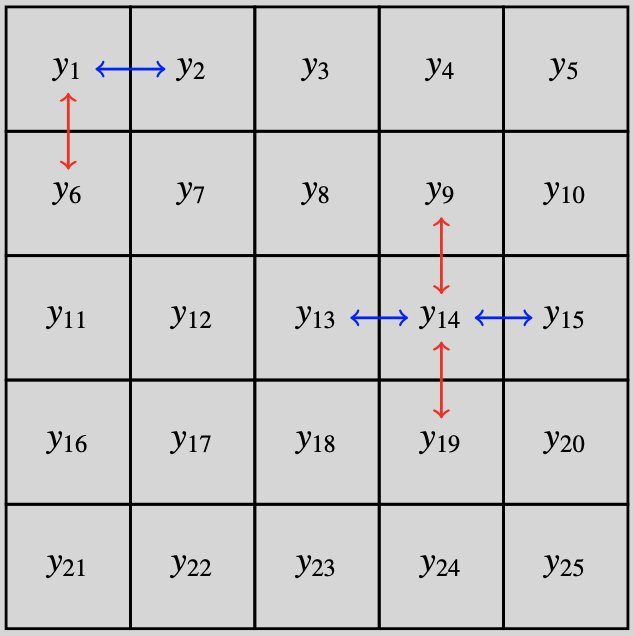}
}
\qquad
\subfigure[]{
\includegraphics[height=0.23\textheight]{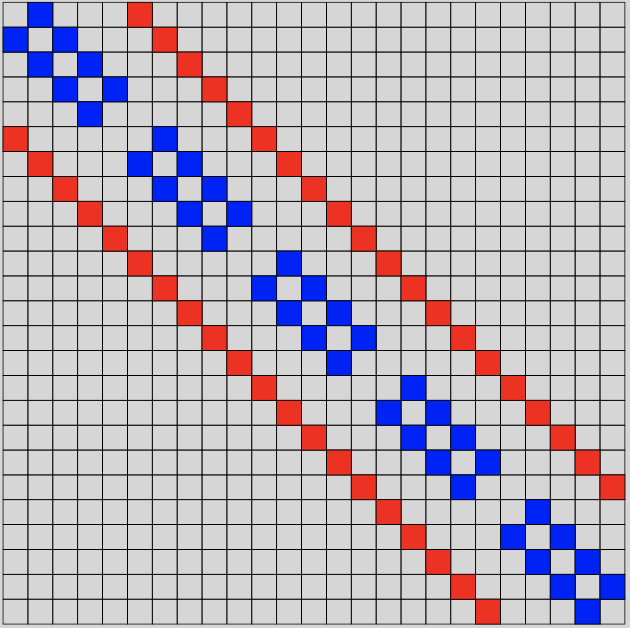}
}\\
\subfigure[]{
\includegraphics[height=0.23\textheight]{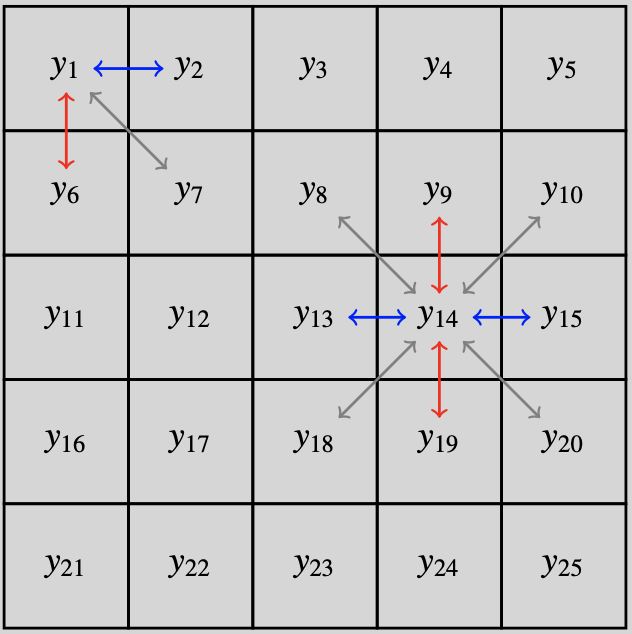}
}
\qquad
\subfigure[]{
\includegraphics[height=0.23\textheight]{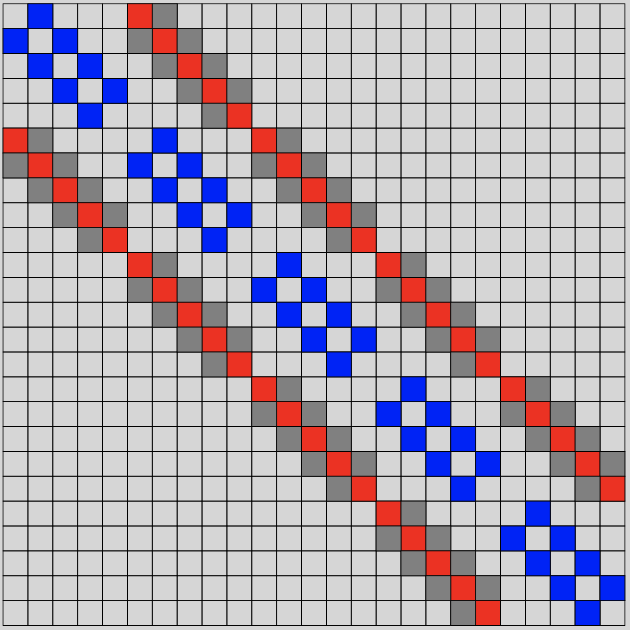}
}
\end{center}
\caption{The left figures show $(5\times 5)$ grids of spatial units with arrows depicting horizontal (blue), vertical (red), and diagonal interactions (grey). The right figures illustrate the sparsity pattern associated with these spatial layouts. For example, spatial unit $y_1$ interacts horizontally with $y_2$, vertically with $y_6$, and diagonally with $y_7$. If these are the only possible interactions for spatial unit $y_1$, then only $a_{12}$, $a_{16}$, and $a_{17}$ in the first row of $\bA$ are possibly nonzero (the interactions for $y_{14}$ are indicated as well). The right figures show the implied corresponding sparsity pattern for all spatial units.}
\label{fig:spatial_overview}
\end{figure}

As an illustrative example, let us consider repeated measurements on the $(5\times5)$ spatial grids shown in the left column of Figure \ref{fig:spatial_overview}. The $N=25$ spatial units are labelled $y_1$ up to $y_{25}$ and enumerated row-wise. This ordering of the spatial entities creates an implicit notion of proximity and we intuitively expect economic/physical interactions to be most pronounced at short length scales. In Figure \ref{fig:spatial_overview}(a) we start from the situation in which the spatial units are restricted to communicate horizontally. Blue arrows indicate explicitly that $y_1$ interacts with $y_2$, and $y_{14}$ interacts with both $y_{13}$ and $y_{15}$. Such interactions occur among all elements in the grid. More importantly, if only these horizontal interaction exist, then the $(25\times 25)$ matrices $\bA$ and $\bB$ feature a sparsity pattern as shown in Figure \ref{fig:spatial_overview}(b). The blue elements are potentially nonzero whereas uncolored elements are zero. The nonzero elements in $\bA$ and $\bB$ are seen to cluster in specific, dense diagonals with the occasional zero when horizontal neighbours are absent (at the boundary of the grid). This diagonal sparsity pattern is not an artifact of allowing horizontal interactions only. Figures \ref{fig:spatial_overview}(c) adds the vertical interactions and the accompanying sparsity pattern again manifests itself along diagonals (Figure \ref{fig:spatial_overview}(d)). Finally, with diagonal nearest neighbours being horizontal neighbors of vertical elements, we observe a ``thickening'' of the diagonals in Figure \ref{fig:spatial_overview}(f). Guided by these considerations we combine generalized Yule-Walker estimation with a sparse group penalty \citep[e.g.][]{Simon2013}. The Yule-Walker estimator will control for endogeneity, while the sparse group penalty will shrink towards diagonal structures by including/omitting complete diagonals and thus selecting the required interactions. Compared to \citet{Gao2019}, we hereby gain the ability to exploit sparsity within banded matrices.

A formal definition of our estimator requires further notation. Part of this notation comes naturally if we briefly review the generalized Yule-Walker estimator. After post-multiplying by $\by_{t-1}\tran$ and taking expectations, we find $\bSigma_1=\bA \bSigma_1 + \bB \bSigma_0$ or, equivalently, 
\begin{equation}
 \bSigma_1\tran =
 \begin{bmatrix}
 \bSigma_1\tran & \bSigma_0
 \end{bmatrix}
 \begin{bmatrix}
  \bA & \bB \end{bmatrix}\tran
 =: \bV \bC\tran.
\label{eq:YWpopulation}
\end{equation}
The $i$\textsuperscript{th} column of $\bC\tran$ contains all coefficients that belong to the $i$\textsuperscript{th} equation in \eqref{eq:ModelSpec}. Assumption \ref{assump:bandedness} requires several of these coefficients to be zero so we exclude these from the outset. We collect all remaining (possibly) nonzero coefficients in the $i$\textsuperscript{th} equation in the vector $\bc_i$, and define $\bV_i$ as the matrix containing the corresponding columns from $\bV$. In the population, we have $\bV_i \bc_i = \bSigma_1\tran \be_i =: \bsigma_i$ for $i=1,\ldots,N$. Sample counterparts of $\bV_i$ and $\bsigma_i$ are readily available from the sample autocovariance matrices. More explicitly, \cite{Gao2019} set $\hat\bsigma_i = \frac{1}{T}\sum_{t=2}^T \by_{t-1}y_{it}$ and construct $\hat{\bV}_i$ from the appropriate columns of $\hat \bV =\big[ \hat\bSigma_1\tran\;\hat\bSigma_0 \big]$. Motivated by $\bsigma_i - \bV_i \bc_i=\bzeros$, they define their estimator $\hat{\bc}^{GMWY}_i$ as the following minimizer:
\begin{equation}
 \hat\bc^{GMWY}_i = \argmin_{\bc} \norm{ \hat\bsigma_i - \hat\bV_i\bc}_2^2.
\label{eq:singleL2contri}
\end{equation}

\begin{figure}
\centering
\begin{tikzpicture}
    [
        box/.style={rectangle,draw=black,thick, minimum size=0.5cm},
    ]
\foreach \x in {1,...,5}{\node[box] at (\xoffset+0.5*\x,3-0.5*\x){};}
\foreach \x in {1,...,4}
{
 \ifthenelse{\x=2}{\node[box,fill=lightgray] at (\xoffset+0.5+0.5*\x,3-0.5*\x){$3$};}{\node[box,fill=lightgray] at (\xoffset+0.5+0.5*\x,3-0.5*\x){};}
}
\foreach \x in {1,...,3}{\node[box] at (\xoffset+1+0.5*\x,3-0.5*\x){};}
\foreach \x in {1,...,2}{\node[box] at (\xoffset+1.5+0.5*\x,3-0.5*\x){};}
\node[box] at (\xoffset+2.5,2.5){};
\foreach \x in {1,...,4}
{
 \ifthenelse{\x=1}{\node[box,fill=lightgray] at (\xoffset+0.5*\x,2.5-0.5*\x){$1$};}{\node[box,fill=lightgray] at (\xoffset+0.5*\x,2.5-0.5*\x){};}
}
\foreach \x in {1,...,3}{\node[box] at (\xoffset+0.5*\x,2-0.5*\x){};}
\foreach \x in {1,...,2}{\node[box] at (\xoffset+0.5*\x,1.5-0.5*\x){};}
\node[box] at (\xoffset+0.5,0.5){};

\draw [decorate,decoration={brace,mirror,amplitude=10pt}]
(0.5,0) -- (3,0)node [black,midway,yshift=-0.62cm] {\footnotesize $\bA$};

\foreach \x in {1,...,5}{
 \ifthenelse{\x=2}{\node[box,fill=lightgray] at (\xoffset+2.75+0.5*\x,3-0.5*\x){7};}{\node[box,fill=lightgray] at (\xoffset+2.75+0.5*\x,3-0.5*\x){};}
}
\foreach \x in {1,...,4}
{
 \ifthenelse{\x=2}{\node[box,fill=lightgray] at (\xoffset+2.75+0.5+0.5*\x,3-0.5*\x){8};}{\node[box,fill=lightgray] at (\xoffset+2.75+0.5+0.5*\x,3-0.5*\x){};}
}
\foreach \x in {1,...,3}{\node[box] at (\xoffset+2.75+1+0.5*\x,3-0.5*\x){};}
\foreach \x in {1,...,2}{\node[box] at (\xoffset+2.75+1.5+0.5*\x,3-0.5*\x){};}
\node[box] at (\xoffset+2.75+2.5,2.5){};
\foreach \x in {1,...,4}
{
 \ifthenelse{\x=1}{\node[box,fill=lightgray] at (\xoffset+2.75+0.5*\x,2.5-0.5*\x){6};}{\node[box,fill=lightgray] at (\xoffset+2.75+0.5*\x,2.5-0.5*\x){};}
}
\foreach \x in {1,...,3}{\node[box] at (\xoffset+2.75+0.5*\x,2-0.5*\x){};}
\foreach \x in {1,...,2}{\node[box] at (\xoffset+2.75+0.5*\x,1.5-0.5*\x){};}
\draw [decorate,decoration={brace,mirror,amplitude=10pt}]
(3.25,0) -- (5.75,0)node [black,midway,yshift=-0.62cm] {\footnotesize $\bB$};
\node[box] at (\xoffset+2.75+0.5,0.5){};

\draw[draw=red,line width=1pt,minimum size = 0.5cm] (\xoffset,1.75cm) rectangle ++(5.75,0.5);
\node at (\xoffset+2.875,-1.25) {(a)};

\foreach \x in {1,...,5}{\node[box,fill=lightgray] at (\xoffset+6+0.5,3-0.5*\x){};}
\foreach \x in {1,...,5}{\node[box] at (\xoffset+6+1,3-0.5*\x){};}
\foreach \x in {1,...,5}{\node[box,fill=lightgray] at (\xoffset+6+1.5,3-0.5*\x){};}
\foreach \x in {1,...,5}{\node[box] at (\xoffset+6+2,3-0.5*\x){};}
\foreach \x in {1,...,5}{\node[box] at (\xoffset+6+2.5,3-0.5*\x){};}
\foreach \x in {1,...,5}{\node[box,fill=lightgray] at (\xoffset+6+3,3-0.5*\x){};}
\foreach \x in {1,...,5}{\node[box,fill=lightgray] at (\xoffset+6+3.5,3-0.5*\x){};}
\foreach \x in {1,...,5}{\node[box,fill=lightgray] at (\xoffset+6+4,3-0.5*\x){};}
\foreach \x in {1,...,5}{\node[box] at (\xoffset+6+4.5,3-0.5*\x){};}
\foreach \x in {1,...,5}{\node[box] at (\xoffset+6+5,3-0.5*\x){};}
\node at (\xoffset+6+0.5,3) {1};
\node at (\xoffset+6+1.5,3) {3};
\node at (\xoffset+6+3.0,3) {6};
\node at (\xoffset+6+3.5,3) {7};
\node at (\xoffset+6+4.0,3) {8};
\draw [decorate,decoration={brace,mirror,amplitude=10pt}]
(6.5,0) -- (11.5,0)node [black,midway,yshift=-0.62cm] {\footnotesize $\hat \bV_h = \left[\hbanded{\hat\bSigma_1}{}\tran \; \hbanded{\hat\bSigma_0}{}\right]$};
\node at (\xoffset+8.75,-1.25) {(b)};

\foreach \x in {1,...,5}{\node[box,fill=lightgray] at (\xoffset+12,3-0.5*\x){};}
\foreach \x in {1,...,5}{\node[box,fill=lightgray] at (\xoffset+12.5,3-0.5*\x){};}
\foreach \x in {1,...,5}{\node[box,fill=lightgray] at (\xoffset+13,3-0.5*\x){};}
\foreach \x in {1,...,5}{\node[box,fill=lightgray] at (\xoffset+13.5,3-0.5*\x){};}
\foreach \x in {1,...,5}{\node[box,fill=lightgray] at (\xoffset+14,3-0.5*\x){};}
\node at (\xoffset+12,3) {1};
\node at (\xoffset+12.5,3) {3};
\node at (\xoffset+13,3) {6};
\node at (\xoffset+13.5,3) {7};
\node at (\xoffset+14,3) {8};
\draw [decorate,decoration={brace,mirror,amplitude=10pt}]
(12,0) -- (14.5,0)node [black,midway,yshift=-0.62cm] {\footnotesize $\hat{\bV}_{2,h}$};
\node at (\xoffset+13,-1.25) {(c)};
\end{tikzpicture}
\caption{A visualization on the construction of $\hat \bV_{2,h}$ for $N=5$. \textbf{(a)} If $h=1$, then grey elements in $\bA$ and $\bB$ are (potentially) nonzero whereas white elements are zero by construction. Enumerating along the second row, the active elements are in the set $\{1,3,6,7,8\}$. \textbf{(b)} We select the columns from $\hat \bV_h$ corresponding to the active set. \textbf{(c)} The matrix $\hat{\bV}_{2,h}$ is the submatrix of $\hat \bV_h$ with only active columns.
}
\label{fig:singleequationoverview}
\end{figure}
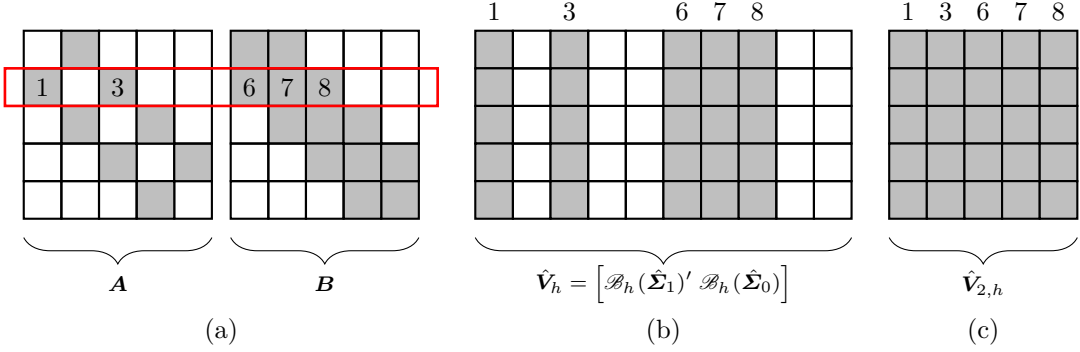

We will adjust this objective function in three ways. First, we define our estimator in terms of banded estimated covariance matrices, which allows us to exploit the results in Theorem \ref{th:diag_approx}. Second, our group penalty penalizes parameters across equations so we can no longer estimate the parameters equation-by-equations. We therefore define $\hat \bsigma_h = \vect\big(\hbanded{\hat\bSigma_1}{}\tran \big)$ and $\hat\bV_h^{(d)}=\diag\big(\hat\bV_{i,h},\ldots,\hat\bV_{N,h} \big)$, with $\hat\bV_{i,h}$ being constructed similarly to $\bV_i$ (see Figure \ref{fig:singleequationoverview} for an illustration).\footnote{In general, all quantities derived from autocovariance matrices come in three versions: (1) the population quantity, (2) the estimated counterpart without banding denoted with an additional ``hat'', and (3) the estimated counterpart with banding featuring both a ``hat'' and the subscript $h$.} In this notation, the expression $\big\|\hat{\bsigma}_h - \hat{\bV}_h^{(d)}\bc\big\|_2^2$ defines the joint objective function that sums the individual contributions in  \eqref{eq:singleL2contri} over all equations. Finally, we construct the penalty function. We define an index set that partitions the vector $\bc$ into sub-vectors, denoted $\lbrace\bc_g\rbrace$, that contain the non-zero diagonals of $\bA$ and $\bB$ that are admissible under Assumption \ref{assump:bandedness} as
\begin{equation}
\begin{split}
    \setG_A &:= \left\lbrace g \subset \mathbb{N} : \bc_g = \mathcal{D}_A(k), k \in \lbrace 1, \ldots, \lfloor N/4 \rfloor \right\rbrace,\\
    \setG_B &:= \left\lbrace g \subset \mathbb{N} : \bc_g = \mathcal{D}_B(k), k \in \lbrace 0, \ldots, \lfloor N/4 \rfloor \right\rbrace,
\end{split}
\end{equation}
respectively, and let $\setG = \setG_A \cup \setG_B$. Based on this notation, we define our objective function as
\begin{equation}
\mathcal{L}_\alpha(\bc;\lambda)
=  \norm{\hat{\bsigma}_h - \hat{\bV}_h^{(d)}\bc}_2^2 + \lambda\Bigg(\underbrace{(1-\alpha)\sum_{g \in \setG}\sqrt{\abs{g}}\norm{\bc_g}_2 + \alpha\norm{\bc}_1}_{=: P_\alpha(\bc)}\Bigg)
= \norm{\hat{\bsigma}_h - \hat{\bV}_h^{(d)}\bc}_2^2 + \lambda P_\alpha(\bc).
\label{eq:SGLobjfunc}
\end{equation}
The spatial lasso-type shrinkage estimator, abbreviated SPLASH($\alpha$,$\lambda$) or SPLASH in short, is defined as the minimizer of \eqref{eq:SGLobjfunc}, i.e. $\hat{\bc}=\argmin_\bc \mathcal{L}_\alpha(\bc;\lambda)$. The importance of the penalty function $P_\alpha(\bc)$ is governed by the penalty parameter $\lambda$ and the second hyperparameter $\alpha$ balances group-structured sparsity versus individual sparsity. At the extremities of  $\alpha\in[0,1]$ we find the group lasso ($\alpha=0$) and the lasso ($\alpha=1$). Intermediate values of $\alpha$ will shrink both groups of diagonal coefficients in $\bA$ and $\bB$ and individual parameters. The SPLASH solution promotes completely sparse diagonals and sparse elements within nonzero diagonals, and thus shrinks towards sparsity patterns of the type displayed in Figure \ref{fig:spatial_overview}(b). As the structure of our estimator is similar to that of the Sparse Group Lasso (SGL), efficient algorithms are available to compute its solution \citep[see, e.g.][]{Simon2013}. An R/C++ implementation of the SPLASH estimator based on this algorithm is available on one of the author's website.\footnote{\url{https://sites.google.com/view/etiennewijler/code}}.

\subsection{Theoretical Properties of the SPLASH($\alpha$,$\lambda$) Estimator}

In this section we derive the theoretical properties of the SPLASH estimator. First, however, we require an additional assumption on the DGP in order to ensure that $\bA$ and $\bB$ in \eqref{eq:ModelSpec} are uniquely identified. To this end, we leverage the bandedness assumption in Assumption \ref{assump:bandedness}, which enables unique identification of $\bA$ and $\bB$ via a straightforward full-rank condition on sub-matrices of the autocovariance matrices that appear in the generalized Yule-Walker equations.

\begin{assumption}[Restricted minimum eigenvalue]\label{assump:min_eigen_bound}
Assume that
$$
\phi_\min(\bx) := \underset{\bx \in \mathbb{R}^{2N}: \mathcal{M}(\bx) \leq N}{\text{min}} \frac{\norm{\bV\bx}_2}{\norm{\bx}_2} \geq \phi_0.
$$
\end{assumption}

Assumption \ref{assump:min_eigen_bound} states that every sub-matrix containing $N$ columns from $\bV$ has full column-rank and a minimum singular value bounded away from zero. Related assumptions appear in \citet[][Section 4]{Bickel2009}, who refer to $\phi_\min(\bx)$ as a \textit{restricted eigenvalue} and use this quantity to construct sufficient conditions for their restricted eigenvalue assumptions. Assumption \ref{assump:min_eigen_bound} fits our framework particularly well, as the assumed maximum bandwidth of the matrices $\bA$ and $\bB$ in Assumption \ref{assump:bandedness} imply that the diagonal blocks of the matrix $\bV^{(d)}$ never contain more than $N$ unique columns of $\bV$. Using this property, we show in Lemma \ref{Lemma:min_vs_restricted_eigen} of Appendix \ref{app:selected_lemmas} that a Sparse Group Lasso compatibility condition is implied by Assumption \ref{assump:min_eigen_bound}. 

Equipped with Assumption \ref{assump:min_eigen_bound}, we find the following finite-sample performance bounds on the prediction and estimation error of SPLASH.
\begin{theorem}\label{Thm:sgl}
Define 
\begin{equation*}
    \bar\omega_\alpha = \max \Big\{(1-\alpha)\sum_{g \in \mathcal{G}_S}\sqrt{\abs{g}},\alpha\sqrt{\abs{S}}\Big\},
\end{equation*}
where $\mathcal{G}_s = \left\lbrace g \in \mathcal{G} : \bc_g \neq \bm{0}\right\rbrace$ and $S = \left\lbrace j : c_j \neq 0\right\rbrace$. Under Assumptions \ref{assump:stability}--\ref{assump:min_eigen_bound} and $\normoneinf{\bV}\leq C_V$, it holds that
\begin{equation}
    \norm{\hat{\bV}_h^{(d)}(\hat{\bc}-\bc)}_2^2 + \lambda\left((1-\alpha)\sum_{g \in \mathcal{G}} \sqrt{\abs{g}}\norm{\hat{\bc}_g - \bc_g}_2 + \alpha\norm{\hat{\bc}-\bc}_1\right) \leq \frac{64\bar{\omega}_\alpha^2\lambda^2}{\phi_0^2}
\label{eq:InequalityTheorem2}
\end{equation}
with a probability of at least
\begin{enumerate}[(a)]
    \item $1- 10\mathcal{P}_1\left(f(\lambda,\phi_0),N,T\right)$ when Assumption \ref{assump:momentcond}(b1) (polynomial tail decay) is valid, or
    \item $1- 10\mathcal{P}_2\left(f(\lambda,\phi_0),N,T\right)$ when Assumption \ref{assump:momentcond}(b2) (exponential tail decay) is valid,
\end{enumerate}
where $\mathcal{P}_1\left(x,N,T\right)$ and $\mathcal{P}_2\left(x,N,T\right)$ are defined in Theorem \ref{th:diag_approx} and $f(\lambda,\phi_0) = \min\left(\frac{\lambda^{1/2}}{24},\frac{\lambda}{96C_v},\frac{\phi_0}{12}\right)$.
\end{theorem}

Theorem \ref{Thm:sgl} contains a finite-sample performance bound on the prediction and estimation error for the SPLASH($\alpha$,$\lambda$) estimator. It offers some interesting insights. First, we focus on the probability with which inequality \eqref{eq:InequalityTheorem2} holds. For VAR estimation with a penalized least-squares objective function, such probabilities are governed by tail probabilities of the process $\{\frac{1}{T} \sum_{t=1}^T y_{it} \epsilon_{jt}\}$ (see, e.g. lemma 4 in \cite{KockCallot2015}, or lemmas 5--6 in \cite{Medeiros2016}). Because Yule-Walker estimation relies primarily on autocovariance matrix estimation, our probability depends on the tail decay of the distribution of $\{\big\|\widehat{\bV}_h - \bV\big\|_\vdash \}$. Overall, the probability of \eqref{eq:InequalityTheorem2} improves through faster tail decay of the innovation distribution (compare cases \emph{(a)} and \emph{(b)}) and banded autocovariance matrix estimation (Theorem \ref{th:diag_approx}). Second, we look closer at the performance upper bound itself. The right-hand side of \eqref{eq:InequalityTheorem2} demonstrates that the upper bound of the prediction and estimation error is increasing in $\bar{\omega}_\alpha$, which in turn is increasing in the bandwidths $k_0$ and $l_0$, increasing in the group sizes ($\alpha<1$), and increasing in the number of relevant interactions $\abs{S}$ ($\alpha>0$). Furthermore, the prediction and estimation error increases in the degree of penalization. Whereas this seemingly suggests to minimize $\lambda$ as to improve performance bounds, we emphasize that the effect of regularization in Theorem \ref{Thm:sgl} is two-fold: increasing regularization deteriorates the performance bound, but increases the probability of the set on which the performance bound holds. Intuitively, shrinkage induces finite-sample bias which worsens accuracy, but simultaneously reduces sensitivity to noise, thereby enabling performance guarantees at higher degrees of certainty.

The aforementioned effects can also be demonstrated by means of an asymptotic analysis. Based on Theorem \ref{Thm:sgl}, we derive the conditions for convergence of the prediction and estimation errors in the following corollary. The exact convergence rates are also provided.

\begin{corollary}\label{cor:splash_rates}
Let $\lambda \in O\left(T^{-q_\lambda}\right)$, $N \in O\left(T^{q_N}\right)$, $\abs{\mathcal{G}_S} \in O\left(T^{q_g}\right)$, $\abs{S} \in O\left(T^{q_s}\right)$, $k_0, l_0 \in O\left(T^{q_k}\right)$, where $q_\lambda$, $q_N$, $q_s$, and $q_k$ are fixed and positive constants. Maintain Assumptions \ref{assump:stability}-\ref{assump:min_eigen_bound} and assume that either (i) $q_\lambda < -\frac{2q_N}{d} - q_k + \frac{\delta(d-1)}{2d}$ for some $0 < \delta < 1$ and Assumption \ref{assump:momentcond}(b1) holds, or (ii) $q_\lambda < \frac{1}{2} - q_k$ and Assumption \ref{assump:momentcond}(b2) holds. Then,
\begin{enumerate}[(a)]
    \item $\norm{\hat{\bV}_h^{(d)}\left(\hat{\bc}-\bc\right)}_2^2 = O_p\left((1-\alpha)T^{2q_g +q_N - 2q_\lambda} + \alpha T^{q_s - 2q_\lambda}\right)$,
    \item $(1-\alpha)\sum_{g \in \mathcal{G}}\sqrt{\abs{g}}\norm{\hat{\bc}_g - \bc_g}_2 + \alpha\norm{\hat{\bc} - \bc}_1 =  O_p\left((1-\alpha)T^{2q_g +q_N - q_\lambda} + \alpha T^{q_s - q_\lambda}\right)$.
\end{enumerate}
\end{corollary}

\bigskip
Corollary \ref{cor:splash_rates} provides insights into the determinants of the convergence rate. In particular, the result confirms that the convergence rate decreases in the bandwidths $k_0$ and $l_0$, the number of spatial units $N$, the number of interactions $\abs{S}$ and the degree of penalization $\lambda$.\footnote{Recall that $\lambda \in O\left(T^{-q_\lambda}\right)$, such that a higher $q_\lambda$ implies a faster decay of the penalty term.} To ensure that the set on which the performance bound in Theorem \ref{Thm:sgl} holds occurs with probability converging to one, conditions (i) and (ii) impose that the degree of penalization does not decay too fast. The optimal convergence rate is obtained by choosing $q_\lambda$ as large as possible without violating these conditions. Some concrete examples are provided in Remark \ref{remark:rate_examples}.

\begin{remark}\label{remark:rate_examples}
Insightful special cases can be examined based on Corollary \ref{cor:splash_rates}. For the sake of brevity, we consider two cases while focusing on the estimation error $P_\alpha\left(\hat{\bc}-\bc\right)$ and assuming errors with at least $d$ finite moments (Assumption \ref{assump:momentcond}(b1)). In the absence of within-group shrinkage ($\alpha=0$), Corollary \ref{cor:splash_rates} demonstrates that $P_0\left(\hat{\bc}-\bc\right) = O_p\left(T^{q_N - q_\lambda}\right)$, with $q_\lambda < \frac{1}{2} - \frac{2q_N}{d} - \frac{1}{2d}$ by choosing $\delta$ arbitrarily close to 1. The estimator now converges almost at rate $\frac{T^{1/2 - 1/2d}}{N^{1+2/d}\abs{\mathcal{G}_S}}$. For fixed $N$ and large $d$, this is close to the common $\sqrt{T}$-rate of fixed-dimensional settings without regularization. If shrinkage is imposed at the individual interaction level only ($\alpha=1)$, then $P_1\left(\hat{\bc}-\bc\right) = O_p\left(T^{q_s - q_\lambda}\right)$ and the estimation error converges almost at the rate $\frac{T^{1/2 - 1/2d}}{\abs{S}N^{2/d}}$. Noting that $N^{1+2/d}\abs{\mathcal{G}_S} > \abs{S}N^{2/d}$, we see that SPLASH($1$,$\lambda$) attains a convergence rate at least as fast SPLASH($0$,$\lambda$), and possibly faster when the sparsity is unstructured or the diagonals are highly sparse.
\end{remark}

\subsection{Exogenous variables}
We generalize model specification \eqref{eq:ModelSpec} by accommodating $K$ exogenous variables, i.e.
\begin{equation}
 \by_t = \bA\by_t + \bB\by_{t-1} + \sum_{k=1}^K \diag(\bbeta_k) \bx_{t,k} + \bepsilon_t, \qquad\qquad\qquad\qquad t=1,\ldots,T.
\label{eq:ModelSpecExo}
\end{equation}
Each vector $\bx_{t,k}=(x_{1t,k},\ldots,x_{Nt,k})\tran$ augments the spatio-temporal vector autoregression with an extra regressor. This regressor may vary over time and it is exogenous, i.e. we have $ \E(\bx_{t,k} \bepsilon_t\tran)=\bZeros$ for $k=1,\ldots,K$. For notational brevity, we consider the situation in which the exogenous regressors $x_{it,1}\ldots,x_{it,K}$ can only directly influence spatial unit $i$. This explains the diagonal structure in $\diag(\bbeta_k)$. In Remark \ref{remark:ExoSettingComparison} we argue that this simplification does not greatly hinder generality. In contrast to \citet{MaGuoWang2021}, we allow $\bbeta_k=(\beta_{1k},\ldots,\beta_{Nk})\tran$ to vary with location. We keep $K$ fixed.

To account for the exogenous variables, we modify the generalized Yule-Walker estimator of Section \ref{subsec:SPLASH}. We recall $\bSigma_j = \E(\by_t \by_{t-j}')$, and define the matrices $\bSigma_{j}^{x_k y}=\E(\bx_{t,k} \by_{t-j}')$ and $\bSigma_j^{x_k x_\ell}=\E(\bx_{t,k} \bx_{t-j,\ell}')$. Two sets of Yule-Walker equations, namely
\begin{subequations}
\begin{equation}
 \bSigma_1 = \bA \bSigma_1 + \bB \bSigma_0 + \sum_{k=1}^K \diag(\bbeta_k) \bSigma_1^{x_k y}
 \label{eq:YWwithEXOa}
\end{equation}
and
 \begin{equation}
 (\bSigma_{0}^{x_{j} y})\tran = \bA (\bSigma_{0}^{x_{j} y})\tran + \bB (\bSigma_1^{x_j y})\tran+ \sum_{k=1}^K \diag(\bbeta_k) \bSigma_0^{x_k x_j}, \qquad\text{for $j=1,\ldots, K$}, \label{eq:YWwithEXOb}
\end{equation}
\end{subequations}
are derived by post-multiplying the model by respectively $\by_{t-1}\tran$ and $\bx_{t,k}\tran$, and taking expectations. Compared to \eqref{eq:YWpopulation}, the Yule-Walker equations in \eqref{eq:YWwithEXOa} contain the additional term $\sum_{k=1}^K \diag(\bbeta_k) \bSigma_1^{x_k y}$ to provide information on $\bbeta_1,\ldots,\bbeta_K$. However, if $\bSigma_1^{x_k y}=\bZeros$ (e.g. when $\{\bx_{t,k}\}$ and $\{\by_t\}$ are independent and $\bbeta_k=\bzeros$), then \eqref{eq:YWwithEXOa} alone will not identify $\bbeta_k$. We therefore add the additional Yule-Walker equations in \eqref{eq:YWwithEXOb}. To develop the estimator, we combine \eqref{eq:YWwithEXOa} and \eqref{eq:YWwithEXOb} into
\begin{equation}
 \begin{bmatrix}
  \bSigma_1\tran \\
  \bSigma_0^{x_1 y} \\
  \vdots \\
  \bSigma_0^{x_K y}
 \end{bmatrix}
 =
 \begin{bmatrix}
  \bSigma_1\tran    & \bSigma_0 \\  
  \bSigma_0^{x_1 y} & \bSigma_1^{x_1 y}\\ 
  \vdots            & \vdots    \\
  \bSigma_0^{x_K y} & \bSigma_1^{x_K y}
 \end{bmatrix}
 \begin{bmatrix}
  \bA & \bB 
 \end{bmatrix}\tran
 +
 \sum_{k=1}^K
 \begin{bmatrix}
  (\bSigma_1^{x_k y})\tran \\
  \bSigma_0^{x_1 x_k} \\
  \vdots \\
  \bSigma_0^{x_K x_k}
 \end{bmatrix} \diag(\bbeta_k)
 := \bV^* \bC\tran + \sum_{k=1}^K \bW_k^* \diag(\bbeta_k).
\label{eq:EXOsystem}
\end{equation}
From this point onward, the development of the SPLASHX($\alpha$,$\lambda$) estimator mimics the reasoning of page \pageref{eq:singleL2contri} closely. First, we focus on the $i$\textsuperscript{th} spatial unit and collect all the nonzero coefficients of $\bA$ and $\bB$ (as stipulated by Assumption \ref{assump:bandedness}) in $\bc_i$. Letting $\bV_i^*$ denote the columns in $\bV^*$ related to $\bc_i$ and defining both $\bsigma_i^* = \begin{bmatrix} \bSigma_1 & (\bSigma_0^{x_1 y})\tran & \cdots & (\bSigma_0^{x_K y})\tran\end{bmatrix}\tran \be_i$ and $\bw_{ik}^*=\bW_k^* \be_i$, result \eqref{eq:EXOsystem} implies
$$
\bV_i^* \bc_i + \sum_{k=1}^K \bw_{ik}^* \beta_{ik}
= \bsigma_i^*.
$$
Second, we define (a) the sample counterparts of $\bSigma_j$, $\bSigma_j^{x_k y}$ and $\bSigma_j^{x_k x_l}$ as respectively $\hat \bSigma_j = \frac{1}{T} \sum_{t=j+1}^T \by_t \by_{t-j}\tran$, $\hat \bSigma_j^{x_k y }=\frac{1}{T} \sum_{t=j+1}^T \bx_{t,k} \by_{t-j}\tran$ and $\hat\bSigma_j^{x_k x_\ell}= \frac{1}{T} \sum_{t=j+1}^T \bx_{t,k} \bx_{t-j,\ell}\tran$, and (b) define the quantities $\hat\bsigma_i^*$, $\hat\bw_{ik}^*$ and $\hat\bV_i^*$ based on their underlying sample covariance matrix estimators. Finally, set $\hat\bsigma^*=(\hat\bsigma_1^{*\prime},\ldots,\hat\bsigma_N^{*\prime})\tran$, $\hat\bV^{*(d)}=\diag(\hat\bV_1^*,\ldots,\hat\bV_N^*)$, and $\hat\bW_k^{*(d)}=\diag(\hat\bw_{1k}^*,\ldots,\hat\bw_{Nk}^*)$. The SPLASHX($\alpha$,$\lambda$) objective function is
\begin{equation}
\begin{aligned}
\mathcal{L}_\alpha^*(\bbeta_1,\ldots,\bbeta_K,\bc;\lambda)
&= \norm{\hat{\bsigma}^* - \hat{\bV}^{*(d)}\bc - \sum_{k=1}^K \hat\bW_k^{*(d)} \bbeta_k }_2^2  \\
 &\qquad + \lambda\left( P_\alpha(\bc) + \sum_{k=1}^K (1-\alpha)\sqrt{N}\norm{\bbeta_k}_2 + \alpha \norm{\bbeta_k}_1 \right).
\label{eq:SGLobjfuncX}
\end{aligned}
\end{equation}
This objective function allows for the estimation of $\bbeta_1,\ldots,\bbeta_K$, sparse coefficients, completely sparse vectors $\bbeta_k$, and completely sparse diagonals in the coefficient matrices $\bA$ and $\bB$. There is a clear mathematical resemblance between the SPLASH and SPLASHX estimators. Accordingly, under appropriate modifications to Assumptions \ref{assump:stability}--\ref{assump:min_eigen_bound}, a finding similar to Theorem is attainable. We provide this result as Theorem \ref{Thm:Sglwithexo} and refer the reader to Supplement \ref{appendix:exogeneousvariables} for detailed assumptions and proofs.

\begin{theorem}\label{Thm:Sglwithexo}
Define $\bq = (\bc',\bbeta_1',\dots,\beta_K')'$, $S^* = \left\{ j : q_j \neq 0\right\}$ and
\begin{equation*}
    \bar\omega_\alpha^* = \max \left\{(1-\alpha) \Bigg(\sum_{g \in \mathcal{G}_S}\sqrt{\abs{g}} + \sqrt{N}\sum_{k=1}^K \ind{\bbeta_k \neq \bzeros}\Bigg),\alpha\sqrt{\abs{S^*}}\right\}.
\end{equation*}
Under Assumptions \ref{assump:stability_exogeneous}--\ref{assump:min_eigen_bound_exogeneous} and $\normoneinf{\bQ}\leq C_Q$, it holds that
\begin{equation*}
\begin{aligned}
  &\norm{\hat{\bV}^{*(d)}(\hat{\bc}-\bc)+\sum_{k=1}^K \hat\bW_k^{*(d)} (\hat\bbeta_k- \bbeta_k)}_2^2
  + \lambda\Bigg[(1-\alpha) \left(\sum_{g \in \mathcal{G}} \sqrt{\abs{g}}\norm{\hat{\bc}_g - \bc_g}_2 + \sum_{k=1}^K \sqrt{N}\norm{\hat\bbeta_k -\bbeta}_2 \right)\\
  &\qquad + \alpha \left(\norm{\hat{\bc}-\bc}_1 +\sum_{k=1}^K \norm{\hat\bbeta_k -\bbeta_k}_1\right) \Bigg] \leq \frac{64\bar{\omega}_\alpha^{*2}\lambda^2}{\phi_0^{*2}}
\end{aligned}
\end{equation*}
with a probability of at least
\begin{enumerate}[(a)]
    \item $1- 7(K+1)(K+2)\mathcal{P}_1^*\left(f^*(\lambda,\phi_0^*),N,T\right)$ under Assumption \ref{assump:momentcond}(b1) (polynomial tail decay), or
    \item $1- 7(K+1)(K+2)\mathcal{P}_2^*\left(f^*(\lambda,\phi_0^*),N,T\right)$ under Assumption \ref{assump:momentcond}(b2) (exponential tail decay),
\end{enumerate}
where $f^*(\lambda,\phi_0^*) = \min\left\{\tfrac{\lambda^{1/2}}{12\sqrt{6}} , \tfrac{\lambda}{144C_Q}, \tfrac{\lambda^{1/2}}{12\sqrt{6} C_\beta} , \tfrac{\lambda}{144C_Q C_\beta}, \tfrac{\phi_0^*}{12} \right\}$,
$$
 \mathcal{P}_1^*(\epsilon,N,T) = N^2\left[ \left(b_1 T^{(1-\delta)/3}+\frac{(K+2)Nb_3}{\epsilon}\right) \exp\left(-\frac{T^{(1-\delta)/3}}{2b_1^2}\right) + \frac{b_2 (K+2)^d N^d}{\epsilon^d T^{\frac{\delta}{2}(d-1)}}\right]
$$
for some $0<\delta<1$, and
$$
 \mathcal{P}_2^*(\epsilon,N,T) = N^2 \left[\frac{\kappa_1(K+2)N}{\epsilon} + \frac{2}{\kappa_2}\left(\frac{T \epsilon^2}{(K+2)N}\right)^{1/7} \right] \exp\left(- \frac{1}{\kappa_3}\left(\frac{T \epsilon^2}{(K+2)^2 N^2}\right)^{1/7} \right)
$$
All constants ($b_1$, $b_2$, $\kappa_1$, etc.) are positive and independent of $N$ and $T$, see Theorem \ref{th:diag_approx}.
\end{theorem}

\begin{remark}\label{remark:NoBanding}
The inclusion of exogenous variables affects the autocovariance structure of the data. For example, if $\bB=\bZeros$, then $\by_t= (\bI_n-\bA)^{-1}\big[\sum_{k=1}^K \diag(\bbeta_k) \bx_{t,k} + \bepsilon_t\big]$ and
$$
\E(\by_t \by_t\tran)= (\bI_N -\bA)^{-1}\left[ \sum_{k,\kappa=1}^K \diag(\bbeta_k) \E(\bx_{t,k}\bx_{t,\kappa}) \diag(\bbeta_\kappa) + \bSigma_\epsilon \right] (\bI_N - \bA\tran)^{-1}.
$$
Clearly, $\E(\by_t \by_t\tran)$ now also depends on the various second moments of the exogenous covariates. We do not make any a priori assumptions on $\E(\bx_{t,k}\bx_{t,\kappa})$ and thus define the SPLASHX($\alpha$,$\lambda$) in terms of the \emph{unbanded} autocovariance matrix estimators.
\end{remark}

\begin{remark}\label{remark:ExoSettingComparison}
 Defining the coefficient matrix in front of $\bx_{t,k}$ as diagonal is not restrictive. That is, by letting $\bx_{t,k+1}$ be a reordered version of $\bx_{t,k}$, the former's addition to the model can accommodate for the situation in which the dependent variable is influenced by the exogenous variable $\bx_{t,k}$ from multiple locations.
\end{remark}

\section{Simulations}\label{sec:MC}
\subsection{Simulation setting}\label{sec:simulation_setting}
In this section, we explore the finite sample performance of our estimator by Monte Carlo simulation. The data generating process underlying the simulations is the spatio-temporal VAR in \eqref{eq:ModelSpec}. We study $T\in\{500,1000,2000\}$ and draw all errors $\epsilon_{it}$ independently and $N(0,1)$ distributed. The matrices $\bA$ and $\bB$ and the cross-sectional dimension $N$ are specified in the two designs below. All simulation results are based on $N_{sim}=500$ Monte Carlo replications.

\bigskip
\noindent
\textbf{Design A (Banded specification)}: We revisit simulation Case 1 in \cite{Gao2019}. The matrices $\bA$ and $\bB$ are banded with a bandwidth of $k_0=3$. Specifically, the elements in the matrices $(\bA)_{i,j=1}^N$ and $(\bB)_{i,j}^N$ are generated according to the following two steps:
\begin{enumerate}[\itshape Step 1:]
    \item If $|i-j|=k_0$, then $a_{ij}$ and $b_{ij}$ are drawn independently from a uniform distribution on the two points $\{-2,2\}$. All remaining elements within the bandwidth are drawn from the mixture distribution $\omega I_{\{0\}}+(1-\omega) N(0,1)$ with $\Prob(\omega=1)=0.4$ and $\Prob(\omega=0)=0.6$.
    \item Rescale the matrices $\bA$ and $\bB$ from Step 1 to $\eta_1\times \bA / \norm{\bA}_2$ and $\eta_2\times \bB/\norm{\bB}_2$, where $\eta_1$ and $\eta_2$ are drawn independently from $U[0.4,0.8]$.\footnote{This rescaling does not necessarily imply $\norm{(\bI_N-\bA)^{-1} \bB}_2<1$ (stability). During the simulations we redraw the matrices $\bA$ and $\bB$ whenever $\norm{(\bI_N-\bA)^{-1} \bB}_2>0.95$.}
\end{enumerate}
We vary the cross-sectional dimension over $N\in\{25,50,100\}$.

\bigskip
\noindent
\textbf{Design B (Spatial grid with neighbor interactions)}: As in Figure \ref{fig:spatial_overview}, we consider an $(m\times m)$ grid of spatial units. For $m=5$ ($m=10$), this results in a cross-sectional dimension of $N=25$ ($N=100$). The matrix $\bA$ contains interactions between first horizontal and first vertical neighbours while all other coefficients are zero. The magnitude of these nonzero interactions are 0.2. For $m=5$ ($m=10$), the temporal matrix $\bB$ is a diagonal matrix with elements 0.25 (0.21) on the diagonal. The reduced form VAR matrix $\bC = (\bI_N - \bA)^{-1}\bB$ has a maximum eigenvalue of 0.814 (0.904).

\bigskip
For each design, we report simulation results for three sets of estimators. The first set includes the estimators developed in this paper: (1) the SPLASH(0,$\lambda$) estimator promotes non-sparse groups only, (2) SPLASH($0.5$,$\lambda$) provides equal weight to sparsity at the group and individual level, and (3) SPLASH(1,$\lambda$) encourages unstructured sparsity only.\footnote{The choice for $\alpha=0.5$ is solely made to illustrate the effect of combining both group and individual penalties. For different designs, this choice may or may not be optimal.} In congruence with Theorems \ref{th:diag_approx} and \ref{Thm:sgl}, we rely on banded autocovariance matrices $\hbanded{\hat\bSigma_0}{}$ and $\hbanded{\hat\bSigma_1}{}$. The bandwidth choice is determined by the bootstrap procedure described in \citet[][p. 7]{guowangyao2016}. Second, we include two unpenalized estimators in the spirit of \cite{Gao2019}: GMWY and GMWY($k_0$). The GMWY estimator implements generalized Yule-Walker estimation for banded $\bA$ and $\bB$ with the bandwidth being chosen by the selection rule proposed by \citet[][eq. 2.17]{Gao2019}, whereas GMWY$(k_0)$ is based on the true bandwidth $k_0$. To allow for easy comparison with the simulation results by the aforementioned authors, we implement these GMWY estimators without banding the covariance matrix estimators $\hat \bSigma_0 = \frac{1}{T} \sum_{t=2}^T \by_t \by_t \tran$ and $\hat \bSigma_1 = \frac{1}{T} \sum_{t=2}^T \by_t \by_{t-1}\tran$.\footnote{In unreported simulation results (available upon request), we find that the results are insensitive to this choice.} As GMWY$(k_0)$ is infeasible in practice, it is given a comparative advantage. 
The third set solely contains the $L_1$-penalized reduced form VAR(1) estimator (abbreviated PVAR). In detail, we consider the reduced form VAR($1$) specification $\by_t = \bC \by_{t-1}+\bu_t$ and estimate $\bC$ by minimizing $\mathcal{L}_{pvar}(\bC) = \sum_{t=2}^T \norm{\by_t - \bC\by_{t-1}}_2^2 + \lambda \sum_{i,j=1}^N\abs{c_{ij}}$. This estimator is well-researched in the literature \citep[see, e.g.][]{KockCallot2015,Gelper2016,Masini2019}, albeit in different settings. It will serve as a competitive benchmark for the forecasting performance of our proposed estimation procedure.

The forecasting performance of each estimator will be assessed using the \emph{Relative Mean-Squared Forecast Error (RMSFE)}. Using a superscript $j$ to index a specific Monte Carlo replication, the RMSFE is calculated as
\begin{equation}
    \text{RMSFE} = \frac{\sum_{j=1}^{N_{sim}}\norm{\by_{T+1}^j - \hat{\bC}^j\by_T^j}_2^2}{\sum_{j=1}^{N_{sim}}\norm{\by_{T+1}^j - \bC\by_T^j}_2^2}.
\label{eq:RMSFEdefinition}
\end{equation}

As the SPLASH and GMWY procedures estimate $\bA$ and $\bB$, we can also compare the estimation accuracy. Using the superscript $j$ as before, the \emph{Estimation Error (EE)} of the coefficient matrices are
\begin{equation}\label{eq:MatrixError}
    \text{EE}_{A}= \frac{1}{N_{sim}}\sum_{j=1}^{N_{sim}} \norm{\hat{\bA}^j - \bA}_2\text{ and }\text{EE}_{B}= \frac{1}{N_{sim}}\sum_{j=1}^{N_{sim}} \norm{\hat{\bB}^j - \bB}_2.
\end{equation}

Finally, a word on the selection of the the penalty parameter. For the SPLASH estimator, we calculate the maximum penalty, $\lambda_\max$, as the smallest value producing the zero solution for all values of $\alpha$, i.e.
\begin{equation*}
    \lambda_\max = \max\left(\max_{g \in \mathcal{G}} \frac{T^{-1}\norm{\hat{\bV}_{h,g}^{(d)\prime}\hat{\bsigma}_h}_2}{\sqrt{\abs{g}}}, \max_{1 \leq i \leq N_c}\abs{\hat{\bV}_{h,i}^{(d)\prime}\hat{\bsigma}_h}\right).
\end{equation*}
Given $\lambda_\max$, we define the smallest penalty $\lambda_\min$ as $10^{-4}\lambda_\max$ ($10^{-6}\lambda_\max$) for Design A (B) and construct an ordered grid of 20 equidistant values on a log-scale, say $\lambda_\max = \lambda_1 > \lambda_2 > \ldots > \lambda_{20} = \lambda_\min$. Estimating SPLASH solutions for each $\lambda_i$, a grid of $\alpha$-values, and each individual simulation trial is computationally expensive (especially for large $N$). We instead perform a small-scale preliminary analysis in which we draw a small set of simulations from Designs A and B on which we estimate all solutions for a given value of $T$. Then, we choose the order $i_T \in \lbrace 1,\ldots,20\rbrace$ that minimizes the RMSFE in this preliminary set of simulations. This process of choosing the order $i_T$ on a log-equidistant grid for each value of $T$, is equivalent to setting $\lambda = m_T\lambda_\max$ with $m_T =  10^{-4(i_T-1)/20}$ or $m_T =  10^{-6(i_T-1)/20}$ for designs A and B, respectively. For Design A (B), our selected orders for $T=\{500,1000,2000\}$ are $i_T = \{9,10,11\}$ ($i_T =\{10,11,12\}$), corresponding to $m_T \approx 0.025,0.015,0.01$ ($m_T \approx 0.002,0.001,0.0005$), respectively. Having fixed the preferred order or multiplier, it remains to estimate a single solution per $\alpha$-value, thus resulting in substantial reductions in computation time. The penalized VAR is computationally less expensive. Accordingly, we choose its penalty parameter based on a time series cross-validation (TSCV) scheme \citep[e.g.][]{Hyndman2018}. In our implementation of TSCV, the first 80\% of the data is used to fit multiple solutions on, which are then evaluated based on the MSFE obtained on the latter 20\% of the data. The preferred penalty is chosen as the solution that attains the smallest MSFE.\footnote{We also tried to select the penalty for the PVAR as the sparsest solution whose prediction error lies within one standard error of the minimum prediction error. This selection rule, however, did not lead to an improvement in forecast or estimation accuracy.}

\subsection{Simulation results}
\begin{table}[htp]
\centering
\caption{Simulation results for Design A (Banded specification).}
\resizebox{0.75\textwidth}{!}{%
\begin{threeparttable}
\begin{tabular}{ll cccc  ccc cc}
\toprule
$N$ & $T$ & SPLASH($0$,$\lambda$) & SPLASH($0.5$,$\lambda$) & SPLASH($1$,$\lambda$) & GMWY & GMWY($k_{0}$) & PVAR\tabularnewline
\midrule
\multicolumn{8}{l}{Panel 1: Mean-Squared Forecast Error (MSFE)}\tabularnewline
\midrule 
25 & 500 & 1.023 & 1.024 & 1.027 & 5.485 & 1.048 & 1.125\tabularnewline
 & 1,000 & 1.012 & 1.011 & 1.012 & 1.387 & 1.016 & 1.116\tabularnewline
 & 2,000 & 1.007 & 1.007 & 1.008 & 1.008 & 1.008 & 1.109\tabularnewline
50 & 500 & 1.026 & 1.027 & 1.034 & 21.802 & 1.034 & 1.115\tabularnewline
 & 1,000 & 1.013 & 1.013 & 1.015 & 1.043 & 1.012 & 1.113\tabularnewline
 & 2,000 & 1.007 & 1.008 & 1.008 & 1.010 & 1.007 & 1.110\tabularnewline
100 & 500 & 1.038 & 1.043 & 1.058 & 1.055 & 1.036 & 1.113\tabularnewline
 & 1,000 & 1.022 & 1.024 & 1.030 & 1.019 & 1.017 & 1.104\tabularnewline
 & 2,000 & 1.013 & 1.014 & 1.016 & 1.008 & 1.008 & 1.098\tabularnewline
\midrule
\multicolumn{8}{l}{Panel 2: Estimation Error in A (EEA)}\tabularnewline
\midrule 
25 & 500 & 0.616 & 0.641 & 0.756 & 1.174 & 0.707 & \tabularnewline
 & 1,000 & 0.561 & 0.582 & 0.688 & 1.095 & 0.579 & \tabularnewline
 & 2,000 & 0.526 & 0.534 & 0.620 & 1.027 & 0.490 & \tabularnewline
50 & 500 & 0.654 & 0.689 & 0.836 & 0.892 & 0.634 & \tabularnewline
 & 1,000 & 0.631 & 0.663 & 0.795 & 0.781 & 0.529 & \tabularnewline
 & 2,000 & 0.592 & 0.622 & 0.747 & 0.699 & 0.443 & \tabularnewline
100 & 500 & 0.652 & 0.690 & 0.854 & 0.769 & 0.599 & \tabularnewline
 & 1,000 & 0.664 & 0.705 & 0.855 & 0.686 & 0.527 & \tabularnewline
 & 2,000 & 0.628 & 0.666 & 0.798 & 0.596 & 0.433 & \tabularnewline
\midrule 
\multicolumn{8}{l}{Panel 3: Estimation Error in B (EEB)}\tabularnewline
\midrule 
25 & 500 & 0.278 & 0.281 & 0.312 & 0.414 & 0.245 & \tabularnewline
 & 1,000 & 0.233 & 0.232 & 0.256 & 0.339 & 0.183 & \tabularnewline
 & 2,000 & 0.201 & 0.196 & 0.214 & 0.298 & 0.142 & \tabularnewline
50 & 500 & 0.313 & 0.323 & 0.373 & 0.364 & 0.251 & \tabularnewline
 & 1,000 & 0.264 & 0.268 & 0.306 & 0.277 & 0.188 & \tabularnewline
 & 2,000 & 0.226 & 0.227 & 0.259 & 0.220 & 0.138 & \tabularnewline
100 & 500 & 0.346 & 0.362 & 0.430 & 0.356 & 0.264 & \tabularnewline
 & 1,000 & 0.291 & 0.300 & 0.347 & 0.268 & 0.197 & \tabularnewline
 & 2,000 & 0.248 & 0.251 & 0.285 & 0.206 & 0.148 & \tabularnewline
\bottomrule
\end{tabular}%
	\begin{tablenotes}
	 \footnotesize
	 \item \textbf{Note}: The relative mean-squared forecast error (RMSFE) and estimation errors ($\text{EE}_A$ and $\text{EE}_B$) are defined in \eqref{eq:RMSFEdefinition} and \eqref{eq:MatrixError}, respectively. In general, lower numbers indicate better performance. As PVAR estimates a reduced form VAR, there are no model errors for $\bA$ and $\bB$ to report for this method.
    \end{tablenotes}
    \end{threeparttable}
   }
    \label{table:SimulationsGMWY}
\end{table}

The results for Design A are reported in Table \ref{table:SimulationsGMWY}. First, we consider the predictive performance in Panel 1. For all methods, we observe a monotonic decrease in RMSFE when $T$ increases. The SPLASH estimators and GMWY($k_0$) exhibit the best overall forecast performance, with SPLASH outperforming for smaller sample sizes ($T=500$). Among the SPLASH estimators, SPLASH($0$,$\lambda$) attains the lowest RMSFE in the majority of specifications but differences are generally marginal. The penalized VAR forecasts are less accurate than the aforementioned methods. An explanation is that sparsity patterns in the reduced form representation are less prevalent and thus more difficult to exploit. Direct estimation of the contemporaneous spatial interactions thus delivers forecast improvements over regularized reduced form estimation. The GMWY estimator is highly competitive when $T=2000$ but performs notably worse for small $N$ and $T$. As GMWY has a tendency to select a too large bandwidth (as in \cite{Gao2019}, table 1), this is probably caused by the estimation of redundant parameters. Given that the majority of sparsity in this design comes from the small bandwidth of $\bA$ and $\bB$, which is fully exploited by the infeasible GMWY($k_0$) estimator, we consider it reassuring that the SPLASH estimators attains comparable, and occasionally better, forecast performance without necessitating an a priori specification of the bandwidth.

Next, we explore the estimation accuracy for $\bA$ and $\bB$ in Panels 2 and 3, respectively. As before, all estimators display an improvement in estimation accuracy when $T$ increases. The SPLASH($0$,$\lambda$) attains a lower estimation error than the SPLASH($0.5$,$\lambda$) estimator, which in turn performs better than the unstructured sparsity variant SPLASH($1$,$\lambda$). The tight bandwidth in this design implies that many diagonals ought to be set to zero, which seems to be best effectuated by means of the group penalty. The GMWY($k_0$) estimator appears to deliver somewhat more accurate estimates than SPLASH for larger values of $T$. This apparently slower convergence of the SPLASH estimator might, at least partly, be considered the price of not knowing the true sparsity pattern, as represented by the term $\bar{\omega}_\alpha$ in Theorem \ref{Thm:sgl}. It is worth mentioning, however, that the choice of penalty parameter is motivated based on the predictive performance, which may not be optimal from the perspective of estimation accuracy. Indeed, in an unreported analysis we find that the penalty that minimizes the estimation error is typically higher and delivers sparser solutions. Regarding the GMWY estimator, we note that the detrimental effect of overestimating the bandwidth in smaller sample sizes is again visible, with the estimation error being substantially larger for the $T=500$ setting.

\begin{table}[t]
\centering
\caption{Simulation results for Design B (Spatial grid with neighbor interactions).}
\resizebox{0.75\textwidth}{!}{%
\begin{threeparttable}
\begin{tabular}{ll cccccc}
\toprule
$N$ & $T$ & SPLASH($0$,$\lambda$) & SPLASH($\alpha$,$\lambda$) & SPLASH($1$,$\lambda$) & GMWY & GMWY($k_{0}$) & PVAR\tabularnewline
\midrule
\multicolumn{8}{l}{Panel 1: Mean-Squared Forecast Error (MSFE)}\tabularnewline
\midrule
25 & 500 & 1.012 & 1.011 & 1.012 & 9.924 & 539.490 & 1.108\tabularnewline
 & 1,000 & 1.004 & 1.004 & 1.005 & 29.538 & 66.157 & 1.067\tabularnewline
 & 2,000 & 1.005 & 1.005 & 1.004 & 894.224 & 527.024 & 1.042\tabularnewline
100 & 500 & 1.012 & 1.012 & 1.019 & 1.148 & 355.561 & 1.170\tabularnewline
 & 1,000 & 1.011 & 1.011 & 1.014 & 1.104 & 245.253 & 1.110\tabularnewline
 & 2,000 & 1.007 & 1.007 & 1.006 & 1.080 & 1.128 & 1.078\tabularnewline
\midrule 
\multicolumn{8}{l}{Panel 2: Estimation Error in A (EEA)}\tabularnewline
\midrule
25 & 500 & 0.329 & 0.335 & 0.467 & 4.277 & 0.433 & \tabularnewline
 & 1,000 & 0.279 & 0.276 & 0.377 & 3.931 & 0.316 & \tabularnewline
 & 2,000 & 0.240 & 0.229 & 0.287 & 3.840 & 0.227 & \tabularnewline
100 & 500 & 0.387 & 0.404 & 0.565 & 2.052 & 0.559 & \tabularnewline
 & 1,000 & 0.362 & 0.374 & 0.531 & 2.019 & 0.471 & \tabularnewline
 & 2,000 & 0.366 & 0.356 & 0.495 & 2.017 & 0.380 & \tabularnewline
\midrule 
\multicolumn{8}{l}{Panel 3: Estimation Error in B}\tabularnewline
\midrule 
25 & 500 & 0.105 & 0.116 & 0.173 & 1.298 & 0.138 & \tabularnewline
 & 1,000 & 0.085 & 0.090 & 0.134 & 1.078 & 0.098 & \tabularnewline
 & 2,000 & 0.068 & 0.070 & 0.102 & 0.987 & 0.069 & \tabularnewline
100 & 500 & 0.140 & 0.150 & 0.207 & 0.898 & 0.217 & \tabularnewline
 & 1,000 & 0.113 & 0.121 & 0.167 & 0.683 & 0.161 & \tabularnewline
 & 2,000 & 0.098 & 0.102 & 0.141 & 0.554 & 0.118 & \tabularnewline
\bottomrule
\end{tabular}%
	\begin{tablenotes}
	 \footnotesize
	 \item \textbf{Note}: The relative mean-squared forecast error (RMSFE) and estimation errors ($\text{EE}_A$ and $\text{EE}_B$) are defined in \eqref{eq:RMSFEdefinition} and \eqref{eq:MatrixError}, respectively. In general, lower numbers indicate better performance. As PVAR estimates a reduced form VAR, there are no model errors for $\bA$ and $\bB$ to report for this method.
    \end{tablenotes}
    \end{threeparttable}
   }
    \label{table:SimulationsSpatialGrid}
\end{table}

Simulation results for Design B are shown in Table \ref{table:SimulationsSpatialGrid}. The high RMSFEs for the GMWY estimators are most striking. In the setting $N=25$ and $T=500$, the GMWY estimator frequently selects a bandwidth equal to 1, translating to inferior performance across all metrics. The GMWY($k_0$) estimator, on the other hand, is based on the correct bandwidth. This method, however, forecasts far worse, while its estimation accuracy instead is competitive to SPLASH. Upon closer inspection, we find that the high RMSFE in this case is driven by a few extreme prediction errors. These prediction outliers in turn correspond to simulation trials in which the smallest absolute eigenvalues of the estimated matrix $\bI - \hat{\bA}$ are close to zero (see Fig \ref{fig:error_vs_eigen} in the Supplementary Appendix). This implies that the GMWY estimator may be prone to stability issues when the bandwidth is large relative to the dimension.\footnote{Recall that converting the spatial representation to the reduced form representation requires inverting $\bI-\hat{\bA}$.} Apparently, owing to the implementation of sparsity, the SPLASH estimator does not suffer from such stability issues. For $N=25$ and $T=2,000$, the bandwidth selection in GMWY improves, while its forecast performance ironically worsens as a result of the increasing stability issues. The remaining results tell the same story as in Design A; SPLASH(0,$\lambda$) and SPLASH(0.5,$\lambda$) are forecasting very close to the optimal forecast, and forecast notably better than the PVAR. While the forecast performance of SPLASH(1,$\lambda$) comes across as equivalent to the SPLASH implementation with group-penalization, the estimation accuracy is superior for the latter. Hence, the group penalty seems especially valuable for the purpose of model interpretation.

A small visual analysis provides further evidence on the favourable estimation accuracy obtained by SPLASH with shrinkage towards diagonally structured sparsity. We visualize the capability of recovering the correct sparsity pattern by displaying the absolute value of the coefficients as averaged across all $N_{sim}$ simulation runs. Figure \ref{Fig:A_plot} illustrates the similarity between the true matrix $\bA$ and the average magnitude of the estimated coefficients.

\begin{figure}[th]
\subfigure[]{
\includegraphics[width=0.45\textwidth]{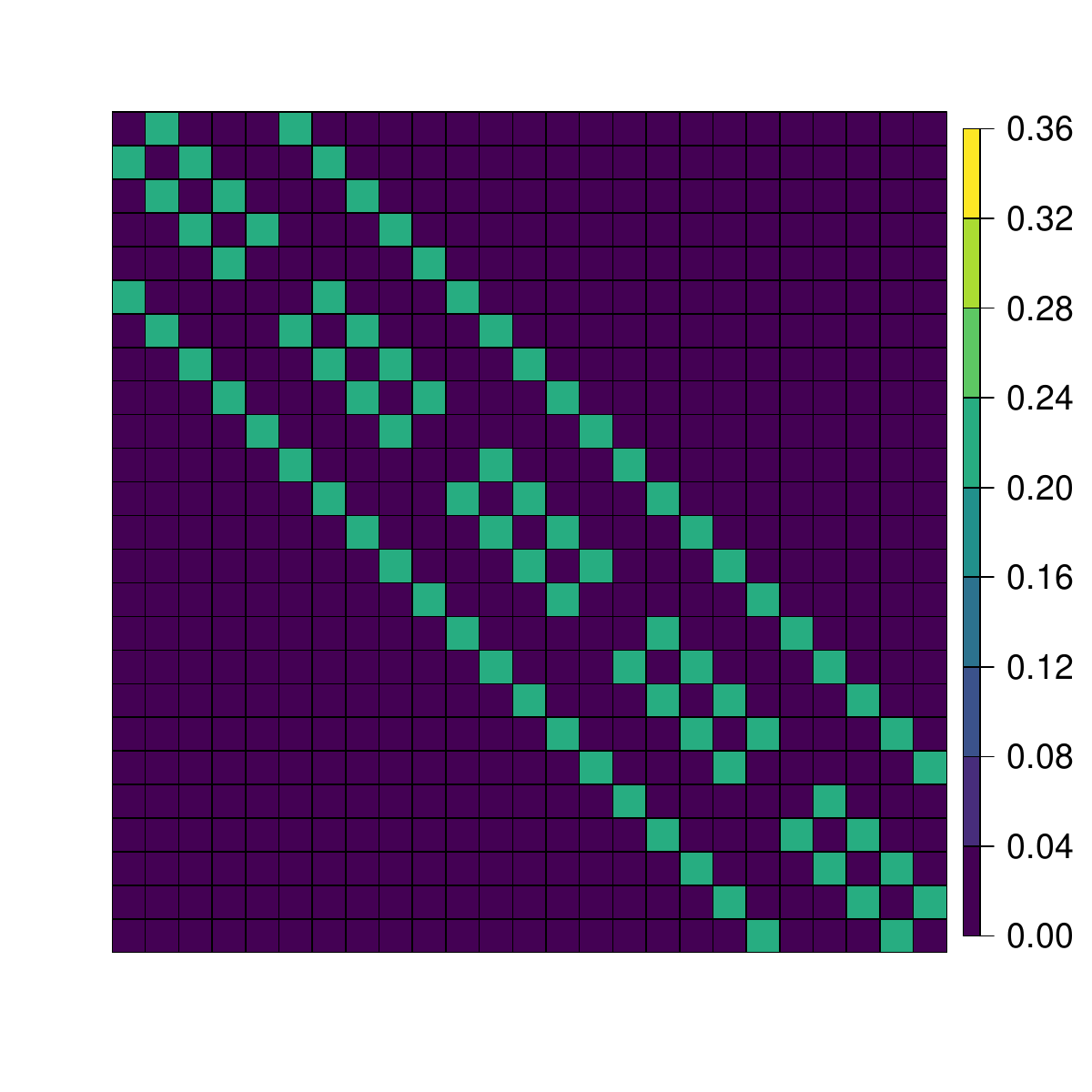}
}
\subfigure[]{
\includegraphics[width=0.45\textwidth]{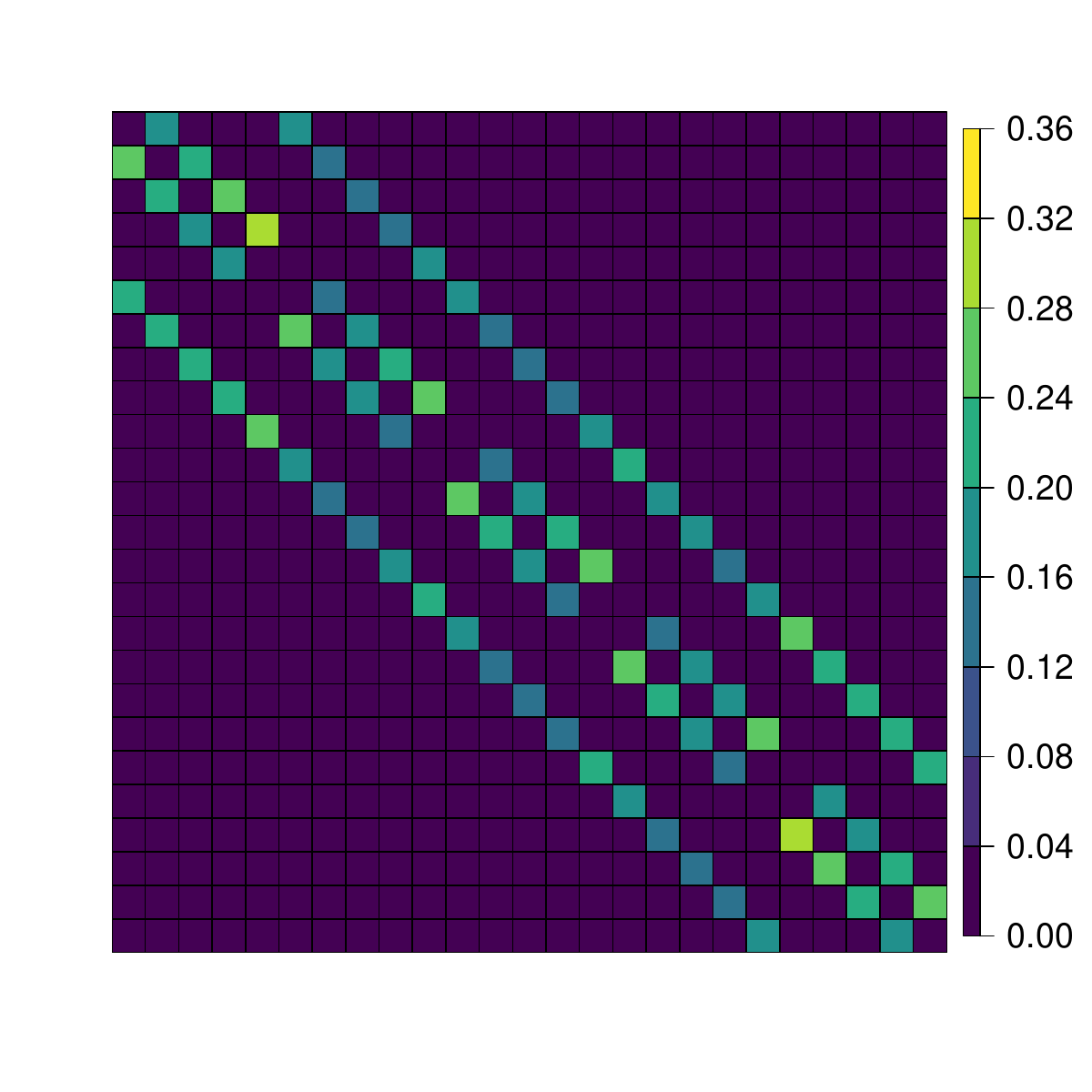}
}
\caption{Visualizations of the true and estimated spatial weight matrix $\bA$ for Design B. (a) The true spatial weight matrix $\bA$ implied by the $(5\times 5)$ spatial grid design (Design B with $m=5$). (b) The average absolute values of the entries in $\hat{\bA}$ as computed by SPLASH($0.5$,$\lambda$) for $N=25$ and $T=1000$. That is, the $(i,j)$\textsuperscript{th} entry in the matrix on the right equals $\frac{1}{N_{sim}} \sum_{k=1}^{N_{sim}} |\hat a_{ij}^k|$ with $\hat a_{ij}^k$ being the estimated $(i,j)$\textsuperscript{th} entry of $\bA$ in the $k$\textsuperscript{th} Monte Carlo replication.}
\label{Fig:A_plot}
\end{figure}

\begin{remark}
In elaborate, though unreported, visual analysis, we find that most zero diagonals are actually not estimated as exactly zero by the (sparse) group lasso. When tuning the penalty parameter by the BIC criterion, in which the number of estimated non-zeros is used as a proxy for the degrees of freedom, the true zero diagonals are typically estimated as exactly zeros. However, the increased amount of regularization that is required to effectuate this is detrimental to the forecast performance.
\end{remark}

\subsection{Simulations with exogenous regressors}

In this section, we examine the estimation performance of our estimator in the presence of exogenous regressors. Simulated data is drawn from
\begin{equation}
    \by_t = \bA\by_t + \bB\by_{t-1} + \diag(\bbeta_1)\bx_{t,1} + \diag(\bbeta_2)\bx_{t,2} + \bepsilon_t,
\label{eq:Results}
\end{equation}
where $\bA$ and $\bB$ are generated analogously to designs A and B in Section \ref{sec:simulation_setting}, and the coefficients of the exogenous regressors are given by $\bbeta_1 = \biota_N$ and $\bbeta_2 = \bm{0}_N$. Hence, only $\bx_{t,1}$ contributes to the variation in $\by_t$. Accordingly, we henceforth refer to $\bx_{t,1}$ and $\bx_{t,2}$ as the relevant and irrelevant exogenous regressor, respectively. All elements of exogenous variables and innovations are drawn i.i.d. from  $N(0,1)$. At each simulation trial, we implement the same estimators as considered in Section \ref{sec:simulation_setting}, with the exception of the penalized VAR which is omitted here. The selection of $\lambda$ for the SPLASH estimator is again done via the construction $\lambda_T = m_T\lambda_\max$, where the sequence of multipliers $m_T$ are the same as those in Section \ref{sec:simulation_setting}. 
Forecasts under \eqref{eq:Results} require predictions of the exogenous variables. This leads to two complications. First, in the reduced-form VAR(1) representation, $\by_t = \bB \by_{t-1} + \big[\bD\diag(\bbeta_1) \big]\bx_{t,1} + \big[\bD\diag(\bbeta_2) \big] \bx_{t,2} + \bD\bepsilon_t$, the coefficient matrices in front of $\bx_{t,1}$ and $\bx_{t,2}$ are no longer diagonal. This would cause an unfair comparison with PVAR so we decided to omit the penalized VAR approach from the comparison. Second, to avoid results that depend on the prediction method employed, we focus solely on the estimation accuracy.

The estimation accuracy for the estimates of $\bA$ and $\bB$ is compared on the basis of the metrics $EE_A$ and $EE_B$, as given in \eqref{eq:MatrixError}. In addition, we also report the average estimation errors of the coefficients for the relevant and irrelevant exogenous regressors, calculated as
\begin{equation}\label{eq:MatrixError_exo}
    \text{EE}_{R}= \frac{1}{N_{sim}}\sum_{j=1}^{N_{sim}} \norm{\hat{\bbeta}_1^j - \bbeta_1}_2,\text{ and }\text{EE}_{IR}= \frac{1}{N_{sim}}\sum_{j=1}^{N_{sim}} \norm{\hat{\bbeta}_2^j}_2.
\end{equation}
The results for Design A and B are reported in Tables \ref{table:SimulationsGMWY_exo} and  \ref{table:SimulationsSpatialGrid_exo}, respectively.

\begin{table}[htp]
\centering
\caption{Simulation results with Exogenous Regressors for Design A (Banded specification).}
\resizebox{0.75\textwidth}{!}{%
\begin{threeparttable}
\begin{tabular}{ll cccc  ccc cc}
\toprule
$N$ & $T$ & SPLASH($0$,$\lambda$) & SPLASH($0.5$,$\lambda$) & SPLASH($1$,$\lambda$) & GMWY & GMWY($k_{0}$)\tabularnewline
\midrule
\multicolumn{7}{l}{Panel 1: Estimation Error in A (EEA)}\tabularnewline
\midrule
25 & 500 & 0.257 & 0.258 & 0.268 & 0.333 & 0.328\tabularnewline
 & 1,000 & 0.200 & 0.193 & 0.191 & 0.224 & 0.223\tabularnewline
 & 2,000 & 0.146 & 0.138 & 0.133 & 0.154 & 0.154\tabularnewline
50 & 500 & 0.311 & 0.323 & 0.356 & 0.449 & 0.403\tabularnewline
 & 1,000 & 0.236 & 0.239 & 0.255 & 0.279 & 0.276\tabularnewline
 & 2,000 & 0.175 & 0.172 & 0.177 & 0.184 & 0.184\tabularnewline
100 & 500 & 0.349 & 0.371 & 0.429 & 0.680 & 0.465\tabularnewline
 & 1,000 & 0.287 & 0.299 & 0.331 & 0.349 & 0.340\tabularnewline
 & 2,000 & 0.217 & 0.221 & 0.238 & 0.232 & 0.232\tabularnewline
\midrule 
\multicolumn{7}{l}{Panel 2: Estimation Error in B (EEB)}\tabularnewline
\midrule
25 & 500 & 0.208 & 0.200 & 0.199 & 0.291 & 0.286\tabularnewline
 & 1,000 & 0.159 & 0.149 & 0.143 & 0.206 & 0.206\tabularnewline
 & 2,000 & 0.120 & 0.108 & 0.101 & 0.145 & 0.145\tabularnewline
50 & 500 & 0.241 & 0.240 & 0.255 & 0.352 & 0.306\tabularnewline
 & 1,000 & 0.186 & 0.179 & 0.182 & 0.220 & 0.216\tabularnewline
 & 2,000 & 0.141 & 0.132 & 0.131 & 0.154 & 0.154\tabularnewline
100 & 500 & 0.276 & 0.283 & 0.317 & 0.513 & 0.318\tabularnewline
 & 1,000 & 0.218 & 0.218 & 0.233 & 0.237 & 0.227\tabularnewline
 & 2,000 & 0.167 & 0.163 & 0.169 & 0.163 & 0.163\tabularnewline
\midrule
\multicolumn{7}{l}{Panel 3: Estimation Error in relevant exogenous regressor (EER)}\tabularnewline
\midrule
25 & 500 & 0.151 & 0.151 & 0.153 & 0.234 & 0.234\tabularnewline
 & 1,000 & 0.102 & 0.102 & 0.103 & 0.142 & 0.142\tabularnewline
 & 2,000 & 0.068 & 0.068 & 0.068 & 0.086 & 0.086\tabularnewline
50 & 500 & 0.181 & 0.183 & 0.187 & 0.364 & 0.366\tabularnewline
 & 1,000 & 0.119 & 0.119 & 0.121 & 0.228 & 0.228\tabularnewline
 & 2,000 & 0.079 & 0.079 & 0.079 & 0.137 & 0.137\tabularnewline
100 & 500 & 0.203 & 0.206 & 0.212 & 0.490 & 0.485\tabularnewline
 & 1,000 & 0.136 & 0.137 & 0.139 & 0.337 & 0.338\tabularnewline
 & 2,000 & 0.092 & 0.092 & 0.093 & 0.215 & 0.215\tabularnewline
\midrule 
\multicolumn{7}{l}{Panel 4: Estimation Error in irrelevant exogenous regressor (EEIR)}\tabularnewline
\midrule
25 & 500 & 0.087 & 0.086 & 0.086 & 0.122 & 0.122\tabularnewline
 & 1,000 & 0.059 & 0.059 & 0.059 & 0.079 & 0.079\tabularnewline
 & 2,000 & 0.043 & 0.042 & 0.042 & 0.054 & 0.054\tabularnewline
50 & 500 & 0.103 & 0.103 & 0.103 & 0.146 & 0.146\tabularnewline
 & 1,000 & 0.072 & 0.072 & 0.072 & 0.097 & 0.097\tabularnewline
 & 2,000 & 0.049 & 0.049 & 0.048 & 0.063 & 0.063\tabularnewline
100 & 500 & 0.117 & 0.116 & 0.116 & 0.157 & 0.157\tabularnewline
 & 1,000 & 0.084 & 0.083 & 0.083 & 0.111 & 0.111\tabularnewline
 & 2,000 & 0.058 & 0.058 & 0.058 & 0.075 & 0.075\tabularnewline
\bottomrule
\end{tabular}%
	\begin{tablenotes}
	 \footnotesize
	 \item \textbf{Note}: This table reports the average estimation errors in $\bA$, $\bB$, $\bbeta_1$ and $\bbeta_2$.
    \end{tablenotes}
    \end{threeparttable}
   }
    \label{table:SimulationsGMWY_exo}
\end{table}

\begin{table}[htp]
\centering
\caption{Simulation results with Exogenous Regressors for Design B (Spatial grid with neighbor interactions).}
\resizebox{0.75\textwidth}{!}{%
\begin{threeparttable}
\begin{tabular}{ll cccc  ccc cc}
\toprule
$N$ & $T$ & SPLASH($0$,$\lambda$) & SPLASH($0.5$,$\lambda$) & SPLASH($1$,$\lambda$) & GMWY & GMWY($k_{0}$)\tabularnewline
\midrule 
\multicolumn{7}{l}{Panel 1: Estimation Error in A (EEA)}\tabularnewline
\midrule
25 & 500 & 0.226 & 0.246 & 0.184 & 0.698 & 0.436\tabularnewline
 & 1,000 & 0.191 & 0.206 & 0.130 & 0.697 & 0.279\tabularnewline
 & 2,000 & 0.161 & 0.173 & 0.098 & 0.700 & 0.177\tabularnewline
100 & 500 & 0.427 & 0.439 & 0.483 & 0.790 & 0.979\tabularnewline
 & 1,000 & 0.332 & 0.339 & 0.383 & 0.778 & 0.735\tabularnewline
 & 2,000 & 0.243 & 0.247 & 0.293 & 0.780 & 0.503\tabularnewline
\midrule
\multicolumn{7}{l}{Panel 2: Estimation Error in B (EEB)}\tabularnewline
\midrule
25 & 500 & 0.161 & 0.162 & 0.114 & 0.554 & 0.504\tabularnewline
 & 1,000 & 0.147 & 0.148 & 0.083 & 0.538 & 0.372\tabularnewline
 & 2,000 & 0.131 & 0.131 & 0.061 & 0.497 & 0.275\tabularnewline
100 & 500 & 0.330 & 0.337 & 0.362 & 0.686 & 0.783\tabularnewline
 & 1,000 & 0.247 & 0.252 & 0.263 & 0.672 & 0.581\tabularnewline
 & 2,000 & 0.214 & 0.227 & 0.183 & 0.651 & 0.442\tabularnewline
\midrule
\multicolumn{7}{l}{Panel 3: Estimation Error in relevant exogenous regressor (EER)}\tabularnewline
\midrule
25 & 500 & 0.149 & 0.148 & 0.137 & 0.416 & 0.340\tabularnewline
 & 1,000 & 0.093 & 0.092 & 0.087 & 0.306 & 0.206\tabularnewline
 & 2,000 & 0.060 & 0.060 & 0.058 & 0.282 & 0.118\tabularnewline
100 & 500 & 0.275 & 0.269 & 0.239 & 0.854 & 0.730\tabularnewline
 & 1,000 & 0.197 & 0.195 & 0.164 & 0.669 & 0.560\tabularnewline
 & 2,000 & 0.150 & 0.146 & 0.109 & 0.446 & 0.391\tabularnewline
\midrule 
\multicolumn{7}{l}{Panel 4: Estimation Error in irrelevant exogenous regressor (EEIR)}\tabularnewline
\midrule 
25 & 500 & 0.107 & 0.105 & 0.097 & 0.361 & 0.137\tabularnewline
 & 1,000 & 0.070 & 0.069 & 0.065 & 0.279 & 0.091\tabularnewline
 & 2,000 & 0.047 & 0.046 & 0.044 & 0.189 & 0.058\tabularnewline
100 & 500 & 0.217 & 0.208 & 0.168 & 0.390 & 0.165\tabularnewline
 & 1,000 & 0.150 & 0.145 & 0.119 & 0.383 & 0.129\tabularnewline
 & 2,000 & 0.098 & 0.095 & 0.079 & 0.325 & 0.092\tabularnewline
\bottomrule
\end{tabular}%
	\begin{tablenotes}
	 \footnotesize
	 \item \textbf{Note}: This table reports the average estimation errors in $\bA$, $\bB$, $\bbeta_1$ and $\bbeta_2$.
    \end{tablenotes}
    \end{threeparttable}
   }
    \label{table:SimulationsSpatialGrid_exo}
\end{table}

First, we consider the results for Design A. The first two panels display the average estimation errors in $\bA$ and $\bB$, respectively. Reassuringly, all estimators display a clear monotonic decrease in estimation accuracy with growing sample size. Comparing the SPLASH estimators among each other, we observe that shrinkage towards group sparsity is most beneficial in high-dimensional settings ($T=500$ or $N=100$). In these instances, the SPLASH($0$,$\lambda$) and SPLASH($0.5$,$\lambda$) estimators obtain the lowest estimation error across all methods. Conversely, when the dimension is small ($N=25$) and sample size is large ($T=2000$), we find little gain in penalizing towards structured sparsity and the SPLASH($1$,$\lambda$) outperforms all other methods. Furthermore, the SPLASH estimators attain a lower estimation error than the GMWY estimators for almost all settings, with the performance gains attained by SPLASH being most pronounced in the case where the sample size is small, i.e. when the exploitation of sparsity matters most. Comparing the GMWY estimators, we find that, in lower-dimensional settings, using a data-driven selection of the bandwidth performs comparable to relying on the true bandwidth. However, when $N=100$ and $T=500$, we find that the bandwidth selection procedure over-estimates the true bandwidth in roughly 40\% of the simulation trials. Accordingly, the GMWY estimator attains inferior estimation accuracy in this particular setting. Regarding the exogenous regressors, we observe a similar monotonic decrease in the estimation error for the coefficients of both the relevant and irrelevant exogenous regressor. The SPLASH estimator outperforms the GMWY estimators across all dimensions and sample sizes, with the performance gain again being most prominent in the higher-dimensional settings. The estimation error obtained by SPLASH for the irrelevant exogenous regressor is remarkably small, further demonstrating the benefits of the incorporated shrinkage.

The results for Design B depict a similar, if not more compelling, story. The SPLASH estimators again outperform across all settings, with the performance differentials between SPLASH and GMWY being more pronounced compared to Design A. Again, we observe that the SPLASH(1,$\lambda$) estimator seems to outperform based on $EE_A$ and $EE_B$ for $N=25$, whereas the SPLASH(0,$\lambda$) and SPLASH($0.5$,$\lambda$) estimators do better when $N=100$. We conjecture that shrinkage towards structured sparsity only becomes beneficial when the group sizes are sizable enough, at which point the accumulation of selection errors by SPLASH(1,$\lambda$) starts to deteriorate the overall estimation accuracy. Contrasting the performance of SPLASH to the GMWY estimators, we observe that the exploitation of sparsity within the bandwidth results in substantial performance gains across all specifications and coefficient matrices. Moreover, the bandwidth selection procedure of the GMWY estimator now frequently selects very small bandwidths. This negatively impacts the estimation accuracy when $N=25$, whilst having a positive impact when $N=100$. In the latter case, the number of parameters to estimate is simply too large without further regularization, such that one might be better off by forcing most diagonals to zero, even if some of those are relevant. Interestingly, the inability to exploit sparsity also affects the estimation accuracy for the relevant exogenous regressors, as the third panel reveals a sizeable difference in the $EE_R$ between the SPLASH and GMWY estimators. Regarding the irrelevant exogenous regressor, we find that the $EE_{IR}$ is substantially larger for GMWY, but comparable across the SPLASH and GMWY($k_0$) estimators. 

Overall, SPLASH unambiguously attains more accurate estimates of all coefficient matrices in the spatial VAR with exogenous regressors. In line with expectations, the performance gain of SPLASH is most notable in high-dimensional designs with substantial degrees of sparsity.  However, even in lower-dimensional designs in which the degree of sparsity is less, SPLASH remains competitive to the GMWY estimators.

\section{Empirical Application} \label{sec:empapplic}

Nitrogen dioxide (NO\textsubscript{2}) is emitted during combustion of fossil fuels (e.g. by motor vehicles) and it has been associated with adverse effects on the respiratory system.\footnote{The direct health effect of nitrogen dioxide is difficult to determine because its emission process is typically accompanied with the emission of other air pollutants (see, e.g. \cite{brunekreefholgate2002}).} The Air Quality Standards Regulations 2010 requires a regular monitoring of NO\textsubscript{2} concentration levels in the UK.\footnote{Source: https://www.legislation.gov.uk/uksi/2010/1001/contents/made.} Using satellite data, we examine the empirical performance of the SPLASH estimator when predicting daily NO\textsubscript{2} concentrations in Greater London. This satellite data is publicly available via the Copernicus Open Access Hub and we consider the time span from 1 August 2018 to 18 October 2020.\footnote{See \url{http://www.tropomi.eu/} for more info on TROPOMI data products and use \url{https://scihub.copernicus.eu/} to access the database.} The original $\text{NO}_2$ concentrations are reported in mol/m$^2$, which we convert to mol/cm$^2$ to avoid numerical instabilities caused by small-scale numbers. The far majority of measurements are captured between 11:00 and 14:00 UTC. The area of interest is divided into a ($5 \times 9$) grid, implying that longitudes and latitudes increment by approximately 0.2 from cell to cell (see Figure \ref{Fig:NO2_coefficient_plot}, part c). All available $\text{NO}_2$ measurements are averaged within each cell and within the same day. The resulting data set contains 0.8\% missing observations, which we impute using the \texttt{Multivariate Time Series Data Imputation (mtsdi)} R package.\footnote{This imputation method is proposed by \cite{JungerDeLeon2015} to impute missing values in time series for air pollutants. The package is written by the same authors and currently maintained by W. L. Junger.}

A rolling-window approach is used to assess the predictive power of the SPLASH estimator. Each window contains 80\% of the data (641 days) allowing 160 one-step ahead forecasts to be made. For each window, we proceed along the following four steps: (i) de-mean the data, (ii) determine the hyperparameters and estimate each model, (iii) produce a forecast for the de-meaned data, and (iv) add the means back to the forecast. In addition to the estimators described in the simulation section (Section \ref{sec:MC}, see page \pageref{sec:MC}), we add another forecast: the window's mean. This new forecast is abbreviated CONST and all other forecasts follow the notational conventions from the simulation section. For SPLASH, we follow the procedure described in Section \ref{sec:simulation_setting} and set $\lambda = 1.8\times 10^{-4}\lambda_\max$. The spatial grid contains $N=5\times 9 = 45$ spatial units, such that the SPLASH($\alpha$,$\lambda$) models contain $2N^2-N = 4,005$ parameters. For the purpose of identifiability, we band the spatial matrix $\bA$ and autoregressive matrix $\bB$ such that $a_{ij} = b_{ij} = 0$ for $\abs{i-j} > \lfloor N/4 \rfloor = 11$. By ordering the spatial units vertically, this banding puts no restrictions on the vertical interactions but allows no more than second-order interaction between horizontal neighbours (see Figure \ref{fig:spatial_grid_London} in the Supplementary Appendix \ref{App:Satellite_Data} for details). 

The forecast performance is measured along three metrics and is always expressed relative to the $L_1$-penalized reduced form VAR($1$) (PVAR) benchmark. That is, we report: (i) the number of spatial units that are predicted more accurately than the PVAR method (\#wins), (ii) the number of spatial units that are predicted \textit{significantly} more accurately based on a Diebold-Mariano test at a 5\% significance level (\#sign. wins), and (iii) the average loss relative to the penalized VAR over all spatial units. These three metrics are calculated based on two loss functions for the forecast errors, namely the mean squared forecast error (MSFE) and the mean absolute forecast error (MAFE). We additionally report the MAFE because the $\text{NO}_2$ column densities display several abrupt spikes which may carry too much weight when relying on a squared loss function. The results are reported in Table \ref{tab:results_satellite}.

\begin{table}[t]
\centering
\caption{Forecast performance of various methods for $\text{NO}_2$ satellite data on a $(5\times 9)$ grid of observations.}
\begin{threeparttable}
\begin{tabular}{l cccc ccc}
\toprule
 & \multicolumn{3}{c}{MSFE} &  & \multicolumn{3}{c}{MAFE}\tabularnewline
 & \#wins & \#sign. wins & RMSFE &  & \#wins & \#sign. wins & RMAFE\tabularnewline
\midrule
CONST & 0 & 0 & 1.187 &  & 0 & 0 & 1.139\tabularnewline
GMWY & 0 & 0 & 153.393 &  & 0 & 0 & 4.855\tabularnewline
SPLASH($0$,$\lambda$) & 45 & 42 & 0.919 &  & 45 & 44 & 0.941\tabularnewline
SPLASH($0.5$,$\lambda$) & 44 & 39 & 0.931 &  & 45 & 42 & 0.949\tabularnewline
SPLASH($1$,$\lambda$) & 44 & 30 & 0.940 &  & 45 & 35 & 0.957\tabularnewline
\bottomrule
\end{tabular}%
	\begin{tablenotes}
	 \footnotesize
	 \item \textbf{Note}: Number of grid points (out of $N=45$) with lower prediction errors (\#wins) and significantly lower prediction errors (\#sign.) compared to the $L_1$-penalized reduced form VAR($1$) estimator (PVAR). Relative MAFE and MSFE are abbreviated by RMAFE and RMSFE, respectively. Values below (above) 1 indicate superior (inferior) performance compared to PVAR.
    \end{tablenotes}
    \end{threeparttable}
\label{tab:results_satellite}
\end{table}

We first look at the mean squared forecast errors (MSFEs). The window-mean forecast (CONST) clearly does not improve the benchmark PVAR forecast for any spatial unit. However, this forecast still attains a RMSFE of 1.185, potentially indicating a low predictability of $\text{NO}_2$ column densities. The GMWY approach obtains the worst forecast performance, possibly because a large bandwidth is needed to allow for second-order horizontal interaction and, consequently, a large number of parameters to estimate. With an RMSFE of 0.919, SPLASH(0,$\lambda$) does manage to improve upon the benchmark. In fact, the MSFEs for all 45 spatial units are smaller than that of the benchmark, 42 of which are found to be significant by a Diebold-Mariano test based on the squared forecast errors. Allowing for sparsity within groups does not seem to deliver additional forecast improvements, as SPLASH(0.5,$\lambda$) attains a slightly worse forecast performance and significantly outperforms the benchmark for only 39 locations. Completely omitting regularization at the group level results in a further deterioration of the forecast performance, with SPLASH(1,$\lambda$) attaining an RMSFE of 0.94 and significantly beating the benchmark for 30 out 45 spatial units. We take this as evidence that the ability to promote diagonally structured sparsity is indeed beneficial in real-life spatial applications, although even estimating the spatial VAR with unstructured sparsity manages to improve upon regularized reduced form VAR estimation.

Next, we focus on the mean absolute forecast error (MAFE). The results are qualitatively similar to those obtained based on the MSFE. In particular, GMWY still has the worst forecast accuracy for GMWY and SPLASH(0,$\lambda$) continues to perform best. The GMWY method, while still standing out, does not score as poorly anymore based on the RMAFE. We conjecture that the absence of regularization may increase sensitivity to noise, thereby resulting in particularly high squared forecast errors at periods of atypical NO$_2$ concentrations.

Finally, we illustrate the second key benefit of SPLASH-type estimators: interpretability. Recall that we convoluted satellite images to a ($5 \times 9$) grid of spatial units. To examine the relevant interactions between these spatial units, we provide several visualizations off the spatial weight matrices estimated by the SPLASH(0.5,$\lambda$) estimator. First, in Figure \ref{Fig:NO2_coefficient_plot}(a), we visualize the absolute magnitude of the spatial interactions. A clear diagonal pattern emerges, with the two diagonals closest to the principal diagonal and the two outer diagonals containing the largest interactions. These four diagonals correspond to first-order vertical and second-order horizontal interactions, respectively. The additional two diagonals, that are sandwiched in between the former, contain the first-order horizontal interactions between spatial units, which surprisingly seem to be smaller in magnitude. In Figure \ref{Fig:NO2_coefficient_plot}(b), each cell indicates the proportion of rolling windows the corresponding spatial interaction is estimated as being non-zero. These proportions are either one (yellow) or zero (purple) indicating very stable selection across samples. It becomes apparent that in addition to the six diagonals that were clear from panel (a), two additional diagonals are always selected, which contain the first order diagonal interactions between spatial units. To facilitate interpretation of this sparsity pattern, we provide a spatial plot of our region of interest with the spatial grid overlaid (Figure \ref{Fig:NO2_coefficient_plot}(c)). We explicitly visualize the interactions implied by \ref{Fig:NO2_coefficient_plot}(b) for two pixels -- pixel 1 (left-top) and pixel 23 (center) -- using arrows whose thickness is determined by the average absolute magnitudes estimated in Figure \ref{Fig:NO2_coefficient_plot}(a). The emerging pattern of spatial interactions shows clearly the interactions between NO$_2$ concentrations of neighbouring districts in London. The wider horizontal interactions, as well as the diagonal interactions, may be explainable by the ``prevailing winds'', which come from the West or South-West and are the most commonly occurring winds in London. Overall, the intuitive sparsity patterns that arise, in combination with the improvement in forecast performance, are encouraging and provide empirical validation for the use of SPLASH on spatial data, especially when the spatial units follow a natural ordering on a spatial grid.

\begin{figure}[htp]
\subfigure[Average absolute magnitudes]{
\includegraphics[width=0.45\textwidth]{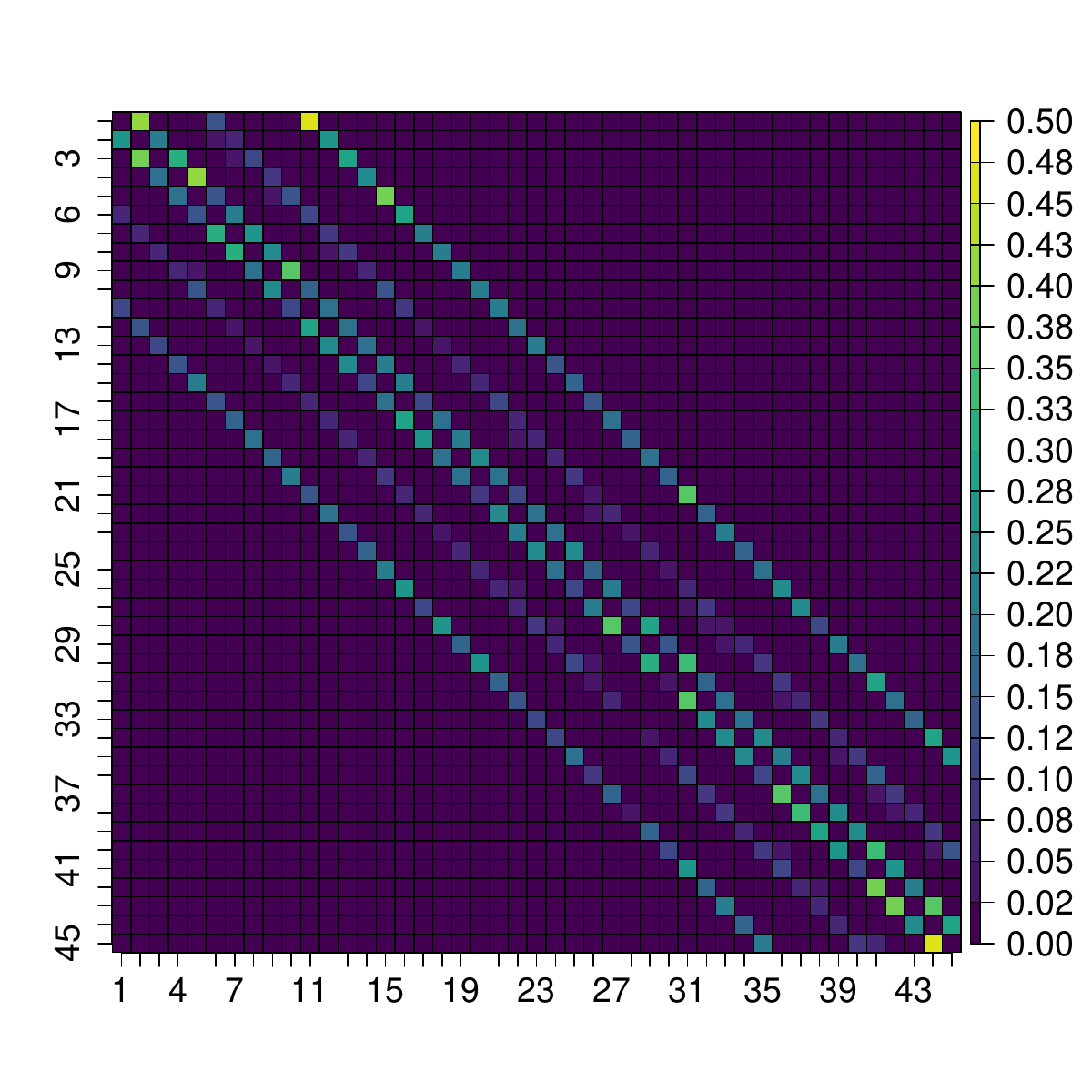}
}
\subfigure[Selection proportions]{
\includegraphics[width=0.45\textwidth]{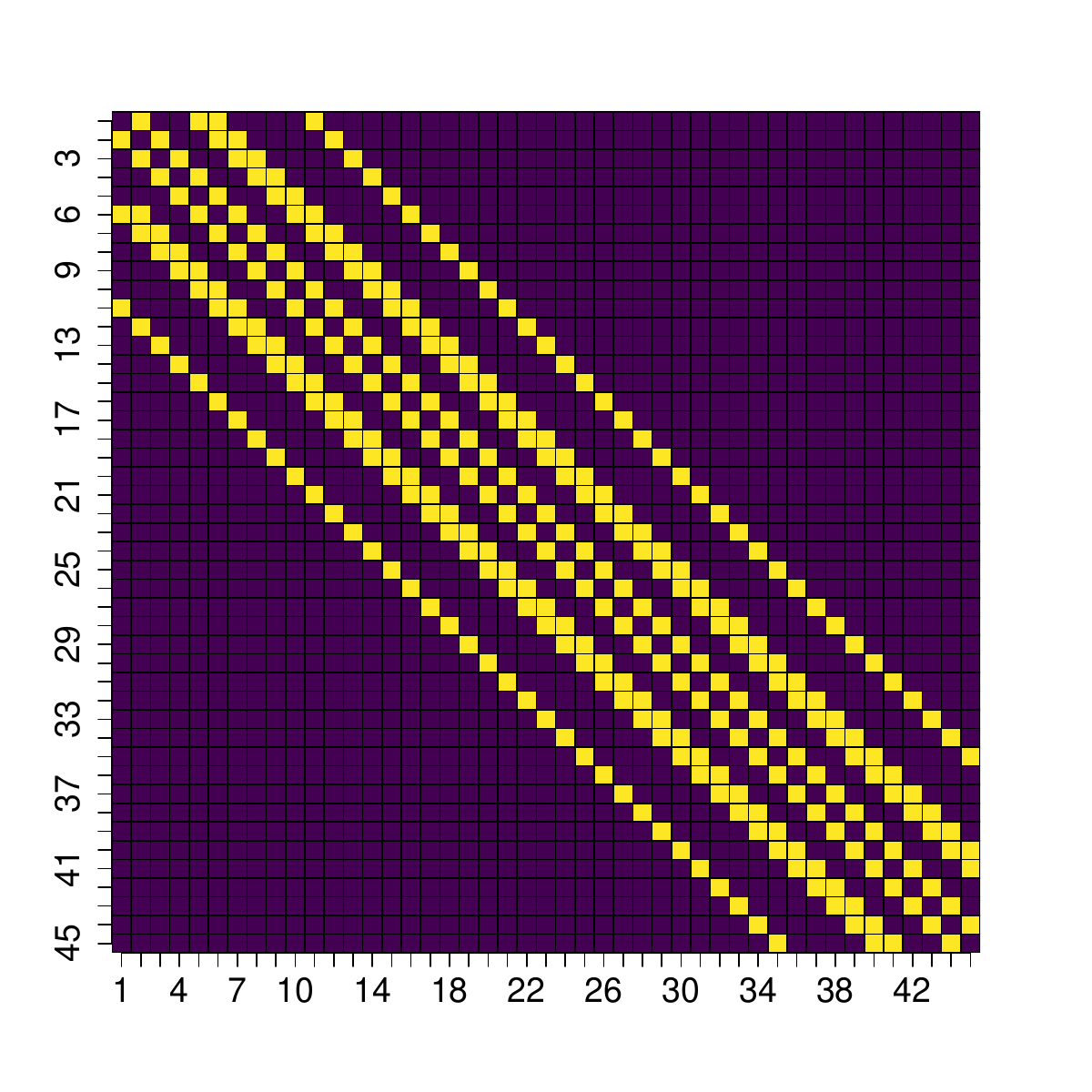}
}

\subfigure[Spatial plot of London]{
\includegraphics[width=0.9\textwidth]{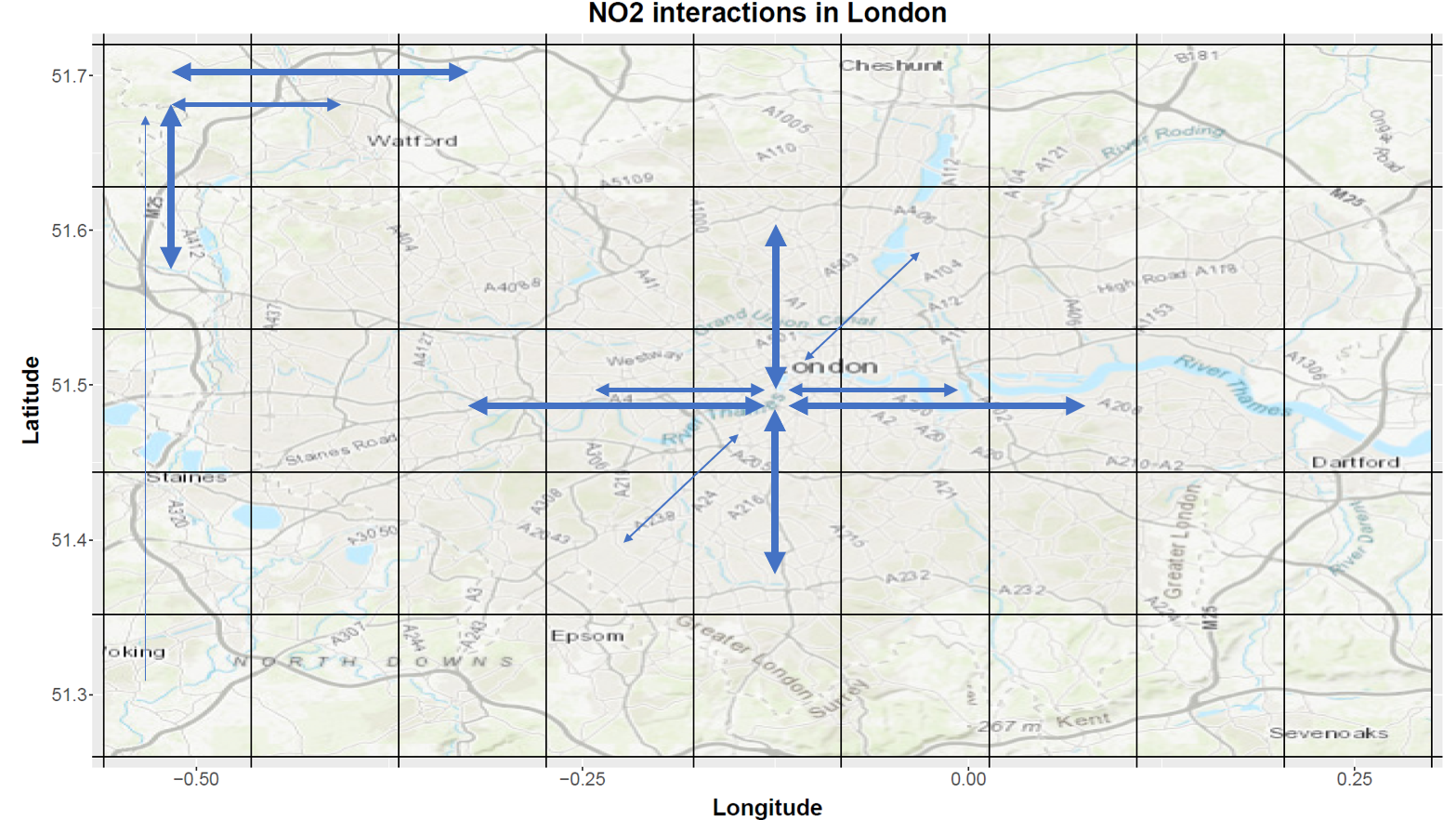}
}
\caption{Sparsity patterns for estimates of $\bA$ based on rolling window samples.}
\label{Fig:NO2_coefficient_plot}
\end{figure}

\section{Conclusion} \label{sec:conclusion}
In this paper, we develop the Spatial Lasso-type Shrinkage (SPLASH) estimator, a novel estimation procedure for high-dimensional spatio-temporal models. The SPLASH estimator is designed to promote the recovery of structured forms of sparsity without imposing such structure a priori. We derive consistency of our estimator in a joint asymptotic framework in which the number of both spatial units and temporal observations diverge. To solve the identifiability issue, we rely on a relatively non-restrictive assumption that the coefficient matrices in the spatio-temporal model are sufficiently banded. Based on this assumption, we consider banded estimation of high-dimensional spatio-temporal autocovariance matrices, for which we derive novel convergence rates that are likely to be of independent interest. The SPLASHX extension explains how to include exogenous variables. As an application, we use SPLASH to predict satellite-measured NO\textsubscript{2} concentrations in London. We find evidence for spatial interactions between neighbouring regions. In addition, our estimator obtains superior forecast accuracy compared to a number of competitive benchmarks, including the recently introduced spatio-temporal estimator by \citet{Gao2019} (the inspiration for the development of SPLASH).

\section{Acknowledgements}
This paper (or earlier versions hereof) has been presented during the internal seminar of the Quantitative Economics department of Maastricht University, the Econometrics Internal Seminar (EIS) at Erasmus University Rotterdam, the Bernoulli-IMS One World Symposium, the workshop on Dimensionality Reduction and Inference in High-dimensional Time Series at Maastricht University, the 2021 Annual Conference of the International Association for Applied Econometrics (IAAE), the 5\textsuperscript{th} Conference on Econometric Models of Climate Change, the internal seminar at Tor Vergata University of Rome, and the 2021 $(\text{EC})^2$ Conference. We gratefully acknowledge comments and feedback from the participants. Suggestions by Stephan Smeekes and Ines Wilms were particularly helpful so we thank them explicitly. All remaining errors are our own.

\label{Bibliography}
\bibliographystyle{apalike}
\bibliography{Bibliography}

\newpage
\begin{appendices}
\begin{small}
\section{Lemmas}\label{app:selected_lemmas}

\begin{lemma}\label{Lemma:min_vs_restricted_eigen}
Define the quantities $N_c = \abs{\bc}$, $S = \lbrace j : c_j \neq 0 \rbrace$, $\mathcal{G}_S = \left\lbrace g \in \mathcal{G} : \bc_g \neq \bm{0}\right\rbrace$, $\mathcal{G}_S^c = \left\lbrace g \in \mathcal{G} : \bc_g = \bm{0}\right\rbrace$, $\bar\omega_\alpha = \max \left\{(1-\alpha)\sum_{g \in \mathcal{G}_S}\sqrt{\abs{g}},\alpha\sqrt{\abs{S}}\right\}$ and consider
\begin{equation*}
    \bDelta \in \mathcal{C}_{N_c}(\mathcal{G},S) := \left\lbrace \bDelta \in \mathbb{R}^{N_c} : P_{\alpha,S^c}(\bDelta) \leq 3P_{\alpha,S}(\bDelta) \right\rbrace,
\end{equation*}
where 
\begin{equation*}
    \begin{split}
        P_{\alpha,S^c}(\bDelta) &= (1-\alpha)\sum_{g \in \mathcal{G}_S^c}\sqrt{\abs{g}}\norm{\Delta_g}_2 + \alpha\norm{\bDelta_{S^c}}_1,\text{ and }\\
        P_{\alpha,S}(\bDelta) &= (1-\alpha)\sum_{g \in \mathcal{G}_S}\sqrt{\abs{g}}\norm{\Delta_g}_2 + \alpha\norm{\bDelta_S}_1.
    \end{split}
\end{equation*}
Then, under Assumption \ref{assump:min_eigen_bound}, it holds that
\begin{equation}\label{eq:SGL_comp_cond}
    \min_{\bDelta \in \mathcal{C}_{N_c}(\mathcal{G},S)} \frac{\bar{\omega}_\alpha\norm{\bV^{(d)}\bDelta}_2}{P_{\alpha,S}(\bDelta)} \geq \frac{\phi_0}{2}.
\end{equation}
\end{lemma}
\begin{proof}
First, we show that \eqref{eq:SGL_comp_cond} is bounded from below by 0.5 times the smallest singular value of $\bV^{(d)}$. For $0 \leq \alpha < 1$, it holds that
\begin{equation*}
    \frac{(1-\alpha)\sum_{g \in \mathcal{G}_S}\sqrt{\abs{g}}\norm{\bDelta_g}_2}{\bar{\omega}_\alpha} \leq \frac{\norm{\bDelta}_2\sum_{g \in \mathcal{G}_S}\sqrt{\abs{g}}}{\sum_{g \in \mathcal{G}_S}\sqrt{\abs{g}}} = \norm{\bDelta}_2
\end{equation*}
Alternatively, for $\alpha = 1$, we have $\frac{\alpha\norm{\bDelta_S}_1}{\bar{\omega}_\alpha} = \frac{\alpha \abs{S} \norm{\bDelta_S}_2}{\bar{\omega}_\alpha} \leq \norm{\bDelta_S}_2 \leq \norm{\bDelta}_2$.
Combining the two previous results yields
\begin{equation}
    \frac{P_{\alpha,S}(\bDelta)}{\bar{\omega}_\alpha} = \frac{(1-\alpha)\sum_{g \in \mathcal{G}_S}\sqrt{\abs{g}}\norm{\bDelta_g}_2+\alpha\norm{\bDelta_S}_1}{\bar{\omega}_\alpha}\leq 2\norm{\bDelta}_2,
\label{equation:boundonPOmegaRatio}
\end{equation}
for any $0 \leq \alpha \leq 1$. 

Next, the result in Lemma \ref{Lemma:min_vs_restricted_eigen} follows by noting that $\bV^{(d)} = \diag\left(\bV_1,\ldots,\bV_N\right)$ is a block-diagonal matrix whose singular values correspond to those of its sub-blocks. Let $N_i$ denote the number of columns of $\bV_i$ and note that $N_i < N$ by construction. Then,
\begin{equation*}
\begin{aligned}
\min_{\bDelta \in \mathcal{C}_{N_c}(\mathcal{G},S)} \frac{\omega_\alpha\norm{\bV^{(d)}\bDelta}_2}{P_{\alpha,S}(\bDelta)} &\geq \min_{\bDelta \in \mathbb{R}^{N_c}} \frac{\norm{\bV^{(d)}\bDelta}_2}{2\norm{\bDelta}_2}
    = \min_{1 \leq i \leq N} \min_{\bDelta \in \mathbb{R}^{N_i}} \frac{\norm{\bV_i\bDelta}_2}{2\norm{\bDelta}_2} \\
    &\overset{(i)}{\geq} \min_{\bDelta \in \mathbb{R}^{2N}: \mathcal{M}(\bDelta) \leq N} \frac{\norm{\bV\bDelta}_2}{2\norm{\bDelta}_2} \overset{(ii)}{\geq} \frac{\phi_0}{2},
\end{aligned}
\end{equation*}
where (i) follows since $N_i < N$ for all $i = 1,\ldots,N$ and (ii) holds by Assumption \ref{assump:min_eigen_bound}.
\end{proof}

\begin{lemma}\label{Lemma:V_transfer}
Define the set $\mathcal{V}(x) := \left\{ \big\|\hat{\bV}_h - \bV\big\|_2 \leq x \right\}$. Then, under Assumption \ref{assump:min_eigen_bound}, it holds on $\mathcal{V}\left(\frac{\phi_0}{4}\right)$ that
\begin{equation*}
    \min_{\bx \in \mathcal{C}_{N_c}(\mathcal{G},S)} \frac{\bar\omega_\alpha\norm{\hat{\bV}_h^{(d)}\bx}_2}{P_{\alpha,S}(\bx)} \geq \frac{\phi_0}{4}.
\end{equation*}
\end{lemma}

\begin{proof}
First, recall the construction of $\hat{\bV}_h^{(d)} = \diag\left(\hat{\bV}_{1,h},\ldots,\hat{\bV}_{N,h}\right)$ with $\hat{\bV}_{i,h}$ containing at most $N-1$ columns of the matrix $\hat{\bV}_h = \left(\mathcal{B}_h\left(\hat{\bSigma}_1^\prime\right),\mathcal{B}_h\left(\hat{\bSigma}_0\right)\right)$. From the block-diagonal construction, it follows that
\begin{equation*}
    \norm{\hat{\bV}_h^{(d)}-\bV^{(d)}}_2 = \max_{1 \leq i \leq N} \norm{\hat{\bV}_{i,h}-\bV_i}_2 \leq \norm{\hat{\bV}_h - \bV}_2.
\end{equation*}
Then,
\begin{equation}\label{eq:Vhat_norm_bound}
    \norm{\hat{\bV}_h^{(d)}\bx}_2 \geq \norm{\bV^{(d)}\bx}_2 - \norm{\left(\hat{\bV}_h^{(d)} - \bV^{(d)}\right)\bx}_2 \geq \norm{\bV^{(d)}\bx}_2 - \norm{\hat{\bV} - \bV}_2 \norm{\bx}_2 \geq \norm{\bV^{(d)}\bx}_2 - \frac{\phi_0}{2}\norm{\bx}_2,
\end{equation}
where the last inequality follows holds on the set $\mathcal{V}\left(\frac{\phi_0}{4}\right)$. Consequently,
\begin{equation*}
    \min_{\bx \in \mathcal{C}_{N_c}(\mathcal{G},S)} \frac{\omega_\alpha\norm{\hat{\bV}_h^{(d)}\bx}_2}{P_{\alpha,S}(\bDelta)} \overset{(i)}{\geq} \min_{\bx \in \mathbb{R}^{N_c}} \frac{\norm{\hat{\bV}_h^{(d)}\bx}_2}{2\norm{\bx}_2} \overset{(ii)}{\geq} \min_{\bx \in \mathbb{R}^{N_c}} \frac{\norm{\bV^{(d)}\bx}_2}{2\norm{\bx}_2} - \frac{\phi_0}{4} \overset{(iii)}{\geq} \frac{\phi_0}{4},
\end{equation*}
where (i) follows from the proof of Lemma \ref{Lemma:min_vs_restricted_eigen}, (ii) holds by \eqref{eq:Vhat_norm_bound}, and (iii) from Lemma \ref{Lemma:min_vs_restricted_eigen}.
\end{proof}

\section{Proofs of Main Results}
\begin{proof}[\textbf{Proof of Theorem \ref{th:diag_approx}}]
We first prove various intermediate results, see (a)--(d) below. We afterwards combine these results and recover Theorem \ref{th:diag_approx}.
\begin{enumerate}[(a)]
 \item The matrix $\widetilde\bC_s=\left( \sum_{j=0}^{s-1} \bA^j \right) \bB =: \widetilde\bD_s \bB$ has a maximum bandwidth of $(s+1)(k_0-1)+1$ and satisfies
 $$
    \normoneinf{\widetilde{\bC}_s - \bC} \leq C_1 \delta_A^s.
 $$
  \item Define $\bSigma_0^{r,s} = \sum_{j=0}^r \widetilde{\bC}_s^j \widetilde\bD_s \bSigma_\epsilon \widetilde\bD_s\tran (\widetilde{\bC}_s^\prime)^j$ with $\widetilde \bC_s$ as in Theorem \ref{th:diag_approx}(a). The matrix $\bSigma_0^{r,s}$ is a banded matrix with bandwidth no larger than $2(rs+r+s)(k_0-1)+2 l_0 +1$. Moreover,
  $$
   \normoneinf{ \bSigma_0^{r,s} - \bSigma_0 }
   \leq C_2 \left( \frac{\delta_A^s}{\big[ 1 - ( C_1 \delta_A^s+ \delta_C)^2 \big]^3}+ \delta_C^{2(r+1)} \right),
  $$
  whenever $s$ is large enough such that $C_1 \delta_A^s+\delta_C<1$.
  \item Define $\bSigma_1^{r,s} = \widetilde \bC_s \bSigma_0^{r,s}$ with $\widetilde \bC_s$ and $\bSigma_0^{r,s}$ as in Theorems \ref{th:diag_approx}(a) and \ref{th:diag_approx}(b), respectively. The matrix $\bSigma_1^{r,s}$ is a banded matrix with bandwidth no larger than $(2rs+2r+3s+1)(k_0-1)+2 l_0+1$. Moreover,
  $$
   \normoneinf{ \bSigma_1^{r,s} - \bSigma_1 }
   \leq C_3 \left( \frac{\delta_A^s}{\big[ 1 - ( C_1 \delta_A^s+ \delta_C)^2 \big]^3}+ \delta_C^{2(r+1)} \right),
  $$
  whenever $s$ is large enough such that $C_1 \delta_A^s+\delta_C<1$.
  \item Take any $h_1\geq 2(rs+r+s)(k_0-1)+2 l_0 +1$ , then
  $$
   \big\|\hbanded{\widehat\bSigma_0}{1} - \bSigma_0 \big\|_2 \leq \epsilon + 2 C_2 \left( \frac{\delta_A^s}{\big[ 1 - ( C_1 \delta_A^s+ \delta_C)^2 \big]^3}+ \delta_C^{2(r+1)} \right),
  $$
  with a probability of at least $1-\mathcal{P}_1(\epsilon,N,T)$ (for polynomial tail decay) or $1-\mathcal{P}_2(\epsilon,N,T)$ (for exponential tail decay).
  \item Take any $h_2 \geq (2rs+2r+3s+1)(k_0-1)+2 l_0+1$ , then
  $$
  \big\|\hbanded{\widehat\bSigma_1}{2} - \bSigma_1 \big\|_2 \leq \epsilon + 2 C_3  \left( \frac{\delta_A^s}{\big[ 1 - ( C_1 \delta_A^s+ \delta_C)^2 \big]^3}+ \delta_C^{2(r+1)} \right),
  $$
  with a probability of at least $1-\mathcal{P}_1(\epsilon,N,T)$ (for polynomial tail decay) or $1-\mathcal{P}_2(\epsilon,N,T)$ (for exponential tail decay).
\end{enumerate}
Explicit expressions for $C_1$, $C_2$, $C_3$, and $0\leq\delta_C<1$ are provided in the proofs below.

\bigskip
 \textbf{(a)} The proof builds upon results from \cite{guowangyao2016} on banded vector autoregressions. Recall that $\bC=\bD \bB$ with $\bD = (\bI_N-\bA)^{-1}$. $\widetilde\bD_s = \sum_{j=0}^{s-1} \bA^j$ has a bandwidth of at most $s(k_0-1)+1$ and satisfies\footnote{If matrices $\bF_1$ and $\bF_2$ are banded matrices with bandwidths $k_1$ and $k_2$, respectively, then the product $\bF_1 \bF_2$ is again a banded matrix with a bandwidth of at most $k_1+k_2-1$.\label{ftnt:productbanded}}
 $$
  \normoneinf{ \widetilde\bD_s - \bD } = \normoneinf{\sum_{j=0}^{s-1} \bA^j - (\bI_N - \bA)^{-1} } = \normoneinf{ - \sum_{j=s}^\infty \bA^j} \leq \sum_{j=s}^\infty \normoneinf{\bA}^j \leq \frac{\delta_A^s}{1-\delta_A}.
 $$
 The product $\widetilde{\bC}_s = \widetilde{\bD}_s \bB$ has a maximal bandwidth of $(s+1)(k_0-1)+1$. Since $\normoneinf{\bB} \leq C_B$ (Assumption \ref{assump:bandedness}(c)), we also have
$$
 \normoneinf{ \widetilde{\bC}_s - \bC  }
  = \normoneinf{ \left( \widetilde\bD_s - \bD \right) \bB }
  \leq \normoneinf{ \widetilde\bD_s - \bD } \normoneinf{ \bB }
  \leq C_B \frac{\delta_A^s}{1-\delta_A}.
$$
 and the claim follows with $C_1 = \frac{C_B}{1-\delta_A}$. \textbf{(b)} Iterating on the observation in footnote \ref{ftnt:productbanded}, we conclude that the bandwidth of $\widetilde \bC_s^r$ is at most $r\Big[(s+1)(k_0-1)+1\Big]-(r-1)= r(s+1)(k_0-1)+1$. The bandwidth of $\bSigma_0^{r,s}$ therefore does not exceed
$$
 2\Big[r(s+1)(k_0-1)+1\Big] + 2(s(k_0-1)+1) +(2l_0+1) - 4 = 2(rs+r+s)(k_0-1)+2 l_0 +1.
$$
 We now bound $\norm{ \bSigma_0^{r,s} - \bSigma_0 }_1$. Assumption \ref{assump:stability}(b), imposes $\normoneinf{\bC}\leq \delta_C$ for some $0\leq\delta_C<1$. Because $\bSigma_0 = \sum_{j=0}^\infty \bC^j \bD \bSigma_\epsilon \bD\tran(\bC^\prime)^j$, it holds that
\begin{equation}
\begin{aligned}
 &\normoneinf{\bSigma_0^{r,s} - \bSigma_0}
  = \normoneinf{ \sum_{j=0}^r \left[ \widetilde{\bC}_s^j \widetilde\bD_s \bSigma_\epsilon \widetilde \bD_s\tran (\widetilde{\bC}_s^\prime)^j - \bC^j \bD \bSigma_\epsilon \bD\tran (\bC^\prime)^j  \right] - \sum_{j=r+1}^\infty\bC^j \bD \bSigma_\epsilon \bD\tran (\bC^\prime)^j} \\
  &\leq
  \sum_{j=0}^r \normoneinf{\bSigma_\epsilon} \normoneinf{\widetilde \bC_s^j}^2 \normoneinf{\widetilde \bD_s - \bD }^2
  + 2 \sum_{j=0}^r \normoneinf{\bSigma_\epsilon} \normoneinf{\widetilde \bC_s^j}^2 \normoneinf{\bD} \normoneinf{\widetilde \bD_s - \bD}^2  \\
  &\qquad + \sum_{j=1}^r \normoneinf{\bSigma_\epsilon} \normoneinf{\widetilde \bC_s^j - \bC^j}^2 \normoneinf{\bD}^2 + 2 \sum_{j=1}^r \normoneinf{\bSigma_\epsilon} \normoneinf{\widetilde \bC_s^j - \bC^j} \normoneinf{\bC^j} \normoneinf{\bD}^2 \\
  &\qquad + \normoneinf{\sum_{j=r+1}^\infty\bC^j \bD \bSigma_\epsilon \bD\tran (\bC^\prime)^j}.
\end{aligned}
\label{eq:sigma0decomp}
\end{equation}
 An inspection of \eqref{eq:sigma0decomp} shows that additional upper bounds are required on $\normoneinf{\widetilde{\bC}_s^j}$, $\normoneinf{\bC^j-\widetilde{\bC}_s^j}$, and $\normoneinf{\sum_{j=r+1}^\infty\bC^j \bD \bSigma_\epsilon \bD\tran (\bC^\prime)^j}$. For $\normoneinf{\widetilde{\bC}_s^j}$, using properties of matrix norms and Theorem \ref{th:diag_approx}(a), we obtain
 \begin{equation}
\normoneinf{\widetilde{\bC}_s^j}
 = \normoneinf{\left(\widetilde \bC_s - \bC + \bC \right)^j}
 \leq \left( \normoneinf{\widetilde\bC_s- \bC} + \normoneinf{\bC} \right)^j
 \leq \Big[ C_1 \delta_A^s+ \delta_C\Big]^j.
\label{eq:mSigmaterm1}
\end{equation}
Furthermore, expanding the matrix powers provides
\begin{equation}
\begin{aligned}
 &\normoneinf{\bC^j-\widetilde{\bC}_s^j}
  = \normoneinf{ \bC^j -\left( \widetilde{\bC}_s -\bC + \bC \right)^j} \\
  &=  \normoneinf{ \bC^j -\left((\widetilde{\bC}_s -\bC)^j + (\widetilde{\bC}_s -\bC)^{j-1}\bC + (\widetilde{\bC}_s -\bC)^{j-2}\bC (\widetilde{\bC}_s -\bC)  +\ldots + \bC (\widetilde{\bC}_s -\bC)^{j-1} \ldots + \bC^j \right)} \\
  &\leq \sum_{k=1}^j\binom{j}{k}\normoneinf{\bC}^{j-k}\normoneinf{\widetilde\bC_s - \bC}^k
  = \sum_{k=0}^{j-1} \binom{j}{k+1} \normoneinf{\bC}^{(j-1)-k} \normoneinf{\widetilde\bC_s - \bC}^{k+1} \\
  &= \normoneinf{\widetilde\bC_s - \bC} \sum_{k=0}^{j-1} \frac{j}{k+1} \binom{j-1}{k} \normoneinf{\bC}^{(j-1)-k} \normoneinf{\widetilde\bC_s - \bC}^{k}
  \leq C_1 j \delta_A^s \left[ \normoneinf{\widetilde\bC_s- \bC} + \normoneinf{\bC} \right]^{j-1} \\
  &\leq C_1 j \delta_A^s \Big[ C_1 \delta_A^s+ \delta_C\Big]^{j-1}
\end{aligned}
\label{eq:mSigmaterm2}
\end{equation}
Finally, since $\normoneinf{\bSigma_\epsilon} \leq C_\epsilon$ and $\normoneinf{\bD}= \normoneinf{\sum_{j=0}^\infty \bA^j}=  \sum_{j=0}^\infty \normoneinf{\bA}^j\leq c_D$ (Assumption \ref{assump:stability}(a)), we have
\begin{equation}
    \normoneinf{\sum_{j=r+1}^\infty\bC^j \bD \bSigma_\epsilon \bD\tran (\bC^\prime)^j}
    \leq \sum_{j=r+1}^\infty \normoneinf{\bSigma_\epsilon} \normoneinf{\bD}^2 \normoneinf{\bC}^{2j}
    \leq C_\epsilon c_D^2 \sum_{j=r+1}^\infty\delta_C^{2j}
    = C_\epsilon c_D^2 \frac{\delta_C^{2(r+1)}}{1-\delta_C^2}.
\label{eq:mSigmaterm3}
\end{equation}
Returning to \eqref{eq:sigma0decomp} and inserting the upper bounds in \eqref{eq:mSigmaterm1}--\eqref{eq:mSigmaterm3}, we find
$$
\begin{aligned}
 &\normoneinf{\bSigma_0^{r,s} - \bSigma_0}
  \leq
  \frac{C_\epsilon}{(1-\delta_A)^2}  \delta_A^{2s} \sum_{j=0}^r \left( C_1 \delta_A^s + \delta_C\right)^{2j}
  + 2 \frac{C_\epsilon c_D}{(1-\delta_A)^2} \delta_A^{2s}  \sum_{j=0}^r \left( C_1 \delta_A^s + \delta_C\right)^{2j} \\
  &+ C_\epsilon C_1^2 c_D^2 \delta_A^{2s} \sum_{j=1}^r j^2 \left( C_1 \delta_A^s+ \delta_C\right)^{2(j-1)}
  + 2 C_\epsilon C_1 c_D^2 \delta_A^s \sum_{j=1}^r j \delta_C^j \left( C_1 \delta_A^s+ \delta_C\right)^{j-1}
  + C_\epsilon c_D^2 \frac{\delta_C^{2(r+1)}}{1-\delta_C^2}.
\end{aligned}
$$
Assuming $s$ is sufficient large, that is assuming $C_1 \delta_A^s+\delta_C<1$, we subsequently use result on geometric series and conclude\footnote{Specifically, $\sum_{j=1}^\infty j z^j = \frac{z}{(1-z)^2}$ and $\sum_{j=1}^\infty j^2 z^j = \frac{z}{(1-z)^3}$ for $|z|<1$.}
$$
\begin{aligned}
 &\normoneinf{\bSigma_0^{r,s}- \bSigma_0}
 \leq \frac{C_\epsilon}{(1-\delta_A)^2} \delta_A^{2s} \frac{1}{1-( C_1 \delta_A^s + \delta_C)^2}
 + 2 \frac{C_\epsilon c_D}{(1-\delta_A)^2} \delta_A^{2s} \frac{1}{1-( C_1 \delta_A^s + \delta_C)^2} \\
 &\qquad + \frac{C_\epsilon C_1^2 c_D^2 \delta_A^{2s}}{( C_1 \delta_A^s+ \delta_C)^2} \frac{( C_1 \delta_A^s+ \delta_C)^2}{\big[ 1 - ( C_1 \delta_A^s+ \delta_C)^2 \big]^3}
 + 2 \frac{C_\epsilon C_1 c_D^2 \delta_A^s}{C_1 \delta_A^s+ \delta_C} \frac{\delta_C(C_1 \delta_A^s +\delta_C)}{\big[1 - \delta_C(C_1 \delta_A^s +\delta_C) \big]^2}
 +C_\epsilon c_D^2 \frac{\delta_C^{2(r+1)}}{1-\delta_C^2} \\
 &=  \frac{C_\epsilon}{(1-\delta_A)^2}\frac{1}{1-( C_1 \delta_A^s + \delta_C)^2} \delta_A^{2s}
 + 2 \frac{C_\epsilon c_D}{(1-\delta_A)^2}\frac{1}{1-( C_1 \delta_A^s + \delta_C)^2} \delta_A^{2s}
 + \frac{C_\epsilon C_1^2 c_D^2}{\big[ 1 - ( C_1 \delta_A^s+ \delta_C)^2 \big]^3} \delta_A^{2s} \\
 &\qquad+ 2\frac{ C_\epsilon C_1 c_D^2 \delta_C}{\big[1 - \delta_C(C_1 \delta_A^s +\delta_C) \big]^2} \delta_A^s
 + C_\epsilon c_D^2 \frac{\delta_C^{2(r+1)}}{1-\delta_C^2} \\
 &\leq \frac{1}{(1-\delta_A)^2} \frac{  C_\epsilon + 2 C_\epsilon  c_D }{1-( C_1 \delta_A^s + \delta_C)^2} \delta_A^{2s} +\frac{  C_\epsilon C_1^2 c_D^2 + 2 C_\epsilon C_1 c_D^2 \delta_C }{\big[ 1 - ( C_1 \delta_A^s+ \delta_C)^2 \big]^3} \delta_A^s+ \frac{C_\epsilon c_D^2}{1-\delta_C^2} \delta_C^{2(r+1)} \\
 &\leq \left[ \frac{C_\epsilon + 2 C_\epsilon c_D}{(1-\delta_A)^2} +  C_\epsilon C_1^2 c_D^2 + 2 C_\epsilon C_1 c_D^2 \delta_C \right]\frac{\delta_A^s}{\big[ 1 - ( C_1 \delta_A^s+ \delta_C)^2 \big]^3} + \frac{C_\epsilon c_D^2}{1-\delta_C^2} \delta_C^{2(r+1)} .
\end{aligned}
$$
The claim is thus indeed valid with $C_2= \max\left\{ \frac{C_\epsilon + 2 C_\epsilon c_D}{(1-\delta_A)^2} +  C_\epsilon C_1^2 c_D^2 + 2 C_\epsilon C_1 c_D^2 \delta_C , \frac{C_\epsilon c_D^2}{1-\delta_C^2} \right\}$. \textbf{(c)} We have $\bSigma_1 = (\bI_N - \bA)^{-1} \bB \bSigma_0 = \bC \bSigma_0$, and hence
\begin{equation}
\normoneinf{ \bSigma_1^{r,s} - \bSigma_1 }
 = \normoneinf{ \widetilde \bC_s (\bSigma_0^{r,s}-\bSigma_0) + (\widetilde \bC_s -\bC) \bSigma_0}
 \leq \normoneinf{\widetilde \bC_s} \normoneinf{\bSigma_0^{r,s}-\bSigma_0} + \normoneinf{\widetilde \bC_s -\bC} \normoneinf{\bSigma_0}
\label{eq:boundingmatrixnorms}
\end{equation}
Now combine $\normoneinf{\widetilde \bC_s} \leq C_1 \delta_\kappa^s+\delta_C \leq 1$ (since $s$ is taken sufficiently large), Theorem \ref{th:diag_approx}(b), \eqref{eq:mSigmaterm2} with $j=1$, and $\normoneinf{\bSigma_0}\leq \sum_{j=0}^\infty \normoneinf{ \bC^j \bD \bSigma_\epsilon \bD\tran (\bC\tran)^j} \leq \frac{C_\epsilon c_D^2}{1-\delta_C^2}$ to obtain the stated result with $C_3= C_2 + \frac{C_1 C_\epsilon c_D^2}{1-\delta_C^2}$. \textbf{(d)} We have
\begin{equation}
\begin{aligned}
    \normoneinf{\hbanded{\widehat\bSigma_0}{1} - \bSigma_0 }
    &\leq \normoneinf{\hbanded{\widehat\bSigma_0}{1} - \hbanded{\bSigma_0}{1}}
        + \normoneinf{\hbanded{\bSigma_0}{1} - \hbanded{\bSigma_0^{r,s}}{1}}
        + \normoneinf{\hbanded{\bSigma_0^{r,s}}{1} - \bSigma_0} \\
    &= \normoneinf{\hbanded{\widehat\bSigma_0-\bSigma_0}{1}}
        + \normoneinf{\hbanded{\bSigma_0 - \bSigma_0^{r,s}}{1}}
        + \normoneinf{\bSigma_0^{r,s} - \bSigma_0} \\
    &\leq \normoneinf{\hbanded{\widehat\bSigma_0-\bSigma_0}{1}} + 2 \normoneinf{\bSigma_0^{r,s} - \bSigma_0}
\end{aligned}
\label{eq:bandedSigma0decomp}
\end{equation}
because $\bSigma_0^{r,s}$ is a banded matrix already and banding can only decrease the norm difference between $\bSigma_0^{r,s}$ and $\bSigma_0$. We consider the three terms in the RHS of \eqref{eq:bandedSigma0decomp} separately. There are at most $2h_1+1$ nonzero elements in any column/row of $\hbanded{\widehat\bSigma_0-\bSigma_0}{1}$ and thus
\begin{equation}
 \normoneinf{\hbanded{\widehat\bSigma_0-\bSigma_0}{1}}
  \leq (2 h_1 + 1) \norm{\widehat\bSigma_0-\bSigma_0 }_{max}
\label{eq:myerrorboundforlater}
\end{equation}
and
\begin{equation}
\begin{aligned}
 &\Prob\left(\normoneinf{\hbanded{\widehat\bSigma_0-\bSigma_0}{1}} \leq x\right) \geq \Prob\left( (2h_1+1) \norm{\widehat\bSigma_0-\bSigma_0 }_{max} \leq x\right)
   = 1 - \Prob\left( \norm{\widehat\bSigma_0-\bSigma_0 }_{max} > \frac{x}{2h_1+1} \right)\\
  &\geq 1- \Prob\left(\bigcup_{1\leq i,j \leq N} \left|\left[\widehat\bSigma_0-\bSigma_0\right]_{ij} \right|  > \frac{x}{2h_1+1}\right)
  \geq 1 - \sum_{i=1}^N \sum_{j=1}^N \Prob\left(\left| \sum_{t=2}^T \xi_{it} \xi_{jt} - \E(\xi_{it} \xi_{jt}) \right| > \frac{T x}{2h_1+1} \right) \\
  &
  \geq
  \begin{cases}
    1 - N^2\left[ \left(b_1 T^{(1-\delta)/3}+\frac{(2h_1+1)b_3}{x}\right) \exp\left(-\frac{T^{(1-\delta)/3}}{2b_1^2}\right) + \frac{b_2 (2h_1+1)^d}{x^d T^{\frac{\delta}{2}(d-1)}}\right] &\qquad (\text{polynomial tails}), \\
    1 - N^2 \left[\frac{\kappa_1(2h_1+1)}{x} + \frac{2}{\kappa_2}\left(\frac{T x^2}{(2h_1+1)^2}\right)^{1/7} \right] \exp\left(- \frac{1}{\kappa_3}\left(\frac{T x^2}{(2h_1+1)^2}\right)^{1/7} \right) &\qquad (\text{exponential tails}),
  \end{cases}
\end{aligned}
\label{eq:probstatements}
\end{equation}
where the last inequality exploits Lemma \ref{lemma:covmatrixelements} (see Supplement). Note that the probabilities in \eqref{eq:probstatements} coincide with the probabilities defined as $1-\mathcal{P}_1(x,N,T)$ and $1-\mathcal{P}_2(x,N,T)$ in Theorem \ref{th:diag_approx}. Overall, if $\normoneinf{\hbanded{\widehat\bSigma_0-\bSigma_0}{1}} \leq \epsilon$ holds, then continuing from \eqref{eq:bandedSigma0decomp},
$$
\begin{aligned}
 \normoneinf{\hbanded{\widehat\bSigma_0}{1} - \bSigma_0 }
 \leq \epsilon + 2 \normoneinf{\bSigma_0^{r,s} - \bSigma_0}
 \leq \epsilon + 2 C_2 \left( \frac{\delta_A^s}{\big[ 1 - ( C_1 \delta_A^s+ \delta_C)^2 \big]^3}+ \delta_C^{2(r+1)} \right),
\end{aligned}
$$
where the last inequality follows from part (b) of this proof. A lower bound on the probability of $\left\{\normoneinf{\hbanded{\widehat\bSigma_0-\bSigma_0}{1}} \leq \epsilon\right\}$ occurring is immediately available from \eqref{eq:probstatements}. 

\textbf{(e)} Mimicking the steps from (d), we find
$$
 \normoneinf{\hbanded{\widehat\bSigma_1}{2} - \bSigma_1 }
  \leq \epsilon + 2 \normoneinf{ \bSigma_1^{r,s} - \bSigma_1 }
  \leq \epsilon + 2 C_3 \left( \frac{\delta_A^s}{\big[ 1 - ( C_1 \delta_A^s+ \delta_C)^2 \big]^3}+ \delta_C^{2(r+1)} \right)
$$
if we use part (c) and if $\Big\{\normoneinf{\hbanded{\widehat\bSigma_1-\bSigma_1}{2}} \leq \epsilon\Big\}$ holds. The probability of the latter event is $1-\mathcal{P}_1(\epsilon,N,T)$ (polynomial tails) or $1-\mathcal{P}_2(\epsilon,N,T)$ (exponential tails).

\bigskip
We now combine all these intermediate results to recover the result from the theorem. Parts (d) and (e) are both applicable since
$$
\begin{aligned}
h&=\max\{ h_1,h_2\}=h_2= (2rs +2r +3s +1)(k_0-1)+2 l_0+1 \\
 &\leq \Big(2r(s+1) +3(s+1) \Big)(k_0-1)+2l_0 +1\leq (s+1)(2r+3) (k_0-1)+2l_0 +1.
\end{aligned}
$$
Lemma \ref{lemma:L2normineqs} implies that
$$
\begin{aligned}
 \normoneinf{\widehat{\bV}_h - \bV }
  &\leq \normoneinf{\hbanded{\widehat\bSigma_1}{2} - \bSigma_1 } +\normoneinf{\hbanded{\widehat\bSigma_0}{1} - \bSigma_0 } \\
  & \leq 2 \epsilon + 2 C_4 \left( \frac{\delta_A^s}{\big[ 1 - ( C_1 \delta_A^s+ \delta_C)^2 \big]^3}+ \delta_C^{2(r+1)} \right),
\end{aligned}
$$
when $\left\{\normoneinf{\hbanded{\widehat\bSigma_1-\bSigma_1}{2}} \leq \epsilon\right\}\cap\left\{\normoneinf{\hbanded{\widehat\bSigma_0-\bSigma_0}{1}} \leq \epsilon\right\}$ takes place and $C_4=C_2+C_3$. From $\Prob(A\cap B)= 1 - \Prob(A^c\cup B^c)\geq 1 - \Prob(A^c)-\Prob(B^c)$ it follows that this joint event takes place with probabilities $1-2\mathcal{P}_1(\epsilon,N,T)$ and $1-2\mathcal{P}_2(\epsilon,N,T)$ in the cases of polynomial and exponential tail decay, respectively.

It remains to determine $s$ and $r$ such that $C_4 \delta_A^s/\big[ 1 - ( C_1 \delta_A^s+ \delta_C)^2 \big]^3\leq \epsilon$ and $C_4 \delta_C^{2(r+1)}\leq \epsilon$ hold. For any $s\geq s^*$ with $s^*= \log\left( \frac{1}{C_1} \big[ (1 - (1-\delta_C)^{1/3})^{1/2} - \delta_C\big] \right)/\log(\delta_A)$, we have $1/[1-(C_1 \delta_A^s + \delta_C)^2]^3\leq 1/[1-\delta_C]$. Under the latter assumption, we determine the $s$ such that
$$
\begin{aligned}
 \frac{C_4 \delta_A^s}{1-\delta_C} \leq \epsilon
 \iff
 s \log(\delta_A) \leq \log\left( \frac{(1-\delta_C)\epsilon}{C_4} \right)
 \iff
 s \geq \frac{\log\left(C_4/(1-\delta_C)\epsilon \right)}{|\log(\delta_A)|}.
\end{aligned}
$$
Similarly, for $r$, the choice $r= \log(C_4/\epsilon)/2|\log(\delta_C)|$ suffices.
\end{proof}

\begin{proof}[\textbf{Proof of Theorem \ref{Thm:sgl}}]
The proof of the theorem relies on the properties of dual norms. Recall $P_\alpha(\bc) = \alpha\sum_{g \in \setG} \sqrt{|g|} \norm{\bc_g}_2 + (1-\alpha)\norm{\bc}_1$. Exploiting the properties of $\norm{\cdot}_1$ and $\norm{\cdot}_2$, it is straightforward to verify that $P_\alpha(\cdot)$ is a norm for any $0\leq\alpha\leq 1$. For any norm $\norm{\cdot}$, we define its dual norm $\norm{\cdot}^*$ through $\norm{\bc}^* = \sup_{\bx\neq \bzeros} \frac{\abs{\bc^\prime\bx}}{\norm{\bx}}$. The dual-norm inequality states that
\begin{equation}
 \bc\tran\bx \leq \norm{\bc}^*\,\norm{\bx}
 \qquad\qquad \text{for all conformable vectors $\bc$ and $\bx$}.
\label{eq:dualnormineq}
\end{equation}
For the norm $P_\alpha(\bc)$, its dual norm $P_\alpha^*(\bc)$ is bounded by
\begin{equation}
\begin{aligned}
 P_\alpha^*(\bc)
 &= \sup_{\bx\neq \bzeros} \frac{|\bc\tran\bx|}{P_\alpha(\bx)}
 = \sup_{\bx\neq \bzeros} \frac{|\bc\tran\bx|}{\alpha\sum_{g \in \setG}\sqrt{|g|} \norm{\bx_g}_2 + (1-\alpha)\norm{\bx}_1} \\
 &\stackrel{(i)}{\leq} \alpha \sup_{\bx\neq \bzeros} \frac{|\bc\tran\bx|}{\sum_{g \in \setG}\sqrt{|g|} \norm{\bx_g}_2}+ (1-\alpha) \sup_{\bx\neq \bzeros} \frac{|\bc\tran\bx|}{\norm{\bx}_1}
 \stackrel{(ii)}{\leq} \alpha \max_{g \in \setG} \frac{\norm{\bc_g}_2}{\sqrt{\abs{g}}} + (1-\alpha)\norm{\bc}_\infty,
 \end{aligned}
\label{eq:Pstarupperbound}
\end{equation}
by convexity of the function $f(x)=x^{-1}$ in step (i), and using for step (ii) both $\norm{\bc}_1^*=\norm{\bc}_\infty$ and
\begin{equation*}
\begin{aligned}
    \sup_{\bx\neq \bzeros }\frac{\abs{\bc\tran\bx}}{\sum_{g \in \setG}\sqrt{\abs{g}}\norm{\bx_g}_2}
    &= \sup_{\bx\neq\bzeros , \sum_{g \in \setG}\sqrt{\abs{g}}\norm{\bx_g}_2 = 1} \abs{\sum_{g \in \setG}\bc_g\tran\bx_g}\\
    &\leq \sup_{\bx\neq \bzeros , \sum_{g \in \setG}\sqrt{\abs{g}}\norm{\bx_g}_2 = 1} \sum_{g \in \setG}\norm{\frac{\bc_g}{\sqrt{\abs{g}}}}_2\norm{\sqrt{\abs{g}}\bx_g}_2
    \leq \max_{g \in \setG}\frac{\norm{\bc_g}_2}{\sqrt{\abs{g}}}.
\end{aligned}
\end{equation*}

We now start the actual proof. Recall $\hat{\bsigma}_h = \vect\left( \hbanded{\hat\bSigma_1}{}\tran \right)$, $\hat{\bV}_h = \begin{bmatrix} \hbanded{\hat{\bSigma}_1}{}\tran & \hbanded{\hat{\bSigma}_0}{}\end{bmatrix}$, $\bV= \begin{bmatrix} \bSigma_1\tran & \bSigma_0 \end{bmatrix}$, and $\bC = \begin{bmatrix} \bA & \bB \end{bmatrix}$. Exploiting standard properties of $\vect(\cdot)$, we find
\begin{equation}
\begin{aligned}
\hat{\bsigma}_h-\hat{\bV}^{(d)}_h\bc
    &= \vect\left(\hbanded{\hat\bSigma_1}{}\tran- \hat{\bV}_h\bC\tran\right) = \vect\left(\left[\hbanded{\hat{\bSigma}_1}{} - \bSigma_1\right]\tran - \left[\hat{\bV}_h-\bV\right]\bC\tran + \underbrace{\left[\bSigma_1\tran - \bV\bC\tran\right]}_{=\bZeros,\text{ see }\eqref{eq:YWpopulation}}\right)\\
    &= \underbrace{\vect\left( \hbanded{\hat{\bSigma}_1}{}\tran - \bSigma_1\tran\right)}_{\hat{\bDelta}_\Sigma} - \underbrace{\left[\hat{\bV}_h^{(d)} - \bV^{(d)}\right]}_{\hat{\bDelta}_V}\bc =  \hat{\bDelta}_\Sigma - \hat{\bDelta}_V\bc.
\end{aligned}
\label{eq:basic_deviation}
\end{equation}
Using \eqref{eq:basic_deviation}, we rewrite
\begin{equation*}
\begin{aligned}
    \norm{\hat{\bsigma}_h - \hat{\bV}_h^{(d)}\hat{\bc}}_2^2
&= \norm{ \Big[\hat{\bsigma}_h - \hat{\bV}_h^{(d)}\bc\Big] - \Big[\hat{\bV}_h^{(d)}(\hat{\bc} - \bc)\Big]}_2^2 \\
&= \norm{\hat{\bsigma}_h - \hat{\bV}_h^{(d)}\bc}_2^2 + \norm{\hat{\bV}_h^{(d)}(\hat{\bc}-\bc)}_2^2  - 2\left(\hat{\bc}-\bc\right)^\prime\hat{\bV}_h^{(d)\prime}\left(\hat{\bDelta}_\Sigma - \hat{\bDelta}_V\bc\right).
\end{aligned}
\end{equation*}
Recalling the objective function $\mathcal{L}_\alpha(\bc;\lambda)=\norm{\hat{\bsigma}_h - \hat{\bV}_h^{(d)}\bc}_2^2 + \lambda P_\alpha(\bc)$ and noting that $\mathcal{L}_\alpha(\hat\bc;\lambda) \leq \mathcal{L}_\alpha(\bc;\lambda)$ by construction, it follows that 
\begin{equation}
\begin{aligned}
\norm{\hat{\bV}_h^{(d)}(\hat{\bc}-\bc)}_2^2 + \lambda P_\alpha\left(\hat{\bc}\right)
    &\leq 2\left(\hat{\bc}-\bc\right)^\prime\hat{\bV}_h^{(d)\prime}\left(\hat{\bDelta}_\Sigma - \hat{\bDelta}_V\bc\right) + \lambda P_\alpha\left(\bc\right)\\
    &\leq 2P_\alpha\left(\hat{\bc}-\bc\right)P_\alpha^*\left(\hat{\bV}_h^{(d)\prime}\left(\hat{\bDelta}_\Sigma - \hat{\bDelta}_V\bc\right)\right) + \lambda P_\alpha\left(\bc\right)\\
    &\leq  P_\alpha\left(\hat{\bc}-\bc\right)\left[2P_\alpha^*\left(\hat{\bV}_h^{(d)\prime}\hat{\bDelta}_\Sigma\right) + 2P_\alpha^*\left(\hat{\bV}_h^{(d)\prime}\hat{\bDelta}_V\bc\right)\right] + \lambda P_\alpha\left(\bc\right),
\end{aligned}
\label{eq:SGL_basic_ineq}
\end{equation}
where we used the dual-norm inequality (see \eqref{eq:dualnormineq}) and the triangle property of (dual) norms in the second and third inequality, respectively. Define the sets
\begin{equation}
 \mathcal{H}_1(x) = \left\{ 2 P_\alpha^*\left( \hat{\bV}_h^{(d)\prime}\hat{\bDelta}_\Sigma\right) \leq x \right\}
 \qquad
 \text{ and }
 \qquad
 \mathcal{H}_2(x) = \left\{ 2 P_\alpha^*\left(\hat{\bV}_h^{(d)\prime}\hat{\bDelta}_V\bc\right) \leq x \right\}.
\end{equation}
On the set $\mathcal{H}_1(\frac{\lambda}{4})\cap \mathcal{H}_2(\frac{\lambda}{4})$, we can scale \eqref{eq:SGL_basic_ineq} by a factor 2 to obtain
\begin{equation}
 2\norm{\hat{\bV}_h^{(d)}(\hat{\bc}-\bc)}_2^2 + 2\lambda P_\alpha\left(\hat{\bc}\right)
 \leq \lambda P_\alpha(\hat{\bc}-\bc) + 2 \lambda P_\alpha (\bc).
 \label{eq:SGL_basic_ineq2}
\end{equation}
We subsequently manipulate $P_\alpha(\hat{\bc})$ and $P_\alpha\left(\hat{\bc}-\bc\right)$. Using the reverse triangle inequality, we have
\begin{equation}
\begin{aligned}
&P_\alpha(\hat{\bc})
    = \alpha\sum_{g \in \mathcal{G}} \sqrt{\abs{g}} \norm{\hat{\bc}_g}_2 + (1-\alpha)\norm{\hat{\bc}}_1\\
    &\geq \alpha\sum_{g \in \mathcal{G}_S} \sqrt{\abs{g}} \Big[\norm{\bc_g}_2 - \norm{\hat{\bc}_g - \bc_g}_2 \Big] + \alpha\sum_{g \in \mathcal{G}_{S^c}} \sqrt{\abs{g}} \norm{\hat{\bc}_g}_2
    + (1-\alpha)\Big[\norm{\bc_S}_1 - \norm{\hat{\bc}_S - \bc_S}_1 + \norm{\hat{\bc}_{S^c}}_1\Big] \\
    &= P_{\alpha,S}(\bc) + P_{\alpha,S^c}\left(\hat{\bc} - \bc\right) - P_{\alpha,S}\left(\hat{\bc} - \bc\right),
\end{aligned}
\label{eq:P_rewrite_1}
\end{equation}
where $\mathcal{G}_S$ and $\mathcal{G}_{S^c}$ are defined in Lemma \ref{Lemma:min_vs_restricted_eigen}. Simple rewriting provides
\begin{equation}
\begin{aligned}
P_\alpha\left(\hat{\bc}-\bc\right)
    &= \alpha\sum_{g \in \mathcal{G}} \sqrt{\abs{g}} \norm{\hat{\bc}_g - \bc_g}_2 +
        (1-\alpha)\norm{\hat{\bc}_S - \bc_S}_1  + \alpha \sum_{g \in \mathcal{G}_{S^c}} \sqrt{\abs{g}} \norm{\hat{\bc}_g}_2 + (1-\alpha)\norm{\hat{\bc}_{S^c}}_1 \\
    &= P_{\alpha,S}\left(\hat{\bc}-\bc\right) + P_{\alpha,S^c}\left(\hat{\bc}-\bc\right).
\end{aligned}
\label{eq:P_rewrite_2}
\end{equation}
Combining results \eqref{eq:SGL_basic_ineq2}--\eqref{eq:P_rewrite_2} yields
\begin{equation}
2\norm{\hat{\bV}_h^{(d)}(\hat{\bc}-\bc)}_2^2 + \lambda P_{\alpha,S^c}\left(\hat{\bc}-\bc\right) \leq 3\lambda P_{\alpha,S}(\hat{\bc} - \bc),
\label{eq:RECconsequence}
\end{equation}
and $\hat\bc-\bc$ is thus a member of the set $\mathcal{C}_{N_c}(\mathcal{G},S)$ as defined in Lemma \ref{Lemma:min_vs_restricted_eigen}. 

In combination with Lemma \ref{Lemma:V_transfer}, thus requiring $\mathcal{H}_1(\frac{\lambda}{4})\cap \mathcal{H}_2(\frac{\lambda}{4})\cap\mathcal{V}\left(\frac{\phi_0}{2}\right)$ to hold, we conclude
\begin{equation}
\begin{aligned}
    2\norm{\hat{\bV}_h^{(d)}(\hat{\bc}-\bc)}_2^2 &+ \lambda P_\alpha(\hat{\bc}-\bc)
    = 2\norm{\hat{\bV}_h^{(d)}(\hat{\bc}-\bc)}_2^2 + \lambda P_{\alpha,S}(\hat{\bc}-\bc) + \lambda P_{\alpha,S^c}(\hat{\bc}-\bc)\\
    &\overset{(i)}{\leq} 4\lambda  P_{\alpha,S}(\hat{\bc}-\bc) \overset{(ii)}{\leq} 16 \norm{\hat{\bV}_h^{(d)}(\hat{\bc}-\bc)}_2 \left(\frac{\bar{\omega}_\alpha \lambda}{\phi_0}\right) \overset{(iii)}{\leq} \norm{\hat{\bV}_h^{(d)}(\hat{\bc}-\bc)}_2^2 + \frac{64\bar{\omega}_\alpha^2\lambda^2}{\phi_0^2},
\end{aligned}
\label{eq:finalineqSGL}
\end{equation}
where step (i) follows from \eqref{eq:RECconsequence}, step (ii) is implied by $P_{\alpha,S}(\hat\bc - \bc)\leq \frac{4\bar{\omega}_0 \norm{\hat{\bV}_h^{(d)}(\hat{\bc}-\bc)}_2}{\phi_0}$ for $\hat\bc-\bc\in\mathcal{C}_{N_c}(\mathcal{G},S)$ (Lemma \ref{Lemma:V_transfer}), and step (iii) uses the elementary inequality $16 uv \leq u^2 + 64 v^2$ (i.e. manipulating $(u-8v)^2\geq0$).
A straightforward rearrangement of \eqref{eq:finalineqSGL} provides the inequality of Theorem \ref{Thm:sgl}.

It remains to determine a lower bound on the probability of $\mathcal{H}_1(\frac{\lambda}{4})\cap \mathcal{H}_2(\frac{\lambda}{4})\cap\mathcal{V}\left(\frac{\phi_0}{2}\right)$. We rely on the elementary inequality $\Prob\left(\mathcal{H}_1(\frac{\lambda}{4})\cap \mathcal{H}_2(\frac{\lambda}{4})\cap\mathcal{V}\left(\frac{\phi_0}{2}\right)\right) \geq 1 - \Prob\left(\mathcal{H}_1(\frac{\lambda}{4})^c\right) - \Prob\left(\mathcal{H}_2(\frac{\lambda}{4})^c \right)- \Prob\left(\mathcal{V}\left(\frac{\phi_0}{2}\right)^c \right)$ to bound the individual probabilities.

We start with $\Prob\left(\mathcal{H}_1(\frac{\lambda}{4})^c\right)= \Prob\left( 2 P_\alpha^*\left( \hat{\bV}_h^{(d)\prime}\hat{\bDelta}_\Sigma\right) > \frac{\lambda}{4}\right) \leq \Prob\left( \norm{\hat{\bV}_h^{(d)\prime}\hat{\bDelta}_\Sigma}_\infty > \frac{\lambda}{8} \right)$. The last inequality is true because continuing from \eqref{eq:Pstarupperbound}, we have
\begin{equation}
P_\alpha^*(\bc)\leq \alpha \max_{g \in \setG}\frac{\norm{\bc_g}_2}{\sqrt{\abs{g}}} + (1-\alpha)\norm{\bc}_\infty \leq \alpha \max_{g \in \setG}\frac{\sqrt{|g|} \norm{\bc_g}_{\infty}}{\sqrt{\abs{g}}} + (1-\alpha) \norm{\bc}_\infty = \norm{\bc}_\infty
\label{eq:maxnorm}
\end{equation}
for any vector $\bc$. Subsequently, we have
\begin{equation}
\begin{aligned}
 \Prob&\left( \norm{\hat{\bV}_h^{(d)\prime}\hat{\bDelta}_\Sigma}_\infty > \frac{\lambda}{8} \right)
 = \Prob\left( \norm{ \left[(\hat{\bV}_h^{(d)} - \bV^{(d)}) + \bV^{(d)}\right]\tran \hat{\bDelta}_\Sigma}_\infty > \frac{\lambda}{8} \right) \\
 &\leq \Prob\left( \norm{\hat{\bV}_h - \bV}_1\norm{\hat{\bDelta}_\Sigma}_\infty + C_V \norm{\hat{\bDelta}_\Sigma}_\infty > \frac{\lambda}{8} \right)
 \leq \Prob\left( \normoneinf{\hat{\bV}_h - \bV}^2 + C_V\normoneinf{\hat{\bV}_h - \bV} > \frac{\lambda}{8} \right)\\
 &\leq \Prob\left( \normoneinf{\hat{\bV}_h - \bV} > \frac{\lambda^{1/2}}{4} \right) + \Prob\left(\normoneinf{\hat{\bV}_h - \bV} > \frac{\lambda}{16C_V} \right),
\end{aligned}
\label{eq:H1}
\end{equation}
exploiting block-diagonality of $\hat{\bV}_h^{(d)} - \bV^{(d)}$ and $\bV^{(d)}$ such that $\norm{\hat{\bV}_h^{(d)} - \bV^{(d)}}_1 = \max_{1\leq i \leq N} \norm{\hat\bV_{i,h} - \bV_i}_1 \leq \normoneinf{\hat\bV_h - \bV}$ and $\norm{\bV^{(d)}}_1\leq \normoneinf{\bV}\leq C_V$ (explicitly assumed in Theorem \ref{Thm:sgl}). Bounds for the final RHS terms in \eqref{eq:H1} are available from Theorem \ref{th:diag_approx}.

Second, we have
\begin{equation}
\begin{aligned}
 \Prob\left(\mathcal{H}_2\left(\frac{\lambda}{4}\right)^c \right)
 &\leq  \Prob\left( \norm{\hat{\bV}_h^{(d)\prime}\hat\bDelta_V \bc}_\infty > \frac{\lambda}{8} \right)
 \leq \Prob\left( \norm{\hat{\bV}_h^{(d)\prime}\left[\hat{\bV}_h^{(d)} - \bV^{(d)}\right]}_\infty > \frac{\lambda}{8\norm{\bc}_\infty} \right)\\
 &\leq \Prob\left( \norm{\hat{\bV}_h - \bV}_1\norm{\hat{\bV}_h - \bV}_\infty +  C_V\norm{\hat{\bV}_h - \bV}_\infty> \frac{\lambda}{8\norm{\bc}_\infty} \right)\\
 &\leq \Prob\left( \normoneinf{\hat{\bV}_h - \bV} > \frac{\lambda^{1/2}}{4} \right) + \Prob\left( \normoneinf{\hat{\bV}_h - \bV} > \frac{\lambda}{16 C_V} \right),
\end{aligned}
\label{eq:ProbabilitySetBoundH2}
\end{equation}
where the last line relies on the union bound and the fact that $\norm{\bc}_\infty < 1$ (implied by Assumption \ref{assump:stability}). Hence, the sets $\mathcal{H}_1\left(\frac{\lambda}{4}\right)^c$ and $\mathcal{H}_2\left(\frac{\lambda}{4}\right)^c$ admit the same probability bound.

Finally, $\Prob\left(\mathcal{V}\left(\frac{\phi_0}{2}\right)^c\right)=\Prob\left(\norm{\hat\bV_h -\bV}_1 > \frac{\phi_0}{2}\right)\leq \Prob\left(\normoneinf{\hat\bV_h -\bV} > \frac{\phi_0}{2}\right)$. Combining all previous results, we conclude
\begin{equation}
\begin{aligned}
\Prob &\left( \mathcal{H}_1\left(\frac{\lambda}{4} \right) \cap \mathcal{H}_2\left(\frac{\lambda}{4}\right) \cap \mathcal{V} \left(\frac{\phi_0}{2}\right) \right) \\ &\geq 1 - 2\Prob\left( \normoneinf{\hat{\bV}_h - \bV} > \frac{\lambda^{1/2}}{4} \right) - 2\Prob\left( \normoneinf{\hat{\bV}_h - \bV} > \frac{\lambda}{16C_V} \right) - \Prob\left(\normoneinf{\hat\bV_h -\bV} > \frac{\phi_0}{2}\right)\\
&\geq 1 - 5\Prob\left(\normoneinf{\hat\bV_h -\bV} > 6f(\lambda,\phi_0)\right),
\end{aligned}
\label{eq:probability_combined_sets}
\end{equation}
where $f(\lambda,\phi_0) = \min \left(\frac{\lambda^{1/2}}{24},\frac{\lambda}{96C_V},\frac{\phi_0}{12}\right)$. The proof is completed by evaluating the final probability in \eqref{eq:probability_combined_sets} using Theorem \ref{th:diag_approx}.
\end{proof}
    
\begin{proof}[\textbf{Proof of Corollary \ref{cor:splash_rates}}]
First, we derive the conditions under which the set on which the performance bound in Theorem \ref{Thm:sgl} holds occurs with probability converging to one. Under Assumption \ref{assump:momentcond}(b1), along with the remaining assumptions in Theorem \ref{Thm:sgl}, this probability is given by $1-\mathcal{P}_1(f(\lambda,\phi_0),N,T)$, where we recall from Theorem \ref{th:diag_approx} that
\begin{equation}\label{eq:P_1}
\begin{split}
    \mathcal{P}_1(f(\lambda,\phi_0),N,T) &= N^2\left(b_1 T^{(1-\delta)/3}+\frac{[2h(f(\lambda,\phi_0))+1]b_3}{f(\lambda,\phi_0)}\right) \exp\left(-\frac{T^{(1-\delta)/3}}{2b_1^2}\right)\\ 
    &\qquad + \frac{b_2 N^2[2h(f(\lambda,\phi_0))+1]^d}{f(\lambda,\phi_0)^d T^{\frac{\delta}{2}(d-1)}}\\
\end{split}
\end{equation}
Given that $\lambda \in O\left(T^{-q_\lambda}\right)$ with $q_\lambda > 0$, it follows immediately that
\begin{equation}\label{eq:f_order}
    f(\lambda,\phi_0) = \min \left(\frac{\lambda^{1/2}}{24},\frac{\lambda}{96C_V},\frac{\phi_0}{12}\right) = O\left(T^{-q_\lambda}\right).
\end{equation}
In addition, following the remark below Theorem \ref{th:diag_approx}, it holds that
\begin{equation}\label{eq:h_order}
    h\left(\frac{\lambda}{96C_V}\right) = O\left(\log\left(\lambda^{-1}\right)^2k_0\right) = O\left(\log(T)^2T^{q_k}\right).
\end{equation}
Based on \eqref{eq:f_order} and \eqref{eq:h_order}, it follows that the first RHS-term in \eqref{eq:P_1} converges to zero exponentially in $T$ for any $\delta<1$. The second RHS-term, however, converges to zero at most at a polynomial rate. Accordingly,
\begin{equation}\label{eq:P_1_orders}
\begin{split}
    \mathcal{P}_1(f(\lambda,\phi_0),N,T) = O\left(\frac{N^2[h\left(\frac{\lambda}{96C_V}\right)+1]^d}{\lambda^d T^{\frac{\delta}{2}(d-1)}}\right) = O\left(\log(T)^{2d}T^{2q_N + dq_k+dq_\lambda - \frac{\delta(d-1)}{2}}\right),
\end{split}
\end{equation}
where the second equality holds from the observation that for any $\delta<1$ the first two RHS terms in the first equality converge to zero at an exponential rate, whereas the third term may converge to zero at most at a polynomial rate. From \eqref{eq:P_1_orders}, it follows that $\mathcal{P}_1(f(\lambda,\phi_0),N,T) \to 0$ if 
\begin{equation*}
    2q_n + dq_k + dq_\lambda - \frac{\delta(d-1)}{2} < 0 \Rightarrow q_\lambda < -\frac{2q_N}{d} - q_k + \frac{\delta(d-1)}{2d}.
\end{equation*}
In a similar fashion, we derive the conditions under which $\mathcal{P}_2(f(\lambda,\phi_0) \to 0$, by noting that
\begin{equation}
\begin{split}
    \mathcal{P}_2(f(\lambda,\phi_0),N,T) = &N^2 \left[\frac{\kappa_1[2h(f(\lambda,\phi_0))+1]}{f(\lambda,\phi_0)} + \frac{2}{\kappa_2}\left(\frac{T f(\lambda,\phi_0)^2}{[2h(f(\lambda,\phi_0))+1]^2}\right)^{\frac{1}{7}} \right] \times \\
    &\exp\left(- \frac{1}{\kappa_3}\left(\frac{T f(\lambda,\phi_0)^2}{[2h(f(\lambda,\phi_0))+1]^2}\right)^{\frac{1}{7}} \right)
\end{split}
\end{equation}
converges to zero exponentially fast in $T$ if $\frac{T f(\lambda,\phi_0)^2}{[2h(f(\lambda,\phi_0))+1]^2}$ diverges at a polynomial rate in $T$. Making use of \eqref{eq:f_order} and \eqref{eq:h_order}, it follows that
\begin{equation*}
    \frac{T f(\lambda,\phi_0)^2}{[2h(f(\lambda,\phi_0))+1]^2} = O\left(\log(T)^{-4}T^{1-2q_\lambda - 2q_k}\right),
\end{equation*}
which translates to the condition $1-2q_\lambda - 2q_k > 0$, or $q_\lambda < \frac{1}{2} - q_k$. This establishes conditions (i) and (ii) in Corollary \ref{cor:splash_rates}.

We proceed by deriving the order of the performance bound in Theorem \ref{Thm:sgl}. Noting that
\begin{equation*}
    \sum_{g \in \mathcal{G}_S}\sqrt{\abs{g}} \leq \abs{\mathcal{G}_S}\max_{g \in \mathcal{G}_S} \sqrt{\abs{g}} = O\left(T^{q_g +q_N/2}\right),
\end{equation*}
it follows that
\begin{equation*}
    \bar{\omega} = O\left((1-\alpha)T^{q_g +q_N/2} + \alpha T^{q_S/2}\right). 
\end{equation*}
Then, by Theorem \ref{Thm:sgl},
\begin{equation*}
    \norm{\hat{\bV}_h^{(d)}(\hat{\bc}-\bc)}_2^2 \leq \frac{4\bar{\omega}_\alpha^2\lambda^2}{\phi_0^2} = O\left((1-\alpha)T^{2q_g +q_N - 2q_\lambda} + \alpha T^{q_S - 2q_\lambda}\right)
\end{equation*}
and
\begin{equation*}
    (1-\alpha)\sum_{g \in \mathcal{G}} \sqrt{\abs{g}}\norm{\hat{\bc}_g - \bc_g}_2 + \alpha\norm{\hat{\bc}-\bc}_1 \leq \frac{4\bar{\omega}_\alpha^2\lambda}{\phi_0^2} = O\left((1-\alpha)T^{2q_g +q_N - q_\lambda} + \alpha T^{q_S - q_\lambda}\right),
\end{equation*}
on a set with probability $1-\mathcal{P}_1(f(\lambda,\phi_0),N,T)$ or $1-\mathcal{P}_2(f(\lambda,\phi_0),N,T)$, depending on whether Assumption \ref{assump:momentcond}(b1) or \ref{assump:momentcond}(b2) applies, respectively. Since we have shown that both $\mathcal{P}_1(f(\lambda,\phi_0),N,T) \to 0$ and $\mathcal{P}_2(f(\lambda,\phi_0),N,T) \to 0$ under the conditions imposed in Corollary \ref{cor:splash_rates}, the proof is complete.
\end{proof}
\end{small}

\end{appendices}

\makeatletter\@input{xx.tex}\makeatother
\end{document}


\maketitle

\tableofcontents

\appendix 
\section{Auxiliary Lemmas related to Theorem \ref{th:diag_approx}}

\emph{\textbf{Note:} Lemmas \ref{lemmamomenttruncation}--\ref{lemma:covmatrixelements} are strongly related to \cite{Masini2019}. That is, small differences occur because there are additional constants that originate from rewriting our Spatial VAR into reduced form. To keep track of all constants, and in order to have all results in the notation of the main paper, we decided to keep the full derivations in this Supplement.}

\begin{lemma}[Moment and Truncation Bounds]\label{lemmamomenttruncation} \
Define $\by_t^{(k)}$ by truncating the VMA($\infty$) representation for $\by_t$ at the $k$\textsuperscript{th} lag, that is $\by_t^{(k)}= \sum_{j=0}^{k-1} \bC^j \bD \bepsilon_{t-j}$. We have
$$
\text{(a) }\norm{y_{it}}_\psi \leq \mu_\psi c_D c_1,
\qquad 
\text{(b) }\|y_{it}^{(k)}\|_\psi \leq \mu_\psi c_D c_1,
\qquad
\text{(c) }\|y_{i,t}-y_{i,t}^{(k)}\|_\psi \leq \mu_\psi c_D c_1 e^{-\gamma_c k}.
$$
under Assumptions \ref{assump:stability} and \ref{assump:momentcond}.
\label{lemma:boundsonmoments}
\end{lemma}
\begin{proof}
 \textbf{(a)} We have $y_{it} = \be_i\tran \by_t$. Using the VMA($\infty$) representation of the spatio-temporal vector autoregression, we have for any Orlicz norm
 \begin{equation}
 \begin{aligned}
  &\norm{y_{it}}_\psi
    = \norm{\be_i\tran \by_t}_\psi
    =  \norm{ \sum_{j=0}^\infty \be_i\tran \bC^j \bD \bepsilon_{t-j} }_\psi
    \leq \sum_{j=0}^\infty \sum_{k=1}^N \left|\be_i^\prime \bC^j \bD \be_k \right| \norm{\epsilon_{k,t-j}}_\psi, \\
    &\sup_{i,t} \norm{y_{it}}_\psi \leq \left( \sup_{i,t} \norm{\epsilon_{it}}_\psi \right)\sum_{j=0}^\infty \norm{\bC^j \bD}_\infty
    \leq  \left( \sup_{i,t} \norm{\epsilon_{it}}_\psi \right) \norm{\bD}_\infty \sum_{j=0}^\infty \norm{\bC^j}_\infty
    \leq \mu_\psi c_D c_1,
 \end{aligned}
 \label{eq:orliczbound}
 \end{equation}
 by the norm properties of the Orlicz norm, Assumption \ref{assump:stability}(a)-(b), and Assumption \ref{assump:momentcond}(b). \textbf{(b)} $\|y_{it}^{(k)}\|_\psi = \|\be_i\tran \by_t^{(k)} \|_\psi \leq \sum_{j=0}^{k-1} \sum_{k=1}^N \left|\be_i^\prime \bC^j \bD \be_k \right| \norm{\epsilon_{k,t-j}}_\psi \leq \sum_{j=0}^\infty \sum_{k=1}^N \left|\be_i^\prime \bC^j \bD \be_k \right| \norm{\epsilon_{k,t-j}}_\psi$ and continue as in \eqref{eq:orliczbound}. \textbf{(c)} A small variation on \eqref{eq:orliczbound} provides
 $$
 \norm{y_{i,t}-y_{i,t}^{(k)}}_\psi
  = \norm{\sum_{j=k}^\infty \be_i^\prime \bC^j \bD \bepsilon_{t-j}}_\psi
  \leq \left( \sup_{i,t} \norm{\epsilon_{it}}_\psi \right)\norm{\bD}_\infty \left(\sum_{j=k}^\infty \norm{\bC^j}_\infty \right)\leq \mu_\psi c_D c_1 e^{-\gamma_c k}.
$$
The proof is complete.
 \end{proof}

\begin{lemma}[Sample Covariance Matrices 1]
Define the $N(p+1)$ vector $\bxi_t = (\by_t\tran,\by_{t-1}\tran,\ldots,\by_{t-p}\tran)\tran$. For all $j=1,\ldots,p$, the elements of $\widehat{\bSigma}_j = \frac{1}{T}\sum_{t=p+1}^T\by_t \by_{t-j}\tran$ can all be expressed as $\frac{1}{T}\sum_{t=p+1}^T \xi_{it} \xi_{jt}$ after an appropriate choice of $(i,j)$. We will look at $\sum_{t=p+1}^T \xi_{it} \xi_{jt}$ (so ignoring the multiplicative factor $\frac{1}{T}$). 
\begin{enumerate}[(a)]
    \item We have $\sum_{t=p+1}^T \xi_{it} \xi_{jt} - \E(\xi_{it}\xi_{jt}) = S_{1T} + S_{2T}$ where $S_{1T} = \sum_{s=0}^{m-1} \sum_{t=p+1}^T V_{t,s}$ with
    $$
     V_{t,s}=\E\left(\left. \xi_{it} \xi_{jt} \ind{ |\xi_{it} \xi_{jt} | < C}\right| \calF_{t-s} \right) -  \E\left(\left.\xi_{it} \xi_{jt} \ind{ |\xi_{it} \xi_{jt} | < C}\right| \calF_{t-s-1} \right)
    $$
    (a martingale difference sequence satisfying $|V_{t,s}|\leq 2 C$), and
    \begin{equation}
     S_{2T}= \sum_{t=p+1}^T \xi_{it} \xi_{jt} \ind{|\xi_{it} \xi_{jt} | \geq C} + \sum_{t=p+1}^T \E\left(\left.\xi_{it} \xi_{jt} \ind{ |\xi_{it} \xi_{jt} | < C} \right| \calF_{t-m} \right) - \E\left( \xi_{it} \xi_{jt} \right).
    \label{eq:S2}
    \end{equation}
    \item $\Prob\left(|S_{1T}| \geq \frac{T \epsilon}{2}  \right) \leq 2m \exp\left( - \frac{T \epsilon^2}{8 m^2 C^2} \right)$.
    \item We have $\frac{1}{T}\E|S_{2T}|\leq I_T + II_T + III_T + IV_T$ with
     \begin{itemize}
         \item $I_T =\frac{1}{T}\E\left|\sum_{t=p+1}^T \xi_{it} \xi_{jt} \ind{|\xi_{it} \xi_{jt} | \geq C}- \E\left(\left. \xi_{it}^{(k)} \xi_{jt}^{(k)} \ind{|\xi_{it} \xi_{jt} | \geq C}\right| \calF_{t-m} \right) \right|$,
         \item $II_T = \frac{1}{T}\E\left|\sum_{t=p+1}^T \E\left(\left. \xi_{it} \xi_{jt} \ind{ |\xi_{it} \xi_{jt} | < C} \right| \calF_{t-m} \right) - \E\left(\left. \xi_{it}^{(k)} \xi_{jt}^{(k)} \ind{|\xi_{it} \xi_{jt} | < C}\right| \calF_{t-m} \right) \right|$,
         \item $III_T = \frac{1}{T}\left|\sum_{t=p+1}^T \E(\xi_{it}^{(k)} \xi_{jt}^{(k)})- \E\left( \xi_{it} \xi_{jt} \right)\right|$,
         \item $IV_T= \frac{1}{T} \E\left| \sum_{t=p+1}^T \E\left(\left. \xi_{it}^{(k)} \xi_{jt}^{(k)}\right| \calF_{t-m} \right) - \E(\xi_{it}^{(k)} \xi_{jt}^{(k)})  \right|$.
     \end{itemize}
     \item If Assumptions \ref{assump:stability} and \ref{assump:momentcond}(b1) (polynomial tails) hold, then we have: $I_T\leq \frac{2 (\mu_{2d} c_D c_1)^{2d} }{C^{d-1}}$, $II_T\leq 2 (\mu_2 c_D c_1)^2 e^{-\gamma_c k}$, $III_T\leq 2 (\mu_2 c_D c_1)^2 e^{-\gamma_c k}$, and $IV_T \leq 6 c_2 (\mu_{4} c_D c_1)^2 e^{-\frac{\gamma_\alpha}{2}(m-k-p)}$.
     \item If Assumptions \ref{assump:stability} and \ref{assump:momentcond}(b2) (subexponential tails) hold, then we have:
     $$
      I_T\leq 2304 (\mu_\infty c_D c_1)^2 \exp\left( - \sqrt{C}/\mu_\infty c_D c_1 \right),
      $$
      $II_T\leq 8 (\mu_\infty c_D c_1)^2 e^{-\gamma_c k}$, $III_T\leq 8 (\mu_\infty c_D c_1)^2 e^{-\gamma_c k}$, and $IV_T\leq 3456 c_2 (\mu_{\infty} c_D c_1)^2 e^{-\frac{\gamma_\alpha}{2}(m-k-p)}$.
\end{enumerate}
\label{lemma:covarelements}
\end{lemma}
\begin{proof}
\textbf{(a)} Start from $\xi_{it} \xi_{jt} = \xi_{it} \xi_{jt} \ind{|\xi_{it} \xi_{jt}| <C}+\xi_{it} \xi_{jt} \ind{|\xi_{it} \xi_{jt}| \geq C}$ and develop the first term as
$$
\begin{aligned}
 \xi_{it} \xi_{jt} &\ind{|\xi_{it} \xi_{jt}| <C}
  = \E\left(\left. \xi_{it} \xi_{jt} \ind{|\xi_{it} \xi_{jt} | < C} \right| \calF_t \right) \\
  &= \sum_{s=0}^{m-1} \Bigg[ \underbrace{\E\left(\left. \xi_{it} \xi_{jt} \ind{|\xi_{it} \xi_{jt} | < C} \right| \calF_{t-s} \right) - \E\left(\left.\xi_{it} \xi_{jt} \ind{|\xi_{it} \xi_{jt} | < C} \right| \calF_{t-s-1} \right)}_{V_{t,s}} \Bigg] \\
  &\qquad+ \E\left(\left.\xi_{it} \xi_{jt} \ind{|\xi_{it} \xi_{jt} | < C} \right| \calF_{t-m} \right).
\end{aligned}
$$
Subsequently, combine these results and sum over $t=p+1,\ldots,T$. \textbf{(b)} We have
$$
\begin{aligned}
 \Prob\left( |S_{1T}| \geq \frac{T\epsilon}{2} \right)
  &= \Prob\left( \left| \sum_{s=0}^{m-1} \sum_{t=p+1}^T V_{t,s} \right| \geq \frac{T\epsilon}{2} \right)
  \leq \sum_{s=0}^{m-1} \Prob\left(\left| \sum_{t=p+1}^T V_{t,s} \right| \geq \frac{T\epsilon}{2m} \right) \\
  &\leq \sum_{s=0}^{m-1} 2\exp\left( - \frac{2(T\epsilon/2m)^2}{(2C)^2(T-p)} \right)
  \leq 2 m \exp\left( - \frac{T\epsilon^2}{ 8 m^2 C^2} \right),
\end{aligned}
$$
where the second line follows from an application of the Azuma-Hoeffding inequality (see, e.g. corollary 2.20 in \cite{Wainwright2019}). \textbf{(c)} The result requires the use of the triangle inequality and some subtracting and adding of identical terms:
$$
\begin{aligned}
 &\frac{1}{T}\E|S_{2T}|
  = \frac{1}{T}\E\left|\sum_{t=p+1}^T \xi_{it} \xi_{jt} \ind{|\xi_{it} \xi_{jt} | \geq C} + \sum_{t=p+1}^T \E\left(\left.\xi_{it} \xi_{jt} \ind{ |\xi_{it} \xi_{jt} | < C} \right| \calF_{t-m} \right) - \E\left( \xi_{it} \xi_{jt} \right) \right| \\
  &\leq  \underbrace{\frac{1}{T}\E\left|\sum_{t=p+1}^T \xi_{it} \xi_{jt} \ind{|\xi_{it} \xi_{jt} | \geq C}- \E\left(\left. \xi_{it}^{(k)} \xi_{jt}^{(k)} \ind{|\xi_{it} \xi_{jt} | \geq C}\right| \calF_{t-m} \right) \right|}_{I_T} \\
  &\qquad+\frac{1}{T}\E\left|\sum_{t=p+1}^T  \E\left(\left. \xi_{it} \xi_{jt} \ind{ |\xi_{it} \xi_{jt} | < C} \right| \calF_{t-m} \right) + \E\left(\left. \xi_{it}^{(k)} \xi_{jt}^{(k)} \ind{|\xi_{it} \xi_{jt} | \geq C}\right| \calF_{t-m} \right) - \E\left( \xi_{it} \xi_{jt} \right) \right| \\
  &= I_T + \frac{1}{T}\E\Bigg|\sum_{t=p+1}^T \E\left(\left.\xi_{it} \xi_{jt} \ind{ |\xi_{it} \xi_{jt} | < C} \right| \calF_{t-m} \right) - \E\left(\left. \xi_{it}^{(k)} \xi_{jt}^{(k)} \ind{|\xi_{it} \xi_{jt} | < C}\right| \calF_{t-m} \right) \\
  &\qquad+ \E\left(\left. \xi_{it}^{(k)} \xi_{jt}^{(k)}\right| \calF_{t-m} \right) - \E\left( \xi_{it} \xi_{jt} \right) \Bigg| \\
  &\leq I_T + \underbrace{\frac{1}{T}\E\left|\sum_{t=p+1}^T \E\left(\left. \xi_{it} \xi_{jt} \ind{ |\xi_{it} \xi_{jt} | < C} \right| \calF_{t-m} \right) - \E\left(\left. \xi_{it}^{(k)} \xi_{jt}^{(k)} \ind{|\xi_{it} \xi_{jt} | < C}\right| \calF_{t-m} \right) \right|}_{II_T} \\
  &\qquad+\frac{1}{T}\E\left|\sum_{t=p+1}^T \E\left(\left. \xi_{it}^{(k)} \xi_{jt}^{(k)}\right| \calF_{t-m} \right) - \E\left( \xi_{it} \xi_{jt} \right) \right| \\
  &= I_T + II_T + \frac{1}{T}\E\left| \sum_{t=p+1}^T \E( \xi_{it}^{(k)} \xi_{jt}^{(k)})- \E\left( \xi_{it} \xi_{jt} \right) + \E\left(\left. \xi_{it}^{(k)} \xi_{jt}^{(k)}\right| \calF_{t-m} \right) - \E( \xi_{it}^{(k)} \xi_{jt}^{(k)}) \right| \\
  &\leq I_T + II_T + \underbrace{\frac{1}{T}\left|\sum_{t=p+1}^T \E( \xi_{it}^{(k)} \xi_{jt}^{(k)})- \E\left( \xi_{it} \xi_{jt} \right)  \right|}_{III_T}
  + \underbrace{\frac{1}{T} \E\left| \sum_{t=p+1}^T \E\left(\left. \xi_{it}^{(k)} \xi_{jt}^{(k)}\right| \calF_{t-m} \right) - \E( \xi_{it}^{(k)} \xi_{jt}^{(k)})  \right|}_{IV_T}
\end{aligned}
$$
The indicated terms coincide with those defined in Lemma \ref{lemma:covarelements}(c). \textbf{(d)} For any random variable $X$ with $\norm{X}_d<\infty$ ($d>1$), it holds that
\begin{equation}
 \E\Big(|X|\ind{|X|\geq C} \Big)
  \leq \norm{X}_d \Big[ \Prob(|X|\geq C) \Big]^{(d-1)/d} \leq 
  \norm{X}_d \left( \frac{\norm{X}_d^d}{C^d} \right)^{(d-1)/d}
    = \frac{\norm{X}_d^d}{C^{d-1}},
\label{eq:boundXind}
\end{equation}
using H\"{o}lder's and Markov's inequality in the first and second inequality, respectively. By the conditional Jensen inequality, \eqref{eq:boundXind}, Cauchy-Schwartz and Lemma \ref{lemma:boundsonmoments}, we can indeed bound $I_T$ as
$$
\begin{aligned}
 I_T &\leq \frac{1}{T} \sum_{t=p+1}^T \E \big|\xi_{it} \xi_{jt} \ind{|\xi_{it} \xi_{jt} | \geq C} \big| + \E\big| \xi_{it}^{(k)} \xi_{jt}^{(k)} \ind{|\xi_{it} \xi_{jt} | \geq C}\big|
 \stackrel{\eqref{eq:boundXind}}{\leq} \frac{1}{T} \sum_{t=p+1}^T\frac{\| \xi_{it} \xi_{jt}\|_d^d+\|\xi_{it}^{(k)}\xi_{jt}^{(k)}\|_d^d}{C^{d-1}} \\
 &\leq \frac{1}{T} \sum_{t=p+1}^T\frac{\|\xi_{it}\|_{2d}^d \;\|\xi_{jt}\|_{2d}^d + \|\xi_{it}^{(k)}\|_{2d}^d \; \|\xi_{jt}^{(k)}\|_{2d}^d}{C^{d-1}} \stackrel{\eqref{lemma:boundsonmoments}}{\leq} \frac{1}{T} \sum_{t=p+1}^T \frac{2 (\mu_{2d} c_D c_1)^{2d} }{C^{d-1}} \leq \frac{2 (\mu_{2d} c_D c_1)^{2d} }{C^{d-1}}.
\end{aligned}
$$
For $II_T$, the triangle and the conditional Jensen's inequality give
$$
II_T
 \leq \frac{1}{T}\sum_{t=p+1}^T \E \left| \E\left(\left. ( \xi_{it} \xi_{jt} - \xi_{it}^{(k)} \xi_{jt}^{(k)} ) \ind{|\xi_{it} \xi_{jt}|<C}\right| \calF_{t-m}\right) \right|
 \leq \frac{1}{T} \sum_{t=p+1}^T \E\left|\xi_{it}\xi_{jt} - \xi_{it}^{(k)}\xi_{jt}^{(k)}\right|,
$$
and its RHS is further bounded as
\begin{equation}
\begin{aligned}
  \frac{1}{T} &\sum_{t=p+1}^T \E\left|\xi_{it}\xi_{jt} - \xi_{it}^{(k)}\xi_{jt}^{(k)}\right|
   = \frac{1}{T} \sum_{t=p+1}^T \E\left|\xi_{it}\xi_{jt} - \xi_{it} \xi_{jt}^{(k)} + \xi_{it} \xi_{jt}^{(k)} - \xi_{it}^{(k)}\xi_{jt}^{(k)}\right| \\
   &\leq \frac{1}{T} \sum_{t=p+1}^T \E\left| \xi_{it}(\xi_{jt}-\xi_{jt}^{(k)}) \right| + \E\left| (\xi_{it}-\xi_{it}^{(k)}) \xi_{jt}^{(k)}\right| \\
   &\leq \frac{1}{T} \sum_{t=p+1}^T \norm{\xi_{it} }_2 \norm{\xi_{jt}-\xi_{jt}^{(k)}}_2 + \norm{\xi_{it}-\xi_{it}^{(k)}}_2 \norm{ \xi_{jt}^{(k)} }_2 \leq 2 (\mu_2 c_D c_1)^2 e^{-\gamma_c k}.
\end{aligned}
\label{eq:boundsIITandIIIT}
\end{equation}
The last inequality uses the results in Lemma \ref{lemma:boundsonmoments} based on the Orlicz norm induced by $\psi(x)=x^2$. As $III_T\leq \frac{1}{T}\sum_{t=p+1}^T \E\left| \xi_{it}^{(k)} \xi_{jt}^{(k)} - \xi_{it} \xi_{jt} \right|$ (triangle inequality), the bound for $III_T$ follows immediately from \eqref{eq:boundsIITandIIIT}. Finally, for $IV_T$, we have $IV_T\leq \frac{1}{T} \sum_{t=p+1}^T \E\Big| \E\big[ \xi_{it}^{(k)} \xi_{jt}^{(k)} \big| \calF_{t-m} \big] -\E\big[ \xi_{it}^{(k)} \xi_{jt}^{(k)} \big] \Big|$. Since $\{\xi_{it}^{(k)} \xi_{jt}^{(k)}\}$ is $\alpha$-mixing with exponentially decaying mixing coefficients $\alpha_m^{(k)}$, we will use theorem 14.2 from \cite{Davidson1994} with $p=1$ and $r=2$:\footnote{Each $\xi_{it}^{(k)} \xi_{jt}^{(k)}$ is a weighted linear combination of $\epsilon_t,\epsilon_{t-1},\ldots,\epsilon_{t-p-k+1}$. According to theorem 14.1 in \cite{Davidson1994}, the mixing coefficients of $\xi_{it}^{(k)} \xi_{jt}^{(k)}$, $\alpha_m^{(k)}$, satisfy $\alpha_m^{(k)}\leq \alpha_{m-p-k+1}$.}
\begin{equation}
\begin{aligned}
\E\Big| \E &\big[ \xi_{it}^{(k)} \xi_{jt}^{(k)} \big| \calF_{t-m} \big] -\E\big[ \xi_{it}^{(k)} \xi_{jt}^{(k)} \big] \Big|
 \leq 6 \; [\alpha_m^{(k)}]^{1-1/2} \;  \| \xi_{it}^{(k)} \xi_{jt}^{(k)} \|_2 \\
 &\leq 6 c_2 e^{-\frac{\gamma_\alpha}{2} (m-k-p+1)} \norm{\xi_{it}^{(k)}}_{4}  \norm{\xi_{jt}^{(k)}}_{4}
 \leq 6 c_2 (\mu_{4} c_D c_1)^2 e^{-\frac{\gamma_\alpha}{2}(m-k-p)}.
\end{aligned}
\end{equation}
\textbf{(e)} We now use $\norm{\cdot}_\psi$ to denote the Orlicz norm induced by the function $\psi(x)=\exp(x)-1$. We first derive an upper bound for $\E\left| \xi_{it}\xi_{jt} \ind{|\xi_{it}\xi_{jt}|>C} \right|$. By the Cauchy-Schwartz inequality, we find
$$
\begin{aligned}
 \E\left| \xi_{it}\xi_{jt} \ind{|\xi_{it}\xi_{jt}|>C} \right|
  &\leq \norm{\xi_{it} \xi_{jt} }_2 \big[ \Prob( |\xi_{it} \xi_{jt}| \geq C) \big]^{1/2} \\
  &\leq \norm{\xi_{it}}_4 \norm{\xi_{jt}}_4 \Big[ \Prob\big( |\xi_{it}| \geq \sqrt{C}\big) + \Prob\big( |\xi_{jt}| \geq \sqrt{C}\big) \Big]^{1/2}.
\end{aligned}
$$
The moment bound on page 95 of \cite{vandervaartwellner1996} and the tail bound as in Exercise 2.18 in \cite{Wainwright2019} provide $\norm{\xi_{it}}_4 \leq \mu_4 c_D c_1 \leq 4! \mu_\infty c_D c_1$ and $\Prob(|\xi_{it}|\geq \sqrt{C}) \leq 2 \exp(-\sqrt{C}/\mu_\infty c_D c_1)$, respectively. The upper bound on $\E\left| \xi_{it}^{(k)} \xi_{jt}^{(k)} \ind{|\xi_{it} \xi_{jt}|>C} \right|$ is the same and follows along the same steps. Overall, we have
$$
\begin{aligned}
 I_T
 &\leq \frac{1}{T} \sum_{t=p+1}^T \E \big|\xi_{it} \xi_{jt} \ind{|\xi_{it} \xi_{jt} | \geq C} \big| + \E\big| \xi_{it}^{(k)} \xi_{jt}^{(k)} \ind{|\xi_{it} \xi_{jt} | \geq C}\big| \\
 &\leq \frac{1}{T}\sum_{t=p+1}^T 4 (4! \mu_\infty c_D c_1)^2 \exp\left( - \sqrt{C}/2\mu_\infty c_D c_1 \right)
 \leq 2304 (\mu_\infty c_D c_1)^2 \exp\left( - \sqrt{C}/2\mu_\infty c_D c_1 \right).
\end{aligned}
$$
For $II_T$, $III_T$ and $IV_T$, we use the same methods of proof but express the moments in terms of $\mu_\infty$, e.g. $\mu_2 \leq 2! \mu_{\infty}$. The claims are immediate.
\end{proof}

\begin{lemma}[Sample Covariance Matrices 2]\
\label{lemma:covmatrixelements}
\begin{enumerate}[(a)]
    \item If Assumptions \ref{assump:stability} and \ref{assump:momentcond}(b1) (polynomial tails) hold, then we have:
    $$
     \Prob\Big( \Big| \sum_{t=p+1}^T \xi_{it} \xi_{jt} -\E(\xi_{it} \xi_{jt}) \Big| >  T \epsilon \Big)
  \leq \left[b_1 T^{(1-\delta)/3}+\frac{b_3}{\epsilon}\right] \exp\left(-\frac{T^{(1-\delta)/3}}{2b_1^2}\right) + \frac{b_2}{\epsilon^d T^{\frac{\delta}{2}(d-1)} },
    $$ for some $0 < \delta < 1$.
    \item If Assumptions \ref{assump:stability} and \ref{assump:momentcond}(b2) (subexponential tails) hold, then we have:
    $$
     \Prob\Big( \Big| \sum_{t=p+1}^T \xi_{it} \xi_{jt} - \E(\xi_{it} \xi_{jt}) \Big| >  T \epsilon \Big)
 \leq \left[\frac{\kappa_1}{\epsilon} + \frac{2(T\epsilon^2)^{1/7}}{\kappa_2} \right] \exp\left(- \frac{(T\epsilon^2)^{1/7}}{\kappa_3} \right).
    $$
\end{enumerate}
Explicit expressions for the constants $b_1$, $b_2$, $b_3$, $\kappa_1$, $\kappa_2$ and $\kappa_3$ are provided in the proof below.
\end{lemma}
\begin{proof}
The starting point is
\begin{equation}
\begin{aligned}
 \Prob\Big( \Big| &\sum_{t=p+1}^T \xi_{it} \xi_{jt} -\E(\xi_{it} \xi_{jt}) \Big| >  T \epsilon \Big)
  = \Prob\left( \big| S_{1T} + S_{2T} \big| >  T \epsilon \right)
  \leq \Prob\left( |S_{1T}| > \frac{T \epsilon }{2} \right) + \Prob\left( |S_{2T}| > \frac{T \epsilon}{2} \right) \\
  &\leq \Prob\left( |S_{1T}| > \frac{T \epsilon }{2} \right) + \frac{2}{T\epsilon} \E|S_{2T}| \leq 2m \exp\left( - \frac{T \epsilon^2}{8 m^2 C^2} \right) + \frac{2}{\epsilon}\left( I_T + II_T + III_T + IV_T \right).
\end{aligned}
\label{eq:decompositionofprob}
\end{equation}
Upper bounds on $I_T$--$IV_T$ in the RHS of \eqref{eq:decompositionofprob} are already available in Lemma \ref{lemma:covarelements}. It remains to decide on: the tail cut-off $C$, the number of terms in the finite order VMA approximation $k$, and the number of martingales differences that approximate the mixingale $m$. \textbf{(a)} From Lemma \ref{lemma:covarelements}(d), we obtain
\begin{equation}
\label{eq:polynomialboundstep1}
\begin{aligned}
 II_T+III_T + IV_T &\leq 4 (\mu_2 c_D c_1)^2 e^{-\gamma_c k}+6 c_2 (\mu_{4} c_D c_1)^2 e^{-\frac{\gamma_\alpha}{2}(m-k-p)} \\
 &\leq 4 (\mu_2 c_D c_1)^2 e^{-\gamma_c k}+6 c_2 (\mu_{4} c_D c_1)^2 e^{-\frac{\gamma_\alpha}{2}(\frac{m}{d}-k)},
\end{aligned}
\end{equation}
where we assumed that $p$ is not too large, more specifically $p\leq \frac{d-1}{d}m$ (thus $\frac{m}{d}\leq m-p$). If we match the exponents (to get the fastest decay rate), then we take $\gamma_c k = \frac{\gamma_\alpha}{2}(\frac{m}{d}-k)$, or $k = \frac{1}{d}\frac{\gamma_\alpha}{\gamma_\alpha + 2\gamma_c} m$. Defining $a_1=2 (\mu_{2d} c_D c_1)^{2d}$, $a_2 = 4 (\mu_2 c_D c_1)^2+6 c_2 (\mu_{4} c_D c_1)^2$, and $a_3 = \frac{\gamma_c}{d}\frac{\gamma_\alpha}{\gamma_\alpha + 2\gamma_c}$, we conclude from \eqref{eq:decompositionofprob} that
$$
 \Prob\Big( \Big| \sum_{t=p+1}^T \xi_{it} \xi_{jt} -\E(\xi_{it} \xi_{jt}) \Big| >  T \epsilon \Big)
  \leq
  2m \exp\left( - \frac{T \epsilon^2}{8 m^2 C^2} \right)
  + \frac{2}{\epsilon}\left(\frac{a_1}{C^{d-1}} + a_2 \exp(-a_3 m) \right).
$$
Again matching exponents, we get $m=\frac{1}{2}\left(\frac{T \epsilon^2}{a_3 C^2}\right)^{1/3}$ and hence
$$
 \Prob\Big( \Big| \sum_{t=p+1}^T \xi_{it} \xi_{jt} -\E(\xi_{it} \xi_{jt}) \Big| >  T \epsilon \Big)
  \leq \left[b_1 \left( \frac{T\epsilon^2}{C^2} \right)^{1/3}+\frac{b_3}{\epsilon}\right] \exp\left(-\frac{1}{2b_1^2}\left(\frac{T\epsilon^2}{C^2} \right)^{1/3}\right) + \frac{b_2}{\epsilon C^{d-1}},
$$
with $b_1 = a_3^{-1/3}$, $b_2=2 a_1$, and $b_3=2 a_2$. Set $C=\epsilon T^{\delta/2}$ for some $0 < \delta <1$ to complete the proof. \textbf{(b)} This time we use the upper bounds in Lemma \ref{lemma:covarelements}(e) to obtain
$$
\begin{aligned}
 I_T&+II_T+III_T+IV_T\\
 &\leq 2304 (\mu_\infty c_D c_1)^2 e^{- \sqrt{C}/2\mu_\infty c_D c_1}
 + 16 (\mu_\infty c_D c_1)^2 e^{-\gamma_c k}
 + 3456 c_2 (\mu_{\infty} c_D c_1)^2 e^{-\frac{\gamma_\alpha}{2}(m-k-p)}, \\
 &\leq \left[2304 \vee 3456 c_2 \right](\mu_\infty c_D c_1)^2
 \left[e^{- \sqrt{C}/2\mu_\infty c_D c_1}+e^{-\gamma_c k}+ e^{-\frac{\gamma_\alpha}{2}(\frac{3}{4}m-k)} \right]
 ,
\end{aligned}
$$
where we (again) assumed that $p$ is not too large, i.e. $p \leq \frac{1}{4}m$. We subsequently set $\gamma_c k = \frac{\gamma_\alpha}{2}(\frac{3}{4} m -k)$ (or $k=\frac{3}{4} \frac{\gamma_\alpha}{\gamma_\alpha+2\gamma_c} m$) and $\gamma_c k = \frac{3}{4} \frac{\gamma_\alpha \gamma_c}{\gamma_\alpha+2\gamma_c} m = \sqrt{C}/2\mu_\infty c_D c_1$ (or $\sqrt{C}=\frac{3}{2} \mu_\infty c_D c_1 \frac{\gamma_\alpha \gamma_c}{\gamma_\alpha+2\gamma_c} m $). Defining $c_\psi=\frac{3}{4} \frac{\gamma_\alpha \gamma_c}{\gamma_\alpha+2\gamma_c}$, we find $I_T+II_T+III_T+IV_T \leq 3\left[2304 \vee 3456 c_2 \right](\mu_\infty c_D c_1)^2 e^{-c_\psi m}$, and
$$
\begin{aligned}
 \Prob\Big( \Big| \sum_{t=p+1}^T \xi_{it} \xi_{jt} - &\E(\xi_{it} \xi_{jt}) \Big| >  T \epsilon \Big)
 \leq 2m \exp\left( - \frac{T \epsilon^2}{8 m^2 C^2} \right)
 + \frac{6}{\epsilon}\left[2304 \vee 3456 c_2 \right](\mu_\infty c_D c_1)^2 e^{- c_\psi m} \\
 &= 2m \exp\left( - \frac{T\epsilon^2}{8 (\mu_\infty c_D c_1 c_\psi)^4 m^6} \right)
 +\frac{6}{\epsilon}\left[2304 \vee 3456 c_2 \right](\mu_\infty c_D c_1)^2 e^{- c_\psi m}.
\end{aligned}
$$
The exponents are now equal when $m=\left( \frac{T\epsilon^2}{8 c_\psi (\mu_\infty c_D c_1 c_\psi)^4} \right)^{1/7}$. For this choice of $m$, we conclude
$$
 \Prob\Big( \Big| \sum_{t=p+1}^T \xi_{it} \xi_{jt} - \E(\xi_{it} \xi_{jt}) \Big| >  T \epsilon \Big)
 \leq \left[\frac{\kappa_1}{\epsilon} + \frac{2(T\epsilon^2)^{1/7}}{\kappa_2} \right] \exp\left(- \frac{(T\epsilon^2)^{1/7}}{\kappa_3} \right),
$$
with $\kappa_1=6\left[2304 \vee 3456 c_2 \right](\mu_\infty c_D c_1)^2$, $\kappa_2=[8 c_\psi (\mu_\infty c_D c_1 c_\psi)^4]^{1/7}$, and $\kappa_3 = \kappa_2/c_\psi$.
\end{proof}

\section{Miscellaneous Lemmas}
\begin{lemma}[Matrix $L_2$-norm Inequalities] \
 \begin{enumerate}[(a)]
     \item Let $\bA\in \sR^{N\times M_1}$ and $\bB\in \sR^{N \times M_2}$ and construct $\bC\in\sR^{N\times(M_1+M_2)}$ as $\bC = [\bA \; \bB]$. We have $\norm{\bC}_2 \leq \norm{\bA}_2+\norm{\bB}_2$.
     \item If $\bA\in \sR^{N\times M}$, and if the matrix $\bD\in\sR^{N\times P}$ ($P\leq M$) by deleting $M-P$ columns from the matrix $\bA$, then $\norm{\bD}_2 \leq \norm{\bA}_2$.
 \end{enumerate}
\label{lemma:L2normineqs}
\end{lemma}
\begin{proof}
\textbf{(a)} By definition, we have $\norm{\bC}_2^2 = \lambda_{max}(\bC\tran\bC)$. The nonzero eigenvalues of $\bC\tran\bC$ and $\bC\bC\tran$ are identical (see, e.g. Exercise 7.25(b) in \cite{Abadirmagnus2005}), and therefore
$$
\begin{aligned}
 \norm{\bC}_2^2
 &= \lambda_{max}(\bC\tran\bC)
 = \lambda_{max}(\bC\bC\tran)
 = \lambda_{max}(\bA \bA\tran + \bB\bB\tran)
 \leq \lambda_{max}(\bA\bA\tran) + \lambda_{max}(\bB\bB\tran) \\
 &= \lambda_{max}(\bA\tran\bA) + \lambda_{max}(\bB\tran\bB)
 = \norm{\bA}_2^2 + \norm{\bB}_2^2
\end{aligned}
$$
The result follows since $\norm{\bC}_2 = \sqrt{\norm{\bA}_2^2 + \norm{\bB}_2^2}\leq \norm{\bA}_2 + \norm{\bB}_2$ by the $c_r$-inequality. \textbf{(b)} The matrix $\bD\tran\bD$ is a principal submatrix of $\bA\tran\bA$. From Exercise 12.48 in \cite{Abadirmagnus2005} we conclude that $\norm{\bD}_2^2=\lambda_{max}(\bD\tran\bD)\leq \lambda_{max}(\bA\tran\bA)=\norm{\bA}_2^2$.
\end{proof}

\section{Assumptions and Proof for the Setting with Exogeneous Regressors}\label{appendix:exogeneousvariables}

\subsection{Assumptions governing the DGP with exogeneous regressors}
The proof of Theorem \ref{Thm:Sglwithexo} requires assumptions that are similar to those encountered in the main text. For the reader's convenience, we make the correspondence explicit by adhering to the original numbering while adding a ``*''. A short discussion of these assumption is found at the end of this section. 
\begin{customassumption}{1*}\label{assump:stability_exogeneous}\
\begin{enumerate}[(a)]
    \item $\normoneinf{\bA}=\max\left\{ \norm{\bA}_1,\norm{\bA}_\infty \right\} \leq \delta_A <1$. 
    \item $\normoneinf{\bB}\leq C_B$ and $\frac{C_B}{1-\delta_A}<1$.
    \item $ \norm{\bbeta^*}_1 = \max\{\norm{\bbeta_1}_1,\ldots,\norm{\bbeta_K}_1\} \leq C_\beta$.
\end{enumerate}
\end{customassumption}

\begin{customassumption}{2*}\label{assump:momentcond_exogeneous}\
\begin{enumerate}[(a)]
    \item The regressor $x_{i_1t_1,k}$ and innovation $\epsilon_{i_2,t_2}$ are independent of each other for all $1\leq i_1 < i_2\leq N$, all $1\leq t_1,t_2 \leq T$, and all $k=1,\ldots,K$.
    \item The sequence $\{\bepsilon_t\}$ is a covariance stationary, martingale difference process with respect to the filtration $\calF_{t-1}=\sigma\left(\bepsilon_{t-1},\bepsilon_{t-2},\ldots\right)$, and geometrically strong mixing ($\alpha$-mixing). That is, its mixing coefficients $\{\alpha_m\}$ satisfy $\alpha_m\leq c_2 e^{-\gamma_\alpha m }$ for all $m$ and some constants $c_2,\gamma_\alpha>0$. The largest and smallest eigenvalues of $\bSigma_\epsilon= \E(\bepsilon_1 \bepsilon_1^\prime)=(\sigma_{ij})_{i,j=1}^N$ are bounded away from $0$ and $\infty$. 
    \item For each $k=1,\ldots,K$, the sequence $\{\bx_{t,k}\}$ is covariance stationary and geometrically strong mixing ($\alpha$-mixing). That is, its mixing coefficients $\{\alpha_m^*\}$ satisfy $\alpha_m^*\leq c_2^* e^{-\gamma_\alpha^* m }$ for all $m$ and some constants $c_2^*,\gamma_\alpha^*>0$.
    \item Either one of the following assumptions holds:
    \begin{enumerate}
     \item[(d1)] For $\psi(x) = x^d$ and $d\geq 4$, we require $\sup_{i,t} \norm{\epsilon_{it}}_\psi = \left(\E |\epsilon_{it}|^d \right)^{1/d} \leq \mu_d < \infty$ and $\sup_{i,t} \norm{x_{it,k}}_\psi = \left(\E |x_{it,k}|^d \right)^{1/d} \leq \mu_d^* < \infty$ ($k=1,\ldots,K$).
     \item[(d2)] For $\psi(x) = \exp(x)-1$, we have $\sup_{i,t} \norm{\epsilon_{it}}_\psi \leq \mu_\infty < \infty$ and $\sup_{i,t} \norm{x_{it,k}}_\psi \leq \mu_\infty < \infty$ ($k=1,\ldots,K$).
    \end{enumerate}
 \end{enumerate}
\end{customassumption}

\begin{customassumption}{3*}\label{assump:bandedness_exogeneous}\
Recall $\bA = (a_{ij})_{i,j=1}^N$, $\bB = (b_{ij})_{i,j=1}^N$, and $\bSigma_\epsilon = (\sigma_{ij})_{i,j=1}^N$. We have: (a) $a_{ij} = b_{ij} = 0$ for all $\abs{i-j} > k$ with $k < \lfloor N/4 \rfloor$, and (b) $\sigma_{ij} = 0$ for all $\abs{i-j} > l_0$.
\end{customassumption}

\begin{customassumption}{4*}[Restricted minimum eigenvalue]\label{assump:min_eigen_bound_exogeneous}
Define the $((K+1)N\times(K+2)N)$ matrix
$$
 \bQ =
 \begin{bmatrix}
  \bSigma_1\tran    & \bSigma_0             & (\bSigma_1^{x_1 y})\tran  & \hdots    &(\bSigma_1^{x_K y})\tran\\  
  \bSigma_0^{x_1 y} & \bSigma_1^{x_1 y}     & \bSigma_0^{x_1 x_1}       & \hdots    &\bSigma_0^{x_1 x_K}\\ 
  \vdots            & \vdots                & \vdots                    & \ddots   & \vdots\\
  \bSigma_0^{x_K y} & \bSigma_1^{x_K y}     & \bSigma_0^{x_K x_1}       & \hdots    & \bSigma_0^{x_K x_K}
 \end{bmatrix}.
$$
We assume that
$$
\phi_\min^*(\bx) := \min_{\bx \in \mathbb{R}^{(K+2)N}: \mathcal{M}(\bx) \leq N+K} \frac{\norm{\bQ\bx}_2}{\norm{\bx}_2} \geq \phi_0^*.
$$
\end{customassumption}

Assumptions \ref{assump:stability_exogeneous}(a)--(b) are also encountered in the main paper. Including exogenous regressors, the reduced-form representation of the model becomes $\by_t =\bC \by_{t-1} +\bD\left( \bepsilon_t + \sum_{k=1}^K \diag(\bbeta_k) \bx_{t,k} \right)$. Assumption \ref{assump:stability_exogeneous}(c) merely assumes an upper bound on the magnitude of the exogenous regressor coefficients. Assumption \ref{assump:momentcond_exogeneous} controls dependencies over time, in the cross-section, and with the exogenous regressors. Similarly to assumption A8(i) in \cite{MaGuoWang2021}, we enforce exogeneity through Assumption \ref{assump:momentcond_exogeneous}(a). Assumption \ref{assump:momentcond_exogeneous}(b) is repeated from the main paper and its counterpart for the $\{\bx_{t,k}\}$'s is encountered as Assumption \ref{assump:momentcond_exogeneous}(c). All original moment conditions are also transferred to the exogenous regressors (Assumption \ref{assump:bandedness_exogeneous}). Assumptions \ref{assump:stability_exogeneous}--\ref{assump:min_eigen_bound_exogeneous} allow for an easy analogy with the assumptions in the main paper (at the cost of possibly being more restrictive than strictly necessary). That is, define $\bepsilon_t^*=\bepsilon_t + \sum_{k=1}^K \diag(\bbeta_k) \bx_{t,k}$ and note the linearity of $\bepsilon_t^*$ in $\bepsilon_t$ and $\{\bx_{t,1},\ldots,\bx_{t,K}\}$. Due to fixed $K$, all mixing properties and moments conditions simply carry over to $\bepsilon_t^*$ causing Lemma \ref{lemma:covmatrixelements} to remain valid.

\subsection{Proof of Theorem \ref{Thm:Sglwithexo}}
Set $\bbeta^*=(\bbeta_1\tran,\ldots,\bbeta_K\tran)\tran$, $\bU^{(d)}=\begin{bmatrix}\bW_1^{*(d)} &\cdots & \bW_K^{*(d)}\end{bmatrix}$, and $\bQ^{(d)}=\begin{bmatrix} \bV^{*(d)} & \bU^{(d)} \end{bmatrix}$ (and similar quantities with ``hats''). The residuals entering the $L_2$ component of the SPLASHX objective function are
\begin{equation}
\begin{aligned}
 \hat{\bsigma}^* - &\hat{\bV}^{*(d)}\bc - \sum_{k=1}^K \hat\bW_k^{*(d)} \bbeta_k
 = \hat{\bsigma}^* - \begin{bmatrix} \hat{\bV}^{*(d)} & \hat\bU^{(d)} \end{bmatrix} \begin{bmatrix} \bc \\ \bbeta^* \end{bmatrix} \\
 &= \underbrace{\bsigma^* - \begin{bmatrix} \bV^{*(d)} & \bU^{(d)} \end{bmatrix}\bq}_{=\bZeros}
 +\underbrace{\big[\hat{\bsigma}^*-\bsigma^*\big]}_{:=\hat\bDelta_1}
 - \underbrace{\big[\hat{\bV}^{*(d)}-\bV^{*(d)}\big]}_{:=\hat\bDelta_2} \bc
 - \underbrace{\big[\hat\bU^{(d)}-\bU^{(d)}\big]}_{:=\hat\bDelta_3}\bbeta^* \\
 &= \hat\bDelta_1 - \hat\bDelta_2 \bc - \hat\bDelta_3 \bbeta^*.
\end{aligned}
\label{eq:exostar}
\end{equation}
Exploiting the previously defined notation and \eqref{eq:exostar}, we have
\begin{equation}
\begin{aligned}
\Big\|\hat{\bsigma}^* - &\hat{\bV}^{*(d)}\hat\bc - \sum_{k=1}^K \hat\bW_k^{*(d)} \hat\bbeta_k\Big\|_2^2
=
 \norm{\hat{\bsigma}^* - \hat\bQ^{(d)}\hat\bq}_2^2 
=
\norm{ \left\{\hat{\bsigma}^* - \hat\bQ^{(d)} \bq\right\}-\left\{\hat\bQ^{(d)} \big(\hat\bq - \bq \big) \right\}}_2^2 \\
&= \norm{\hat{\bsigma}^* - \hat\bQ^{(d)} \bq}_2^2
+ \norm{\hat\bQ^{(d)} \big(\hat\bq - \bq \big)}_2^2
- 2 (\hat\bq-\bq)\tran \hat\bQ^{(d)\prime}\left(\hat\bDelta_1 - \hat\bDelta_2 \bc - \hat\bDelta_3 \bbeta^*\right).
\end{aligned}
\end{equation}

We subsequently adjust the penalty function $P_\alpha(\bc)$ to incorporate the penalty on the coefficients in front of the exogenous variables. Define the index set $g_k$ such that $\bq_{g_k}=\bbeta_k$ ($k=1\ldots,K$) and enlarge $\mathcal{G}$ to $\mathcal{G}^*=\mathcal{G}\cup\bigcup_{k=1}^K g_k$. As $\abs{\bbeta_k}=\abs{g_k}=N$, we have
$$
 P_\alpha(\bc) + \sum_{k=1}^K (1-\alpha)\sqrt{N}\norm{\bbeta_k}_2 + \alpha \norm{\bbeta_k}_1
 = (1-\alpha) \sum_{g\in\mathcal{G}^*} \sqrt{\abs{g}} \norm{\bq_g}_2 + \alpha \norm{\bq}_1 =: \widetilde{P}_\alpha(\bq).
$$
Using this newly defined norm $\widetilde{P}_\alpha(\bq)$, a more concise notation for the SPLASHX objective function is $\mathcal{L}_\alpha^*(\bq;\lambda)=\big\|\hat{\bsigma}^* - \hat\bQ^{(d)} \bq\big\|_2^2+\lambda \widetilde{P}_\alpha(\bq)$ and its dual norm $\widetilde{P}_{\alpha}^*(\bq)$ satisfies $\widetilde{P}_\alpha^*(\bq)\leq \norm{\bq}_\infty$ (see \eqref{eq:maxnorm} in the main text). From $\mathcal{L}_\alpha^*(\hat\bq;\lambda)\leq\mathcal{L}_\alpha^*(\bq;\lambda)$, it follows that
\begin{equation}
\begin{aligned}
    &\norm{\hat\bQ^{(d)} \big(\hat\bq - \bq \big)}_2^2+\lambda \widetilde{P}_\alpha(\hat\bq)
    \leq2 (\hat\bq-\bq)\tran \hat\bQ^{(d)\prime}\left(\hat\bDelta_1 - \hat\bDelta_2 \bc - \hat\bDelta_3 \bbeta^*\right)+\lambda \widetilde{P}_\alpha(\bq) \\
    &\quad\leq \widetilde{P}_\alpha(\hat\bq-\bq) \left[2\widetilde{P}_\alpha^*\left(\hat\bQ^{(d)\prime}\hat\bDelta_1\right) +2\widetilde{P}_\alpha^*\left(\hat\bQ^{(d)\prime}\hat\bDelta_2 \bc\right) +2\widetilde{P}_\alpha^*\left(\hat\bQ^{(d)\prime}\hat\bDelta_3 \bbeta^* \right)\right] + \lambda \widetilde{P}_\alpha(\bq) \\
    &\quad\leq \widetilde{P}_\alpha(\hat\bq-\bq)
    \left[2\norm{\hat\bQ^{(d)\prime}\hat\bDelta_1}_\infty
    +2\norm{\hat\bQ^{(d)\prime}\hat\bDelta_2 \bc}_\infty
    +2\norm{\hat\bQ^{(d)\prime}\hat\bDelta_3 \bbeta^*}_\infty\right] + \lambda \widetilde{P}_\alpha(\bq).
\end{aligned}
\label{eq:startingpointLassoInequality}
\end{equation}
We subsequently define the following three sets:
$$
\mathcal{H}_1^*(x) = \left\{2\norm{\hat\bQ^{(d)\prime}\hat\bDelta_1}_\infty \leq x \right\},\qquad\qquad
\mathcal{H}_2^*(x) = \left\{2\norm{\hat\bQ^{(d)\prime}\hat\bDelta_2 \bc}_\infty\leq x\right\}
$$
and
$$
\mathcal{H}_3^*(x) = \left\{2\norm{\hat\bQ^{(d)\prime}\hat\bDelta_3 \bbeta^*}_\infty\leq x \right\}.
$$
Rescaling \eqref{eq:startingpointLassoInequality} by a factor 2 and on $\mathcal{H}_1^*(\tfrac{\lambda}{6})\cap\mathcal{H}_2^*(\tfrac{\lambda}{6})\cap\mathcal{H}_3^*(\tfrac{\lambda}{6})$, we get
\begin{equation}
 2\norm{\hat\bQ^{(d)} \big(\hat\bq - \bq \big)}_2^2
 + 2 \lambda \widetilde{P}_\alpha(\hat\bq)
 \leq \lambda\widetilde{P}_\alpha(\hat\bq-\bq) +  2 \lambda \widetilde{P}_\alpha(\bq).
\label{eq:maininequalityexo}
\end{equation}
We first manipulate the term $\widetilde{P}_\alpha(\hat\bq)$ in the LHS of \eqref{eq:maininequalityexo}. As in \eqref{eq:P_rewrite_1} of the main paper, we get
$$
\begin{aligned}
 \widetilde{P}_\alpha(\hat\bq) \geq
\widetilde{P}_{\alpha,S^*}(\bq) + \widetilde{P}_{\alpha,S^{*c}}(\hat\bq-\bq)
 - \widetilde{P}_{\alpha,S^*}(\hat\bq-\bq),
\end{aligned}
$$
where $S^*$ is the index sets of all nonzero coefficients in $\bq$. For the term $\widetilde{P}_\alpha(\hat\bq-\bq)$ in the RHS of \eqref{eq:maininequalityexo}, if follows from \eqref{eq:P_rewrite_2} that $\widetilde{P}_\alpha\left(\hat{\bq}-\bq\right)=\widetilde{P}_{\alpha,S^*}\left(\hat{\bq}-\bq\right) +\widetilde{P}_{\alpha,S^{*c}}\left(\hat{\bq}-\bq\right)$. We obtain
\begin{equation}
 2\norm{\hat\bQ^{(d)} \big(\hat\bq - \bq \big)}_2^2
  + \lambda \widetilde{P}_{\alpha,S^{*c}}(\hat\bq-\bq)
  \leq 3 \lambda \widetilde{P}_{\alpha,S^*}(\hat\bq-\bq).
\label{eq:RECconsequenceEXO}
\end{equation}
Accordingly, with $N^*=N_c+NK$,
$$
(\hat\bq-\bq) \in \mathcal{C}_{N^*}^*(\mathcal{G},S) := \left\{ \bDelta \in \mathbb{R}^{N^*} : \widetilde{P}_{\alpha,S^{*c}}(\bDelta) \leq 3\widetilde{P}_{\alpha,S^*}(\bDelta) \right\}.
$$
Lemma \ref{Lemma:min_vs_restricted_eigen} now goes through with $\bar\omega_\alpha^*$ as in Theorem \ref{Thm:Sglwithexo}. For applicability of Lemma \ref{Lemma:V_transfer}, we define $\mathcal{V}^*(x)=\big\{ \big\|\hat\bQ - \bQ\big\|_2 \leq x\big\}$ and assume $\mathcal{V}^*(\tfrac{\phi_0^*}{2})$. Derivations identical to those in \eqref{eq:finalineqSGL} prove the main inequality of Theorem \ref{Thm:Sglwithexo}. 

To derive the probability of the inequality being true, we need the probability of the occurrence of the event $\mathcal{H}_1^*(\tfrac{\lambda}{6})\cap\mathcal{H}_2^*(\tfrac{\lambda}{6})\cap\mathcal{H}_3^*(\tfrac{\lambda}{6})\cap \mathcal{V}^*(\tfrac{\phi_0^*}{2})$. This probability is no smaller than
\begin{equation}
 1- \Prob\left(\mathcal{H}_1^*(\tfrac{\lambda}{6})^c\right) - \Prob\left(\mathcal{H}_2^*(\tfrac{\lambda}{6})^c\right) - \Prob\left(\mathcal{H}_3^*(\tfrac{\lambda}{6})^c\right) - \Prob\left(\mathcal{V}^*(\tfrac{\phi_0^*}{2})^c\right) ,
\label{eq:ExogeneousProbLowerBound}
\end{equation}
and all probabilities in \eqref{eq:ExogeneousProbLowerBound} can be retraced to probabilities involving $\|\hat\bQ-\bQ\|_{\vdash}$. That is, bounding terms as in \eqref{eq:H1}--\eqref{eq:ProbabilitySetBoundH2}, we find
\begin{equation}
\begin{aligned}
\Prob\left(\mathcal{H}_1^*(\tfrac{\lambda}{6})^c\right)
 &= \Prob\left( \norm{\hat\bQ^{(d)\prime}\hat\bDelta_1}_\infty  > \frac{\lambda}{12} \right)
 \leq \Prob\left( \normoneinf{\hat\bQ -\bQ}^2 + C_Q \normoneinf{\hat\bQ -\bQ} > \frac{\lambda}{12} \right) \\
 &\leq \Prob\left( \normoneinf{\hat\bQ -\bQ} > \frac{\lambda^{1/2}}{\sqrt{24}} \right)
 + \Prob\left( \normoneinf{\hat\bQ -\bQ} > \frac{\lambda}{24 C_Q} \right),
\end{aligned}
\label{eq:ProbH1Exo}
\end{equation}
\begin{equation}
\begin{aligned}
\Prob\left(\mathcal{H}_2^*(\tfrac{\lambda}{6})^c\right)
 &=\Prob\left( \norm{\hat\bQ^{(d)\prime}\hat\bDelta_2 \bc}_\infty > \frac{\lambda}{12}  \right) \leq \Prob\left( \norm{\hat\bQ^{(d)\prime} \left[\hat\bQ^{(d)} -\bQ^{(d)} \right] }_\infty > \frac{\lambda}{12}  \right) \\
   &\leq \Prob\left( \normoneinf{\hat\bQ -\bQ} > \frac{\lambda^{1/2}}{\sqrt{24}} \right)
 + \Prob\left( \normoneinf{\hat\bQ -\bQ} > \frac{\lambda}{24 C_Q} \right),
\end{aligned}
\label{eq:ProbH2Exo}
\end{equation}
and
\begin{equation}
\begin{aligned}
\Prob\left(\mathcal{H}_3^*(\tfrac{\lambda}{6})^c\right)
 &=\Prob\left( \norm{\hat\bQ^{(d)\prime}\hat\bDelta_3 \bbeta^*}_\infty > \frac{\lambda}{12}  \right) \\
 &\leq \Prob\left( \normoneinf{\hat\bQ -\bQ} > \frac{\lambda^{1/2}}{\sqrt{24} C_\beta} \right)
 + \Prob\left( \normoneinf{\hat\bQ -\bQ} > \frac{\lambda}{24 C_Q C_\beta} \right)
\end{aligned}
\label{eq:ProbH3Exo}
\end{equation}
For \eqref{eq:ProbH2Exo} and \eqref{eq:ProbH3Exo}, Assumptions \ref{assump:stability_exogeneous}(a)--(b) and Assumption \ref{assump:stability_exogeneous}(c) are needed to guarantee $\norm{\bc}_1\leq 1$ and $\norm{\bbeta^*}_1\leq C_\beta$, respectively. Also, $\Prob\left(\mathcal{V}^*(\tfrac{\phi_0^*}{4})^c \right)=\Prob\left(\norm{\hat\bQ - \bQ}_2 > \tfrac{\phi_0^*}{2} \right)\leq \Prob\left(\normoneinf{\hat\bQ - \bQ} > \tfrac{\phi_0^*}{2} \right)$. Given these bounds, \eqref{eq:ExogeneousProbLowerBound} translates to
$$
 \Prob\Big(\mathcal{H}_1^*(\tfrac{\lambda}{6})\cap\mathcal{H}_2^*(\tfrac{\lambda}{6})\cap\mathcal{H}_3^*(\tfrac{\lambda}{6})\cap \mathcal{V}^*(\tfrac{\phi_0^*}{2}) \Big)
 \geq 1 - 7 \Prob\left( \normoneinf{\hat\bQ -\bQ} > 6 f^*(\lambda, \phi_0^*)\right),
$$
where $f^*(\lambda,\phi_0^*) = \min\left\{\tfrac{\lambda^{1/2}}{12\sqrt{6}} , \tfrac{\lambda}{144C_Q}, \tfrac{\lambda^{1/2}}{12\sqrt{6} C_\beta} , \tfrac{\lambda}{144C_Q C_\beta}, \tfrac{\phi_0^*}{12} \right\}$.

All that remains is a lower bound for $\Prob\left( \normoneinf{\hat\bQ -\bQ} > x\right)$. We instead derive an upper bound for $\Prob\left( \normoneinf{\hat\bQ -\bQ} \leq x\right)$ as follows
$$
\begin{aligned}
 \Prob\Bigg( &\normoneinf{\hat\bQ -\bQ} \leq x\Bigg)
  =  \Prob\left( \max\left\{ \norm{\hat\bQ -\bQ}_1, \norm{\hat\bQ -\bQ}_\infty \right\} \leq x\right) \geq \Prob\left( (K+2)N \norm{\hat\bQ -\bQ}_{max} \leq x\right) \\
  &= 1 -  \Prob\left(\norm{\hat\bQ -\bQ}_{max} > \frac{x}{(K+2)N}\right) \\
  &\geq 1 -  \Prob\left( \bigcup_{1\leq i \leq (K+1)N, 1\leq j \leq (K+2)N } \left| \left[\hat\bQ -\bQ \right]_{ij} \right| > \frac{x}{(K+2)N}\right) \\
  &\geq 1 - \sum_{i=1}^{(K+1)N} \sum_{j=1}^{(K+2)N} \Prob\left(\left| \sum_{t=2}^T \xi_{it} \xi_{jt} - \E(\xi_{it} \xi_{jt}) \right| > \frac{T x}{(K+2)N} \right) \\
  &\geq
  \begin{cases}
    1 - (K+1)(K+2)N^2\left[ \left(b_1 T^{(1-\delta)/3}+\frac{(K+2)Nb_3}{x}\right) \exp\left(-\frac{T^{(1-\delta)/3}}{2b_1^2}\right) + \frac{b_2 (K+2)^d N^d}{x^d T^{\frac{\delta}{2}(d-1)}}\right] \\ \qquad\qquad\qquad\qquad\qquad\qquad\qquad\qquad\qquad\qquad\qquad\qquad (\text{polynomial tails}), \\
    1 - (K+1)(K+2)N^2 \left[\frac{\kappa_1(K+2)N}{x} + \frac{2}{\kappa_2}\left(\frac{T x^2}{(K+2)N}\right)^{1/7} \right] \exp\left(- \frac{1}{\kappa_3}\left(\frac{T x^2}{(K+2)^2 N^2}\right)^{1/7} \right) \\
    \qquad\qquad\qquad\qquad\qquad\qquad\qquad\qquad\qquad\qquad\qquad\qquad (\text{exponential tails}),
  \end{cases}
\end{aligned},
$$
where $\xi_{it}$ denotes a generic element of an autocovariance matrix (as in Lemma \ref{lemma:covmatrixelements}). In accordance with Theorem \ref{Thm:Sglwithexo}, these RHS probabilities are equivalent to $1-\mathcal{P}_1^*(x,N,T)$ (polynomial tails) and $1-\mathcal{P}_2^*(x,N,T)$ (exponential tails).

\clearpage
\section{Miscellaneous Figures}

\subsection{Additional Simulation Results}

\begin{figure}[h]
    \centering
    \includegraphics[width=\textwidth]{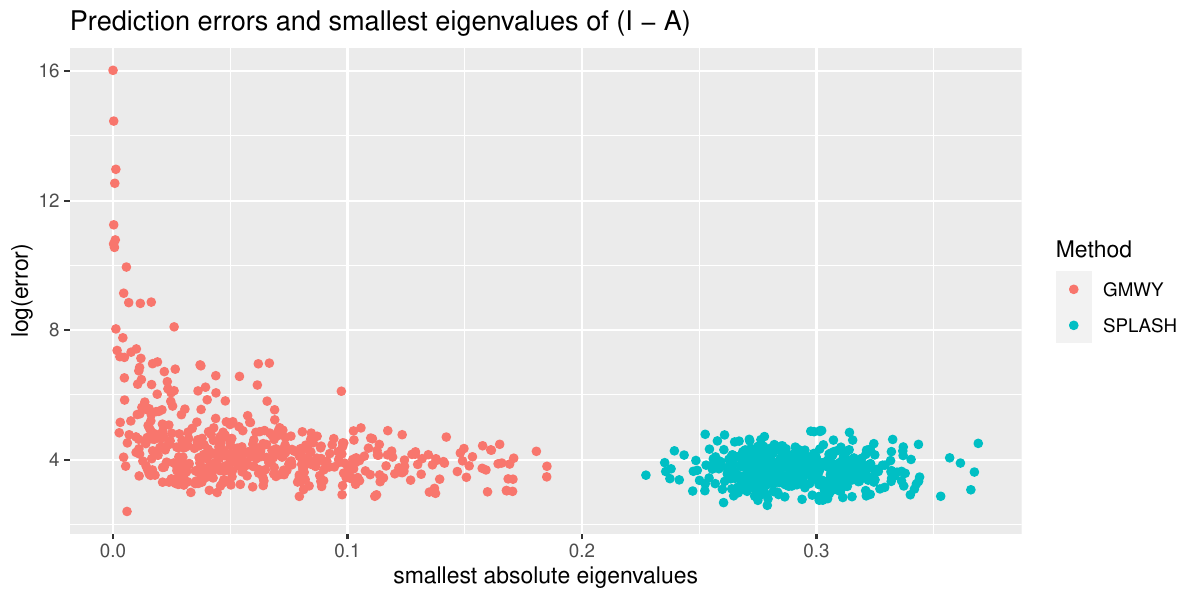}
    \caption{The scatterplot of the (logarithmic) squared prediction errors and the minimum eigenvalues of $(\bI-\bA)$ for GMWY($k_0$) and SPLASH($0$,$\lambda$). The figure relates to Design B without exogenous regressors with $N=25$ and $T=500$.}
    \label{fig:error_vs_eigen}
\end{figure}

\clearpage
\subsection{Satellite Data}\label{App:Satellite_Data}

\begin{figure}[h]
\center
\begin{tikzpicture}
    [
        box/.style={rectangle,draw=black,thick, minimum size=1.25cm},
    ]
\foreach \x in {1,...,9}{
    \foreach \y in {1,...,5}
        \node[box] at (1.25*\x,1.25*\y){};
}
\node[box,fill=green] at (1.25*1,1.25*1){};
\node[box,fill=green] at (1.25*1,1.25*2){};
\node[box,fill=green] at (1.25*1,1.25*3){};
\node[box,fill=green] at (1.25*1,1.25*4){};
\node[box,fill=black!40!green] at (1.25*1,1.25*5){};
\node[box,fill=green] at (1.25*2,1.25*3){};
\node[box,fill=green] at (1.25*2,1.25*4){};
\node[box,fill=green] at (1.25*2,1.25*5){};
\node[box,fill=green] at (1.25*3,1.25*5){};
\node[box,fill=white!75!blue] at (1.25*3,1.25*3){};
\node[box,fill=white!75!blue] at (1.25*4,1.25*1){};
\node[box,fill=white!75!blue] at (1.25*4,1.25*2){};
\node[box,fill=white!75!blue] at (1.25*4,1.25*3){};
\node[box,fill=white!75!blue] at (1.25*4,1.25*4){};
\node[box,fill=white!75!blue] at (1.25*4,1.25*5){};
\node[box,fill=white!75!blue] at (1.25*5,1.25*1){};
\node[box,fill=white!75!blue] at (1.25*5,1.25*2){};
\node[box,fill=white!40!blue] at (1.25*5,1.25*3){};
\node[box,fill=white!75!blue] at (1.25*5,1.25*4){};
\node[box,fill=white!75!blue] at (1.25*5,1.25*5){};
\node[box,fill=white!75!blue] at (1.25*6,1.25*1){};
\node[box,fill=white!75!blue] at (1.25*6,1.25*2){};
\node[box,fill=white!75!blue] at (1.25*6,1.25*3){};
\node[box,fill=white!75!blue] at (1.25*6,1.25*4){};
\node[box,fill=white!75!blue] at (1.25*6,1.25*5){};
\node[box,fill=white!75!blue] at (1.25*7,1.25*3){};

\node (A1) at (1*1.25,5*1.25) {$y_1$};
\node at (1*1.25,4*1.25) {$y_2$};
\node at (1*1.25,3*1.25) {$y_3$};
\node at (1*1.25,2*1.25) {$y_4$};
\node (A2) at (1*1.25,1*1.25) {$y_5$};
\node at (2*1.25,5*1.25) {$y_6$};
\node (A4) at (2*1.25,4*1.25) {$y_7$};
\node at (2*1.25,3*1.25) {$y_8$};
\node at (2*1.25,2*1.25) {$y_9$};
\node at (2*1.25,1*1.25) {$y_{10}$};
\node (A3) at (3*1.25,5*1.25) {$y_{11}$};
\node at (3*1.25,4*1.25) {$y_{12}$};
\node (B2) at (3*1.25,3*1.25) {$y_{13}$};
\node at (3*1.25,2*1.25) {$y_{14}$};
\node at (3*1.25,1*1.25) {$y_{15}$};
\node at (4*1.25,5*1.25) {$y_{16}$};
\node (B6) at (4*1.25,4*1.25) {$y_{17}$};
\node at (4*1.25,3*1.25) {$y_{18}$};
\node (B9) at (4*1.25,2*1.25) {$y_{19}$};
\node at (4*1.25,1*1.25) {$y_{20}$};
\node (B4) at (5*1.25,5*1.25) {$y_{21}$};
\node at (5*1.25,4*1.25) {$y_{22}$};
\node (B1) at (5*1.25,3*1.25) {$y_{23}$};
\node at (5*1.25,2*1.25) {$y_{24}$};
\node (B5) at (5*1.25,1*1.25) {$y_{25}$};
\node at (6*1.25,5*1.25) {$y_{26}$};
\node (B7) at (6*1.25,4*1.25) {$y_{27}$};
\node at (6*1.25,3*1.25) {$y_{28}$};
\node (B8) at (6*1.25,2*1.25) {$y_{29}$};
\node at (6*1.25,1*1.25) {$y_{30}$};
\node at (7*1.25,5*1.25) {$y_{31}$};
\node at (7*1.25,4*1.25) {$y_{32}$};
\node (B3) at (7*1.25,3*1.25) {$y_{33}$};
\node at (7*1.25,2*1.25) {$y_{34}$};
\node at (7*1.25,1*1.25) {$y_{35}$};
\node at (8*1.25,5*1.25) {$y_{36}$};
\node at (8*1.25,4*1.25) {$y_{37}$};
\node at (8*1.25,3*1.25) {$y_{38}$};
\node at (8*1.25,2*1.25) {$y_{39}$};
\node at (8*1.25,1*1.25) {$y_{40}$};
\node at (9*1.25,5*1.25) {$y_{41}$};
\node at (9*1.25,4*1.25) {$y_{42}$};
\node at (9*1.25,3*1.25) {$y_{43}$};
\node at (9*1.25,2*1.25) {$y_{44}$};
\node at (9*1.25,1*1.25) {$y_{45}$};

\end{tikzpicture}
\caption{The vertically ordered spatial grid. Colors indicate the admissible interactions for spatial units $y_1$ and $y_{23}$ under the identification constraint that the spatial units $y_i$ and $y_j$ can interact only if $\abs{i-j} < \lfloor N/4 \rfloor = 11$. Unit $y_1$ (dark green) possibly interacts with all spatial units with a lighter green color. The allowed interactions for $y_{23}$ (dark purple) are displayed similarly.}
\label{fig:spatial_grid_London}
\end{figure}

\newpage
\bibliographystyle{apalike}
\bibliography{Bibliography}

\makeatletter\@input{xx2.tex}\makeatother